%% file: main.tex
\documentclass{LMCS}
\pdfoutput=1
\def\dOi{9(3:16)2013}
\lmcsheading%
{\dOi}
{1--65}
{}
{}
{Apr.~\phantom05, 2012}
{Sep.~17, 2013}
{}
\input{macros}

\begin{document}

\title[Connector algebras for C/E and P/T nets' interactions]{Connector algebras for C/E and P/T nets' interactions} 

\author[R.~Bruni]{Roberto Bruni\rsuper a}
 \address{{\lsuper{a,c}}Dipartimento di Informatica, Universit{\`a} di Pisa - Largo Bruno Pontecorvo 3, I-56127 Pisa, Italy}
 \email{\{bruni,ugo\}@di.unipi.it}

\author[H.~Melgratti]{Hern\'an Melgratti\rsuper b}
 \address{{\lsuper b}Departamento de Computaci\'on, FCEyN, Universidad de Buenos Aires - CONICET. Pabell\'on I, Ciudad Universitaria, (C1428EGA) Buenos Aires, Argentina}
 \email{hmelgra@dc.uba.ar}

\author[U.~Montanari]{Ugo Montanari\rsuper c}
\address{\vskip-6 pt}
%\email{ugo@di.unipi.it}

\author[P.~Soboci{\'n}ski]{Pawe{\l} Soboci{\'n}ski\rsuper d}
\address{{\lsuper d}ECS, University of Southampton, SO17 1BJ United Kingdom}
\email{ps@ecs.soton.ac.uk}

%% mandatory lists of keywords and classifications:
%\subjclass{MANDATORY list of acm classifications}
\ACMCCS{[{\bf Theory of computation}]: Models of computation---Concurrency; Formal languages and automata theory---Formalisms}

%\titlecomment{OPTIONAL comment concerning the title, e.g., if a variant
%or an extended abstract of the paper has appeared elsewehere}
%%%%%%%%%%%%%%%%%%%%%%%%%%%%%%%%%%%%%%%%%%%%%%%%%%%%%%%%%%%%%%%%%%%%%%%%%%%

\begin{abstract}
A quite flourishing research thread in the recent literature on component-based systems is concerned with the algebraic properties of different classes of connectors.
In a recent paper,  an algebra of stateless connectors was presented that consists of five kinds of basic connectors, namely symmetry, synchronization, mutual exclusion, hiding and inaction, plus their duals, and it was shown how they can be freely composed in series and in parallel to model sophisticated ``glues''.
In this paper we explore the expressiveness of stateful connectors obtained by adding one-place buffers or unbounded buffers to the stateless connectors. 
The main results are:
i) we show how different classes of connectors exactly correspond to suitable classes of Petri nets equipped with compositional interfaces, called \emph{nets with boundaries};
ii) we show that the difference between strong and weak semantics in stateful connectors is reflected in the semantics of nets with boundaries by moving from the classic step semantics (strong case) to a novel \emph{banking semantics} (weak case), where a step can be executed by taking some ``debit'' tokens to be given back during the same step;
iii) we show that the corresponding bisimilarities are congruences (w.r.t. composition of connectors in series and in parallel);
iv) we show that suitable monoidality laws, like those arising when representing stateful connectors in the tile model, can nicely capture concurrency (in the sense of step semantics) aspects; and
v) as a side result, we provide a basic algebra, with a finite set of symbols, out of which we can compose all P/T nets with boundaries, fulfilling a long standing quest. 

\end{abstract}

\maketitle

\section{Introduction}
\label{sec:intro}

%COMPONENT-BASED APPROACHES
A successful and widely adopted approach to modern software architectures is the so-called \emph{component-based} approach~\cite{Wolf92foundationsfor}.
At its core, it is centred around three main kinds of elements: processing elements (also called \emph{components}), \emph{data} elements and connecting elements (also called \emph{connectors}).
The main idea is to assemble heterogeneous and separately developed components that exchange data items via their programming interfaces by synthesising the appropriate ``glue'' code, i.e., by linking components via connectors.
In this sense, connectors must take care of all those aspects that lie outside of the scopes of individual components and for which the operating infrastructure is held responsible. 
Typically, components and connectors are made available and assembled off-the-shelf. To favour their re-usability, their semantic properties, including requirements and offered guarantees must be unambiguously specified.
Thus, connectors are first class entities and assessing rigorous mathematical theories for them is of crucial relevance for the analysis of component-based systems.

Connectors can live at different levels of abstraction (architecture, software, processes) and several kinds of connectors have been studied in the literature~\cite{DBLP:journals/mscs/Arbab04,DBLP:journals/scp/FiadeiroM97,DBLP:journals/tcs/BruniLM06,DBLP:journals/tc/BliudzeS08,BarbosaB04}. 
Here we focus on the approach initiated in~\cite{DBLP:journals/tcs/BruniGM02} and continued in~\cite{DBLP:journals/tcs/BruniLM06}, where a basic algebra of \emph{stateless connectors} was presented. It consists of five kinds of basic connectors (plus their duals), namely symmetry, synchronisation, mutual exclusion, hiding and inaction.
 The connectors can be composed in series or in parallel and the resulting circuits are equipped with a normal form axiomatization. These circuits are quite expressive: they can model the coordination aspects of the architectural design language CommUnity~\cite{DBLP:journals/scp/FiadeiroM97} and, using in addition a simple 1-state buffer, the classic set of ``channels'' provided by the coordination language Reo~\cite{DBLP:journals/mscs/Arbab04} (see~\cite{DBLP:conf/wadt/ArbabBCLM08}).

%Contribution in a nutshell
In~\cite{DBLP:conf/concur/Sobocinski10,DBLP:conf/concur/BruniMM11} the aforementioned
stateless connectors were presented in process algebra form and given a subtly different operational semantics, emphasising the role of the algebra of labels, in particular with a label $0$ meaning inaction~\cite{DBLP:conf/concur/Sobocinski10} and, in~\cite{DBLP:conf/concur/BruniMM11}
with a monoidal structure (of which $0$ is the identity).
Moreover, they were extended with certain simple buffer components: one-place buffers in~\cite{DBLP:conf/concur/Sobocinski10} and unbounded buffers in~\cite{DBLP:conf/concur/BruniMM11}. In both cases close semantic correspondences were shown to exist with
certain versions of Petri nets, called \emph{nets with boundaries}. They come equipped with left- and right-interfaces to be exploited for composition. Interfaces are just plain lists of \emph{ports} (not just shared places) that are used to coordinate the firing of net transitions with the surrounding environment.

%Petri nets vs process algebra
Petri nets~\cite{Reisig1985} are frequently used both in theoretical and applied research to specify systems and visualise their behaviour.  On the other hand, process algebras are built around the principle of compositionality: their semantics is given \emph{structurally}
so that the behaviour of the whole system is a function of the behaviour of its subsystems. As a consequence, the two are associated with different modelling methodologies and reasoning techniques. This paper improves and extends the results of~\cite{DBLP:conf/concur/Sobocinski10,DBLP:conf/concur/BruniMM11}, 
which were initial and fragmented in the two aforementioned papers.
Our results bridge the gap between the Petri net theory and process algebra by showing very close semantic correspondence between a family of process algebras based on connectors on the one hand and a family of nets with boundaries on the other.
Still, we want to stress out the fact that our operators for composition of systems and interaction are fundamentally different to those traditionally considered by process algebraists. 

%Semantics of bounded nets
As usual, in the case of Condition/Event systems (C/E nets), each place can contain a token at most, and 
transitions compete both for resources in their presets and their postsets---two transitions that produce a token
at the same place cannot fire together.
In the case of Place/Transition systems (P/T nets), each place can contain an unbounded number of tokens, arcs between places and transitions are weighted, with the weights defining how many tokens are produced/consumed in each place by a single firing of the transition, and the firing of a transition is allowed also when some tokens are already present in its post-set. In both cases, ports of the interface can be 
connected to
transitions to account for the interactions with the environment when a transition fires. 

We focus on the step semantics, where (multi)sets of transitions can fire at the same time.
In the case of P/T nets we consider two different kinds of semantics: 
an ordinary firing semantics in which a concurrently enabled multiset of transitions
can fire together, as well as a second semantics in which any 
 multiset of transitions can
 fire together when the number of tokens consumed from each place does not exceed the number of tokens available at the beginning plus those that are produced. This means that not all of the transitions are necessarily enabled at
 the start: by analogy with the bank system, we can consider that the multiset of transitions is enabled by each place in the net initially taking some ``loan'' tokens that are given back after the firing.
Because of this analogy we will refer to this semantics as the \emph{banking semantics}.
The weak semantics resembles the firing condition for P/T nets with a/sync places proposed in~\cite{DBLP:conf/birthday/KleijnK11,Kleijn2012185,Kleijn2012189}, in which tokens in a/sync places  can be produced and consumed at the same execution step. 

In the case of C/E nets we also consider two different kinds of semantics: in the strong one, non-interfering sets of enabled transitions can fire at the same time; in the weak one, multisets of transitions can fire at the same time, as for P/T nets, as long as the capacity of places is not exceeded after the firing.
Still, several alternatives are also possible, depending on the order in which the tokens are assumed to be consumed and produced during the step. For example, if we assume that first all transitions consume the tokens and then new tokens are produced, we have a step semantics that is more liberal than the strong one, but stricter than the weak one.
Essentially, the possible different semantics are those  studied for nets with (place) capacities in~\cite{DBLP:conf/apn/Devillers88}, when 
regarding C/E nets as P/T nets with capacity one for all places. 
 All the alternatives are discussed in Remark~\ref{rmk:otherCEsemantics}, and the results presented in this paper smoothly extend to each variant. 

% Petri calculus
On the process algebra side, we call \emph{Petri calculus} the calculus of connectors accounting for C/E nets and \emph{P/T calculus} the one accounting for P/T nets. 
Quite surprisingly, we show that the same set of stateless connectors is needed to deal with C/E nets with boundaries and with P/T nets with boundaries. The difference is the use of one-state buffers for C/E nets and unbounded buffers for P/T nets.
Our studies also show that the correspondence results between connectors and nets carry over the preferred model of coordination, just depending on the absence or presence of a simple rule (called \ruleLabel{Weak}) for composing consecutive steps of the operational semantics, using a natural monoidal structure 
on the set of labels. Remark~\ref{rmk:otherCEsemantics} shows that the different semantics for C/E nets can be easily classified by changing the operational semantics rules for one-state buffers.

%Tiles
While the Petri calculus relies on a finite set of symbols and rules, one possible drawback of the P/T calculus is that it requires a finite \emph{scheme} of rules, that are parametric on some natural numbers. Then, we show that by using the \emph{tile model}~\cite{DBLP:conf/birthday/GadducciM00} this limitation can be overcome and P/T nets can be modelled using a finite set of symbols and tiles. The technical key to achieve the main result is the functoriality of the monoid of observations w.r.t. the so-called vertical composition of tiles. To be more precise, since interfaces are lists of ports and we want to observe, at each port, how many steps are performed and how many tokens are exchanged during each step, we take lists of sequences of natural numbers as observations. Since we want to deal with a finite set of symbols, we represent any natural number $n$ as the sequence of symbol $1$ of length $n$. Notably, the observation $0$ is just the identity of the category of observations. Roughly, the functoriality law of the monoid of observations establishes that observations at different ports are not necessarily ``aligned'' or synchronised. Yet, in the strong case, we want to separate the tokens exchanged in one step from the tokens exchanged at the next step. This is achieved by introducing an additional symbol $\tau$ as a separator and we show that it can be used to align independent sequences by a default policy.

%Summing up
Overall, the Petri calculus and tile model provide small, basic algebras of nets, out of which we can build any C/E and P/T nets with boundaries compositionally. As discussed in the section on related work, this result provides our personal answer to a long-standing quest for the universal algebra,
both sound and complete, of nets. Although we are aware that the constants we start from reduce nets to their very basic atoms and hence their algebra is very fine grained and cannot provide by itself the right level of abstraction for manipulating complex systems, we argue that one can still look for building suitable ``macros'' as derived operators on our basic atoms and then work in the corresponding subalgebra. Note also that the only forms of composition we rely on are the parallel and sequential compositions that constitute essential operations and should always be present. We think the key novel issue in our setting is a simple but powerful notion of interface, that exposes ``pending arcs'', unlike classical approaches, where places and/or transitions are exposed. Additionally, it allows to attach many competing pending arcs to the same port.

\subsubsection*{Origin of the work.}
In~\cite{DBLP:conf/concur/Sobocinski10} the fourth author employed essentially the same stateful extension of the connector algebra to compose Condition-Event (C/E) Petri nets (with consume/produce loops). 
Technically speaking, the contribution in~\cite{DBLP:conf/concur/Sobocinski10} can be summarised as follows.
C/E nets with boundaries are first introduced that can be composed in series and in parallel and come equipped with a bisimilarity semantics.
Then, a suitable instance of the \emph{wire calculus} from~\cite{DBLP:journals/corr/abs-0912-0555} is presented,
called \emph{Petri calculus}, that models circuit diagrams with one-place buffers and interfaces.
The first result enlightens a tight semantics correspondence: it is shown that a Petri calculus process can be defined for each net such that the translation preserves and reflects the  semantics.
The second result provides the converse translation, from Petri calculus to nets. Unfortunately, some problems arise in the latter direction that complicate a compositional definition of the encoding: Petri calculus processes must be normalised before  translating them, via a set of transformation rules that add new buffers to the circuit (and thus new places to the net).
The difference between the work in~\cite{DBLP:conf/concur/Sobocinski10} and the results presented in this paper are:
i) by improving the definition of C/E nets with boundaries
we simplify the translation from Petri calculus to nets, avoiding the normalisation procedure and giving a compositional encoding;
ii) the weak semantics is novel to this paper.
The idea of composing nets via boundaries made of ports was novel to~\cite{DBLP:conf/concur/Sobocinski10}.

In~\cite{DBLP:conf/concur/BruniMM11} the first three authors exploited the tile model to extend the correspondence result of~\cite{DBLP:conf/concur/Sobocinski10} to deal with P/T nets with boundaries, providing an elegant and compositional translation from the relevant tile model to P/T nets that neither involves normalising transformation, nor introduces additional places. During the preparation of this full version, we realised that since the $\tau$ observations were not considered there, the semantics addressed in the correspondence was the weak one, not the strong one. As a consequence, the main theorem, stating the correspondence in both directions, worked in one direction only (from nets to tiles) and not in the opposite direction (tiles allowed for more behaviours than nets). 
The difference between the work in~\cite{DBLP:conf/concur/BruniMM11} and the results presented in this paper are:
i) we changed the arity of the symbol for modelling tokens (from arity $(1,1)$ to $(0,1)$) because we found it more convenient in many proofs (but the change has no consequences whatsoever on the overall expressiveness of the model);
ii) we fixed the correspondence theorems for the strong case by introducing the $\tau$ observations (only one basic tile needs to be adjusted);
iii) we fixed the correspondence theorems for the weak case by finding a more compact and elegant presentation of the P/T net semantics (in terms of multisets of transitions instead of processes).
Incidentally the idea of the banking semantics for our weak coordination model originated from the tile semantics in~\cite{DBLP:conf/concur/BruniMM11}.

The definition of the P/T calculus is also a novel contribution of this paper. Its main advantages are:
i) in the strong case, it can be seen as the natural extension of the Petri calculus (where only $0$ and $1$ are observed) to deal with P/T nets (where any natural number can be observed);
ii) the extension to the weak case relies on exactly the same rule as the Petri calculus (\ruleLabel{Weak});
iii) it offers a convenient intermediate model for proving the correspondence between the tile model and the P/T nets with boundaries.

\subsubsection*{Roadmap.} 
The content of this paper is of a rather technical nature but is self-contained, in the sense that we do not assume the reader be familiar with nets, process algebras, category theory or tile model.
As it may be evident by the above introduction, this work addresses the expressiveness of connectors models along several dimension: i) semantics, we can move from the strong view (``clockwork'' steps) to the weak view (that matches with banking semantics); ii) models, we can move from C/E nets to P/T nets; iii) algebras, we can move from the Petri calculus and P/T calculus to instances of the tile model.

The first part of the paper is devoted to
 two categories of nets with boundaries, C/E nets and P/T nets.  
The transitions of the composed net are minimal synchronisations 
(see Definitions~\ref{defn:CEsynchronisation} and \ref{defn:PTsynchronisation}) 
of transitions of the original nets.
To each model of net we assign a labelled semantics, in the case
of P/T nets we study both a strong semantics and a weak semantics that
captures the banking semantics of P/T nets.
The key results (Theorem~\ref{thm:netdecomposition} for
C/E nets and Theorem~\ref{thm:ptnetdecomposition} for P/T nets) 
are that labelled transitions
are compatible with composition of nets. These results
guarantee that (labelled) bisimilarity of nets is always compositional.

Next we study the process algebraic approaches. First the Petri calculus,
with a strong and weak semantics. The important result is Proposition~\ref{pro:petricongruence} which states that both strong
and weak bisimilarity is a congruence with respect to the two operations.
Next we extend the Petri calculus with unbounded buffers, obtaining
the P/T calculus, again with appropriate strong and weak semantics.
We then develop enough theory of the calculi
to show that  they are semantically
equivalent to their corresponding model of nets with boundaries. Our final
technical contribution is a reformulation of the P/T calculus in the tile framework.

%All combinations are studied here and consequentially we had several options for organising their order of presentation. One possibility would have been to start from the most general case and then present the other options as suitable instances. But the C/E and P/T policies are different in nature, so there is no ``more general'' case. So we decided to structure the presentation in such a way that most technicalities of later options are discussed incrementally, by difference with preceding options. 
%The C/E case under strong semantics relies on bounded buffers and simpler ($0$ / $1$) observations and therefore comes first, followed by the banking semantics, where natural numbers are observed.
%Then we move to P/T case, first strong, then weak.
%In both the C/E and P/T cases, we start by giving the process algebraic version, then fix the corresponding flavour of nets with boundaries and prove the correspondence theorems between the two. In the P/T case, we end by showing a minimal tile model that is equivalent to the P/T calculus.
%The reader will only find in the main text most lemmas that assert key semantics properties and the correspondence theorems. Most technical lemmas and some involved proofs are relegated to the Appendices. 
%We recommend the reader to follow the order of presentation; e.g., although the part on P/T nets or the part on tiles can be read without strong dependencies on the part on C/E nets, the content can be better understood by observing the progressive enhancement of the semantics frameworks.

\subsubsection*{Structure of the paper.} 
In detail, the paper is structured as follows:
Section~\ref{sec:background} fixes the main notation and gives the essential background on 
C/E nets and P/T nets.
Section~\ref{sec:nets} introduces C/E nets with boundaries, together with their
labelled semantics. Section~\ref{sec:ptboundaries} introduces P/T nets with boundaries, under both the strong and weak labelled semantics.  In Section~\ref{sec:properties}
we show that both the models are actually monoidal categories and that there are
functors that take nets to their underlying \emph{processes}---bisimilarity classes with
respect to the labelled semantics.
Section~\ref{sec:syntax} introduces the Petri calculus, fixing its syntax, its strong and weak operational semantics and the corresponding bisimulation equivalences. 
P/T calculus, introduced in Section~\ref{sec:petritile},  extends the Petri calculus  by allowing unbounded buffers and by generalising the axioms of the Petri calculus to deal with natural numbers instead of just $0$ and $1$. 
In Section~\ref{sec:syntaxtonets} we translate process algebra terms to nets; these
translations are easy because there are simple nets that account
for the basic connectors and so our translations can be defined compositionally.
In Section~\ref{sec:netstosyntax} we develop enough of the process algebra
theory thats allow us to give a translation from net models
to the process algebras.  All the translations
in Sections~\ref{sec:syntaxtonets} and \ref{sec:netstosyntax} both preserve and reflect labelled transitions.
Section~\ref{sec:petri-tile-calculus} recasts the P/T calculus within the tile model. First, some essential definition on the tile model are given. Then, an instance of the tile model, called \emph{Petri tile model}, is introduced. In the strong case the tile model includes a special observation $\tau$ that is used to mark a separation between the instant a token arrives in a place and the instant it is consumed from that place. In the weak case, the $\tau$ are just (unobservables) identities, so that the same token can arrive and depart from a place in the same step. The main result regarding the tile model shows that the Petri tile calculus is as expressive as the P/T calculus and therefore, by transitivity, as the P/T nets with boundaries. 
Section~\ref{sec:related} accounts for the comparison with some  strictly related approaches in the literature.
Finally, some concluding remarks are given in Section~\ref{sec:conclusions}.

\section{Background}
\label{sec:background}

For $n\in\N$ write $\underline{n}\Defeq\{0,1,\dots,n-1\}$ for the $n$th ordinal (in particular, $\underline{0}\Defeq \emptyset$).
For sets $X$ and $Y$ we write $X+Y$ for $\{(x,0)\;|\;x\in X\}\union \{(y,1)\;|\;y\in Y\}$.
A multiset on a set $X$ is a function $X\to \N$.
The set of multisets on $X$ is denoted $\multiset{X}$. We shall use
$\mathcal{U},\mathcal{V}$ to range over $\multiset{X}$.
For $\mathcal{U},\mathcal{V} \in \multiset{X}$, we write 
$\mathcal{U}\subseteq\mathcal{V}$ iff $\forall x\in X: \mathcal{U}(x)\le \mathcal{V}(x)$. 

We shall frequently use the following operations on multisets:
\begin{eqnarray*}
\cup:\multiset{X}\times \multiset{X}\to\multiset{X}: && (\mathcal{U} \cup \mathcal{V})(x) \Defeq \mathcal{U}(x) + \mathcal{V}(x)
%\end{equation*}
\\
-:\multiset{X}\times \multiset{X}\to\multiset{X}: && (\mathcal{U} - \mathcal{V})(x) \Defeq \mathcal{U}(x) - \mathcal{V}(x)
\mbox{ when }\mathcal{V}\subseteq \mathcal{U}
%\end{equation*}
\\
%\begin{equation*}
\cdot:\N\times\multiset{X} \to \multiset{X}: && (k\cdot \mathcal{U})(x) \Defeq k \mathcal{U}(x)
%\end{equation*}
\\
%\begin{equation*}
+: \multiset{X} \times \multiset{Y} \to \multiset{X+Y}:
&& (\mathcal{U} + \mathcal{V})(z) \Defeq 
\begin{cases}  
\mathcal{U}(x) \mbox{ if }z = (x,0) \\
\mathcal{V}(y) \mbox{ if }z = (y,1)
\end{cases}
%\end{equation*} 
\\
%\begin{equation*}
-|_Y: \multiset{X} \to \multiset{Y} : &&
\mathcal{M}_Y(y) \Defeq \mathcal{M}_X(y) \mbox{ when }Y\subseteq X
\end{eqnarray*}
Given a finite set $X$ and
$\mathcal{U}\in\multiset{X}$ let $|\mathcal{U}|\Defeq \sum_{x\in X} \mathcal{U}(x)$.
Given a finite $X$,
if $f:X\to \multiset{Y}$ and $\mathcal{U}\in\multiset{X}$ then
we shall abuse notation and write 
$f(\mathcal{U})=\bigcup_{x\in X} \mathcal{U}(x)\cdot f(x)$. Another slight abuse of notation will be the use of $\varnothing\in\multiset{X}$ for the multiset s.t. $\varnothing(x)=0$ for all $x\in X$. 
 
Given $f: X \rightarrow Y$ and $U\subseteq Y$, we will write $f^{-1}(U)$ to the denote the inverse image (or preimage) of the set $U$  under $f$, i.e.,   $f^{-1}(U) = \{x\in X \,|\, f(x) \in U\}$. 

\medskip
Throughout this paper we use two-labelled transition systems (cf. Definition~\ref{defn:component}). 
Depending on the context, labels will be words 
in $\{0,1\}^*$ or $\N^*$, and will be ranged over by $\alpha$, $\beta$, $\gamma$. 
Write $\len{\alpha}$ for the length of a word $\alpha$.
Let $i\in[1,\len{\alpha}]$, we denote by $\alpha_{i}$ the $i$th element of $\alpha$.
%We denote by $\#\alpha$ the length of the sequence $\alpha$. 
Let $\alpha,\beta\in \nat^{*}$ with $\len{\alpha} = \len{\beta}$, then we denote by $\alpha+\beta$ the sequence
such that $\len{(\alpha+\beta)} = \len{\alpha}$ and 
$(\alpha+\beta)_{i} = \alpha_{i}+\beta_{i}$ for any $i\in[1,\len{\alpha}]$.

The intuitive idea is that a transition $p\dtrans{\alpha}{\beta}q$
means that a system in state $p$ can, in a single step, 
synchronise with $\alpha$ on its left
boundary, $\beta$ on it right boundary and change its state
to $q$. 
%The notion of interaction on left and right boundaries should
%not be confused with input and output. 

\begin{defi}[Two-labelled transition system]\label{defn:component}
Fix a set of labels $A$ (in this paper $A=\{0,1\}$ or $A=\N$).
For $k$, $l\in\N$, a 
\emph{$\sort{k}{l}$-transition}
is a two-labelled transition of the form 
$\dtrans{\alpha}{\beta}$ where $\alpha,\beta\in A^*$, 
$\len{\alpha}=k$ and $\len{\beta}=l$.
%A $\sort{k}{l}$-\emph{component} $\mathscr{C}$
%is a \emph{pointed}, \emph{reflexive} and \emph{transitive} 
%
A \emph{$\sort{k}{l}$-labelled transition system} (\sortedLTS{k}{l}) is a transition system
that consists of $\sort{k}{l}$-transitions: concretely, it is a pair $(V,T)$ where
$V$ is a set of states, and $T\subseteq V\times A^* \times A^* \times V$,
where for all $(v,\alpha,\beta,v')\in T$ we have $\len{\alpha}=k$ and
$\len{\beta}=l$.
%Concretely, it is a triple
%$(v_0,V,T)$ where 
%$V$ is a set of states with $v_0\in V$ a \emph{chosen state} and
%$T\subseteq V\times ((\Sigma+\{\iota\})^k\times(\Sigma+\{\iota\})^l) \times V$.
%Reflexivity means that
%for all $v\in V$ there exists a transition
%$v \dtrans{\iota^k}{\iota^l} v$; transitivity that if 
%$v \dtrans{\iota^k}{\iota^l} v''$, $v''\dtrans{\vec{a}}{\vec{b}} v'''$
%and $v'''\dtrans{\iota^k}{\iota^l}v'$ then also
%$v \dtrans{\vec{a}}{\vec{b}} v'$.
%
A \emph{two-labelled transition system} is a family of $\sort{k}{l}$-labelled transition systems for $k,l\in\N$.
\end{defi}

%Note that when the underlying \textsc{lts} is clear from the context we shall
%often identify a component with the chosen state of its \textsc{lts},
%writing $v_0\typ\sort{k}{l}$
%to indicate that $v_0$ is (the chosen state of) a $\sort{k}{l}$-component.

\begin{defi}[Bisimilarity]
A \emph{simulation} on a two-labelled transition system is a relation $S$ on its set of states that
satisfies the following:  if $(v,w)\in S$ and 
$v\dtrans{\alpha}{\beta}v'$ then $\exists w'$ s.t.\ 
$w\dtrans{\alpha}{\beta}w'$ and $(v',w')\in S$. 
A \emph{bisimulation} is a relation $S$ where both $S$ and $S^{op}$, the inverse (or opposite) relation, are simulations.
\emph{Bisimilarity} is the largest bisimulation relation and can be obtained
as the union of all bisimulations.
\end{defi}

\subsection{Petri Nets}
\label{sec:cenets}
Here we introduce the underlying models of nets, together with the different
notions of firing semantics that are considered in the paper.
%STRONG AND WEAK NETS
\begin{defi}[C/E net]\label{defn:net}
A \emph{C/E net} is a 4-tuple $N=(P,\,T,\,\pre{-},\,\post{-})$ where:\footnote{In the context of C/E nets some authors call places \emph{conditions} and transitions \emph{events}.}
\begin{iteMize}{$-$}
\item $P$ is a set of \emph{places};
\item $T$ is a set of \emph{transitions};% such that $P\cap T = \emptyset$;
\item $\pre{-},\post{-}\from T \to 2^P$ are functions.
\end{iteMize}
A C/E net $N$ is \emph{finite} when both $P$ and $T$ are finite sets. 
For a transition $t\in T$, $\pre{t}$ and $\post{t}$ are called, respectively, 
its \emph{pre-} and  \emph{post-sets}. Moreover, we write $\preandpost{t}$ for $\pre{t}\cup\post{t}$.
%In a \emph{weak} C/E net, transitions can be connected to multisets of places, thus 
%we change the type of the functions $\pre{-},\post{-}\from T\to \multiset{P}$.
\end{defi} 

The obvious notion of net homomorphisms $f\from N\to M$ is a pair of functions 
$f_T\from T_N\to T_M$, $f_P\from P_N\to P_M$ such that $\pre{-}_N \seqComp 2^{f_P}=f_T \seqComp \pre{-}_M$
and $\post{-}_N\seqComp 2^{f_P}=f_T \seqComp \post{-}_M$, where $2^{f_P}(X)=\bigcup_{x\in X}\{f_P(x)\}$.

Notice that Definition~\ref{defn:net} allows transitions $t$ with \emph{both} empty
pre- and post-sets, that is, $\pre{t}=\post{t}=\varnothing$. 
Such transitions (e.g.,  transition $\zeta$ in \fig{\ref{fig:boundedNet}}), while usually excluded for ordinary nets, are
necessary when defining nets with boundaries in 
Section~\ref{sec:nets} (see Definition~\ref{defn:boundedcenets}).

%INDEPENDENCE OF TRANSITIONS
%For $t\in T$, let $\preandpost{t}\Defeq \pre{t} \cup \post{t}$.
Transitions $t\neq u$ are said to be \emph{independent} when 
\[
\pre{t}\intersection\pre{u}=\varnothing
\quad\text{and}\quad
\post{t}\intersection\post{u}=\varnothing.
%\preandpost{t} \intersection \preandpost{u} = \varnothing.
\]
A set $U$ of transitions is said to be \emph{mutually independent} when
%for all $t\in U$, $t$ is consume/produce loop free, and 
for all $t,u\in U$, if $t\neq u$ then $t$ and $u$ are independent.
%In~\cite{DBLP:conf/concur/Sobocinski10} a more liberal notion of independence was considered.

%Note that this notion of independence is quite liberal
%\footnote{In the literature it is common to require 
%that $\preandpost{t}\intersection\preandpost{u}=\varnothing$.}
%%\footnote{In the literature it is common to require 
%%that $\preandpost{t}\intersection\preandpost{u}=\varnothing$.}
%and allows contact situations.

Given a set of transitions $U$ let $\pre{U}\Defeq\bigcup_{u\in U}\pre{u}$
and $\post{U}\Defeq\bigcup_{u\in U}\post{u}$.

Given a net $N=(P,T,\pre{-},\post{-})$, a \emph{(C/E) marking} is a subset of places $X\subseteq P$. We shall use the notation $\marking{N}{X}$ to denote the marking $X$ of net $N$. 
%States of the transition systems will be referred to as \emph{markings} of $N$.

%STRONG NET SEMANTICS
\begin{defi}[C/E firing semantics]\label{defn:CENetStrongSemantics}
Let $N=(P,\,T,\,\pre{-},\,\post{-})$ be a C/E net, $X,Y\subseteq~P$ and 
%write:
%\[
%N_X \rightarrow_{\{t\}} N_Y \quad\Defeq\quad 
%\pre{t}\subseteq X,\ \post{t}\subseteq Y \ \&\ X\backslash \pre{t} 
%	= Y\backslash\post{t}.
%\]
for $U\subseteq T$ a set of mutually independent transitions, write:
%\[
%N_X \rightarrow_{U} N_Y \quad\Defeq\quad
%\pre{U} \subseteq X,\ 
%\post{U} \subseteq Y\ \&\ 
%X\backslash \pre{U} =
%Y\backslash \post{U}.
%\]
\[
N_X \rightarrow_{U} N_Y \quad\Defeq\quad
\pre{U} \subseteq X,\ 
\post{U} \cap X = \varnothing \ \&\ 
Y = (X \backslash \pre{U}) \cup \post{U}.
\]
\end{defi}

%CONSUME-PRODUCE LOOPS
\begin{rem}\label{rmk:consumeproduceloops}
Notice that Definition~\ref{defn:net} allows the presence of transitions $t$ for which there
exists a place $p$ with $p\in \pre{t}$ and $p\in \post{t}$.
Some authors refer to this as a consume/produce loop. The semantics in Definition~\ref{defn:CENetStrongSemantics} implies that such transitions can never fire. We will return to this in Remark~\ref{rmk:loops}, and in Remark~\ref{rmk:otherCEsemantics} where we consider alternative semantics for nets with boundaries. \qed
\end{rem}
%the notion of \emph{contextual net}~\cite{Montanari1995} 
%is related. 
%A transition $t$ is consume/produce loop free if $\pre{t}\cap \post{t} = \varnothing$.

Places of a {\em Place/Transition net} (P/T net) can hold zero, one or more tokens and arcs are weighted. The state of a P/T net is described in
terms of \emph{(P/T) markings}, \ie, (finite) multisets of tokens available in the places of the net. 

\begin{defi}[P/T net]\label{def:ptnet}
  A \emph{P/T net} is a 4-tuple  $(P,\,T,\,\pre{-},\,\post{-})$  where:
    \begin{iteMize}{$\bullet$}
      \item $P$ is a set of \emph{places};
      \item $T$ is a set of \emph{transitions};% such that $P\cap T = \emptyset$;
      \item $\pre{-},\post{-}\from T \to  \mathcal{M}_{P}$.
    \end{iteMize}  
\end{defi}

Let $\mathcal{X} \in \mathcal{M}_{P}$, we write $N_{\mathcal{X}}$ for the P/T net $N$  with  
 marking ${\mathcal{X}}$.

We can extend $\pre{-}$ and $\post{-}$ in the obvious way to multisets of transitions:
for $\mathcal{U}\in\multiset{T}$ define
$\pre{\mathcal{U}} \Defeq \bigcup_{t\in T}  \mathcal{U}(t) \cdot \pre{t}$
and similarly
$\post{\mathcal{U}} \Defeq \bigcup_{t\in T} \mathcal{U}(t) \cdot \post{t}$ .

\begin{defi}[P/T strong firing semantics] \label{def:strong-firing-pt}
Let $N=(P,\,T,\,\pre{-},\,\post{-})$ be a P/T net, ${\mathcal{X}},{\mathcal{Y}}\in\mathcal{M}_P$ and $t\in T$. 
%Write:
%\[
%N_{\mathcal{X}} \rightarrow_{\{t\}} N_{\mathcal{Y}} \quad\Defeq\quad 
%\pre{t}\subseteq {\mathcal{X}},\ \post{t}\subseteq {\mathcal{Y}} \ \&\ %{\mathcal{X}}- \pre{t} = {\mathcal{Y}}-\post{t}.
%\]
For $\mathcal{U}\in\mathcal{M}_T$ a multiset of transitions, write:
\[
N_{\mathcal{X}} \rightarrow_{\mathcal{U}} N_{\mathcal{Y}} \quad\Defeq\quad
\pre{\mathcal{U}} \subseteq {\mathcal{X}},\ 
\post{\mathcal{U}} \subseteq {\mathcal{Y}}\ \&\ 
{\mathcal{X}}- \pre{\mathcal{U}} =
{\mathcal{Y}}- \post{\mathcal{U}}.
\]
\end{defi}

Although the conditions $\pre{\mathcal{U}} \subseteq {\mathcal{X}}$ and   
$\post{\mathcal{U}} \subseteq {\mathcal{Y}}$ in the above definition are redundant (since ${\mathcal{X}}- \pre{\mathcal{U}}$ and 
${\mathcal{Y}}- \post{\mathcal{U}}$ are defined only under such assumption), we explicitly state them in order to stress this requirement for 
firing. Also, 
we remark that Definition~\ref{defn:CENetStrongSemantics} can be obtained as a special case of Definition~\ref{def:strong-firing-pt} when considering only 1-safe markings, i.e., markings that hold at most one token. Indeed, the conditions $\pre{\mathcal{U}} \subseteq {\mathcal{X}}$ and $\ 
\post{\mathcal{U}} \subseteq {\mathcal{Y}}$ with $\mathcal{X}$ and $\mathcal{Y}$ 1-safe only holds when  
$\mathcal{U}$ is a set of mutually independent transitions.

\begin{defi}[P/T weak firing semantics]\label{defn:PTWeakSemantics}
Let $N=(P,\,T,\,\pre{-},\,\post{-})$ be a P/T net, ${\mathcal{X}},{\mathcal{Y}}\in\mathcal{M}_P$ and $\mathcal{U} \in \multiset{T}$. 
Write:
\[
N_{\mathcal{X}} \Rightarrow_{\mathcal{U}} N_{\mathcal{Y}} \quad\Defeq\quad
{\mathcal{Y}} \cup \pre{\mathcal{U}} =
{\mathcal{X}} \cup \post{\mathcal{U}}.
\]
\end{defi}

%Given $f:S \rightarrow S'$ (or even $f:S \rightarrow \mathcal{M}_{S'}$)
% we denote the canonical extension of $f$ to $\mathcal{M}_{S}$
 %with the symbol $f$ itself, i.e., we let $f(\mathcal{X} + \mathcal{Y}) = f(\mathcal{X}) + f(\mathcal{Y})$ and 
 %$f(\emptyset) = \emptyset$. 
  Let $N, M$ be P/T nets, a net homomorphism $f\from N\to M$ is a pair of functions 
 $f_T\from T_N\to T_M$, $f_P\from P_N\to P_M$ such that 
 such that $\pre{-}_N \seqComp f_P =f_T \seqComp \pre{-}_M$
and $\post{-}_N\seqComp f_P=f_T \seqComp \post{-}_M$.  
 
\begin{exa}
Figure~\ref{fig:simple-pt} depicts a simple P/T net $N$. 
We use  the traditional graphical representation in which places are circles, transitions are rectangles and directed edges connect transitions
to its pre- and post-set. 
 When considering the strong semantics, the net $N_{\{{p_1}\}}$ can
evolve as follows: $N_{\{{p_1}\}} \rightarrow_{\{{t_1}\}} N_{\{{p_2}\}} \rightarrow_{\{{t_2}\}} N_{\{{p_1}\}} \ldots$. 
We remark that transition ${t_2}$ cannot be fired at $N_{\{p_1\}}$ since the side condition $\mathcal{X} - \pre{t}$ of Definition~\ref{def:strong-firing-pt}  is not satisfied (in fact, $\{{p_1}\}-\pre{{t_2}}$ is not defined).    When considering the weak semantics, the net $N_{\{{p_1}\}}$  has additional transitions such as $N_{\{{p_1}\}} \Rightarrow_{\{{t_1},{t_2}\}} N_{\{{p_1}\}}$, in which ${t_2}$ can be fired by consuming in advance  the token that will be produced by  ${t_1}$. 
\end{exa}

\begin{figure*}[t]
\includegraphics[height=3cm]{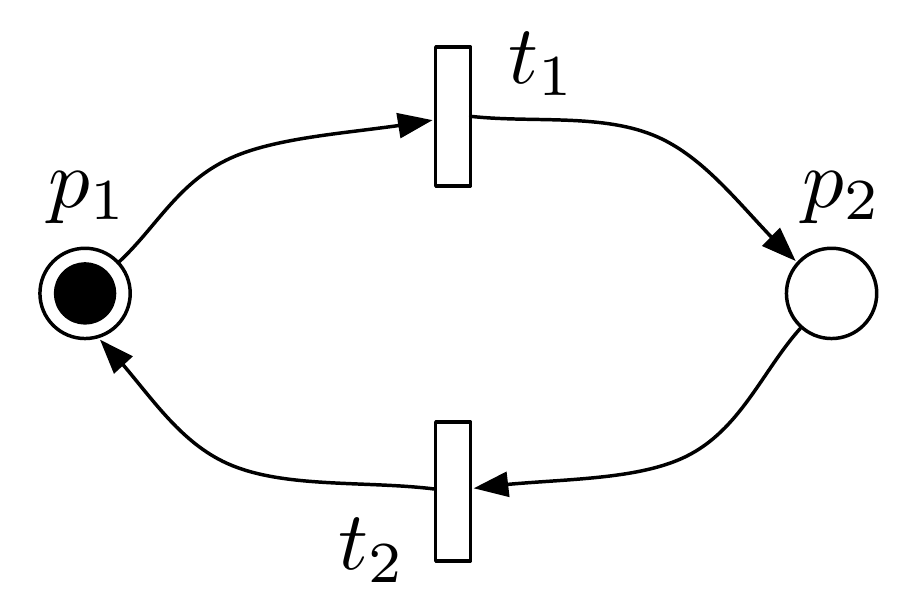}
%\subfigure[A simple {P/T} net $N$.]{
%\begin{math}
% \xymatrix@R=.2pc@C=3pc{
% &
% *=<.9pc,2pc>[F-]{\mathtt{t_1}}\ar@/^/[dr]
% &
%\\
% *[o]++[F-]{\phantom{a}}\ar@/^/[ur]
% \POS+<0pc,1.5pc>\drop{\mathtt{p_1}}
% &
% &
% *[o]++[F-]{\phantom{a}}\ar@/^/[dl]
% \POS+<0pc,1.5pc>\drop{\mathtt{p_2}}
%\\
% &
% *=<.9pc,2pc>[F-]{\mathtt{t_2}}\ar@/^/[ul]
% &
%}
%\end{math}
\caption{A simple {P/T} net $N$.}
\label{fig:simple-pt}
\end{figure*}

%WEAK NET SEMANTICS
We need to consider another kind of weak semantics of P/T nets that is related to C/E nets
in that markings hold at most one token.

\begin{defi}[P/T restricted weak firing semantics]\label{defn:CENetWeakSemantics}
Let $N=(P,\,T,\,\pre{-},\,\post{-})$ be a P/T net, $X,Y\subseteq P$ and 
$\mathcal{U} \in \multiset{T}$. Write:
\[
N_X \Rightarrow_{\mathcal{U}} N_Y \quad\Defeq\quad
Y \cup \pre{\mathcal{U}} =
X \cup \post{\mathcal{U}}.
\]
where the operation $\cup$ refers to \emph{multiset} union and
the sets $Y$ and $X$ are considered as multisets.
\end{defi}

Note  that Definition~\ref{defn:CENetWeakSemantics} is a special case of Definition~\ref{defn:PTWeakSemantics}, when considering just 1-safe markings.

\section{C/E Nets with boundaries}
\label{sec:nets}

% C/E Nets, two semantics, two ways of giving boundaries
In Definition~\ref{defn:net} we recalled the notion of C/E 
nets together with a firing semantics in 
Definition~\ref{defn:CENetStrongSemantics}. 

%and a less standard ``banking'' semantics 
%in Definition~\ref{defn:CENetWeakSemantics}. 

In this section we introduce a way of extending C/E nets with boundaries
 that allows nets to be composed along a common boundary.
%
%This results in \emph{two} models of C/E nets with boundaries: one that generalises
%C/E nets with the traditional semantics and a second that generalises C/E nets
%with the banking semantics. 
%
We give a labelled semantics to C/E nets with boundaries in Section~\ref{sec:labelledSemantics}.
The resulting model is semantically equivalent to the strong semantics of
the Petri Calculus, introduced in Section~\ref{sec:syntax}; the translations
are amongst the translations described in Sections~\ref{sec:syntaxtonets} and~\ref{sec:netstosyntax}.
%The semantics-preserving translation between the
%process algebra and bounded C/E nets are  
%the topic of Section~\ref{sec:cetranslations}.

\smallskip

%GRAPHICAL REPRESENTATION
In order to illustrate marked C/E nets with boundaries, it will first be useful to change 
the traditional graphical representation of a net and use a representation closer
in spirit to that traditionally used in string 
diagrams.\footnote{See~\cite{Selinger2009} for 
a survey of classes of diagrams used to characterise free monoidal categories.}
The diagram on the left in \fig{\ref{fig:graphicalRepresentationsNets}} demonstrates the 
traditional graphical representation of a (marked) net. Places are circles; a 
marking is represented by the presence or absence of tokens. Each 
transition $t\in T$ is a rectangle; there are directed edges from each place 
in $\pre{t}$ to $t$ and from $t$ to each place in $\post{t}$. This 
graphical language is a particular way of drawing hypergraphs;
the right diagram in \fig{\ref{fig:graphicalRepresentationsNets}} demonstrates 
another graphical representation, more suitable for drawing 
nets with boundaries. Places are again circles, but each place has exactly 
two \emph{ports} (usually drawn as small black triangles): one in-port, which we shall usually draw on the left, and one out-port,
usually drawn on the right. Transitions are simply 
undirected \emph{links}---each link 
can connect to any number of ports. Connecting 
$t$ to the out-port of $p$ means that $p\in\pre{t}$, connecting $t$ to $p$'s 
in-port means that $p\in\post{t}$. The position of the
``bar'' in the graphical representation of each link is irrelevant, they are used solely to distinguish individual links. A moment's thought ought to convince the 
reader that the two ways of drawing nets are equivalent, in that they both faithfully 
represent
the same underlying formal structures.
%Variants of link graphs have been used to 
%characterise various free monoidal categories: see for 
%instance~\cite{Joyal1991,Abramsky2005}.

%In this section we shall consider two variants of C/E nets with boundaries:
%a ``strong'' variant that we will show matches the expressivity of
%Petri calculus with strong semantics.
%Moreover, we shall consider a ``weak'' variant where
%a translation from syntax to nets and nets to syntax are functors, 
%and induce an isomorphism
%of categories between the category of weak processes and the category
%of marked weak net processes. The isomorphism of categories ensures that there is
%a very close operational correspondence between the two models.

\begin{figure}[t]
\[
\includegraphics[height=4cm]{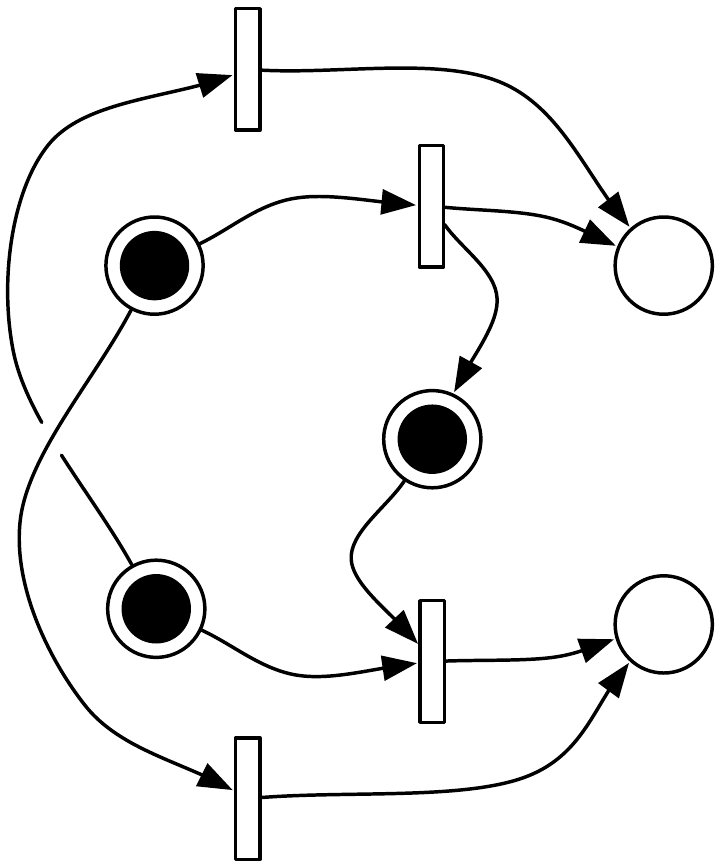}
\qquad
\includegraphics[height=4cm]{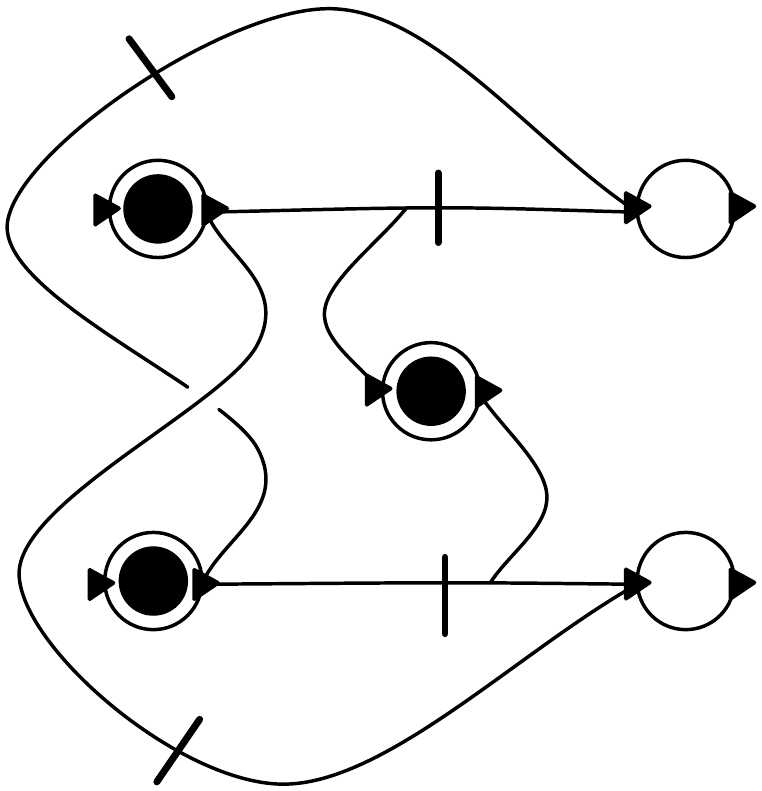}
\]
\caption{Traditional and alternative graphical representations of a net.\label{fig:graphicalRepresentationsNets}}
\end{figure}

Independence of transitions in C/E nets is an important concept---only independent transitions are permitted to fire concurrently. We will say that any two transitions $t$, $u$ with $t\neq u$ that are \emph{not} independent are in \emph{contention},
and write $t\# u$. Then, in ordinary C/E nets, $t\# u$ precisely when
$t\neq u$ and %$\preandpost{t}\cap\preandpost{u}\neq\emptyset$. 
$\pre{t}\cap\pre{u}\neq\emptyset$ or $\post{t}\cap\post{u}\neq\emptyset$.
In particular, the firing rule for the semantics of C/E nets 
(Definition~\ref{defn:CENetStrongSemantics})
can be equivalently restated as follows:
\[
%\marking{N}{X} \rightarrow_{U} \marking{N}{Y} \ \Defeq\  \pre{U}\subseteq X,\, 
%\post{U}\subseteq Y,\, X\backslash\pre{U} = Y\backslash \post{U}
%\ \&\  \forall u,v\in U.\, \neg(u \# v).
\marking{N}{X} \rightarrow_{U} \marking{N}{Y} \ \Defeq\  \pre{U}\subseteq X,\, 
\post{U}\cap X = \varnothing,\,  Y  = (X\backslash\pre{U}) \cup \post{U}
\ \&\  \forall u,v\in U.\, \neg(u \# v).
\]

Our models connect transitions to ports on boundaries.
Nets that share a common boundary can be composed---the transitions of 
the composed net are certain \emph{synchronisations} between the 
transitions of the underlying nets, as we will explain below. 
Connecting two C/E net transitions
to the same port on the boundary introduces a new source of contention---moreover
this information must be preserved by composition. For this reason
 the contention relation is an explicit part of the structure of C/E nets
with boundaries.

The model of C/E nets with boundaries
originally proposed in~\cite{DBLP:conf/concur/Sobocinski10} lacked
the contention relation and therefore the translation between Petri calculus
terms and nets was more involved. Moreover, the model of C/E nets with boundaries 
therein was less well-behaved in that composition was suspect; for
example bisimilarity was not a congruence with respect to it. Incorporating the
contention relation as part of the structure allows us to repair these shortcomings and obtain a simple translation of the Petri calculus that is similar to the other translations in this paper. 

%For these reasons, we believe that
%including the contention relation as part of the structure results in
%a more canonical mathematical structure.

% NETS WITH BOUNDARIES
%\subsection{C/E nets with boundaries}
\label{subsec:netsWithBoundaries}

We start by introducing a version of C/E nets with boundaries.
% that
%generalises ordinary C/E nets with the classical semantics.
%
%\subsubsection{Strong C/E nets with boundaries}
Let $\underline{k},\,\underline{l},\underline{m},\,\underline{n}$ range over 
finite ordinals.

\begin{defi}[C/E nets with boundaries]\label{defn:boundedcenets}
Let $m,n\in\N$.
A (finite, marked) \emph{C/E net with boundaries} 
$\marking{N}{X}\from m\to n$, is an 8-tuple 
$(P,\,T,X,\,\#,\,\pre{-},\,\post{-},\,\source{-},\,\target{-})$
where:
\begin{iteMize}{$-$}
\item $(P,\,T,\,\pre{-},\,\post{-})$ is a finite C/E net; 
\item $\source{-}\from T \to 2^{\underline{m}}$ and $\target{-}\from T\to 2^{\underline{n}}$
connect each transition to a set of ports on the \emph{left boundary} $\ord{m}$ and 
\emph{right boundary} $\ord{n}$;
\item $X\subseteq P$ is the \emph{marking};
\item $\#$ is a symmetric and irreflexive binary relation on $T$ called \emph{contention}.
\end{iteMize}
The contention relation must include all those transitions that
are not independent in the underlying C/E net, and those 
that share a place on the boundary, i.e. for all $t,u\in T$ where $t\neq u$:
\begin{enumerate}[(i)]
%\item if $\preandpost{t} \intersection \preandpost{u} \neq \varnothing$ then $t\# u$;
\item if $\pre{t} \intersection \pre{u} \neq \varnothing$, then $t\# u$;
\item if $\post{t} \intersection \post{u} \neq \varnothing$, then $t\# u$;
\item if $\source{t}\cap\source{u} \neq \emptyset$, then $t\# u$;
\item if $\target{t}\cap\target{u} \neq \emptyset$, then $t\# u$.
\end{enumerate}
Transitions $t,t'\in T$ are said to have the same \emph{footprint} when $\pre{t}=\pre{t'}$,
$\post{t}=\post{t'}$, $\source{t}=\source{t'}$ and $\target{t}=\target{t'}$.
From an operational point of view, transitions with the same footprint are indistinguishable. We assume that if $t$ and $t'$ have the same footprint then
$t=t'$. This assumption is operationally harmless and somewhat simplifies reasoning about composition.
\end{defi}

An example of C/E net with boundaries is pictured in \fig{\ref{fig:boundedNet}}. Note that  $\zeta$ is  a transition with  empty pre and postset, and 
transitions $\delta$ and $\gamma$ are in contention because they share a port.

\begin{figure}[t]
\[
\includegraphics[height=3cm]{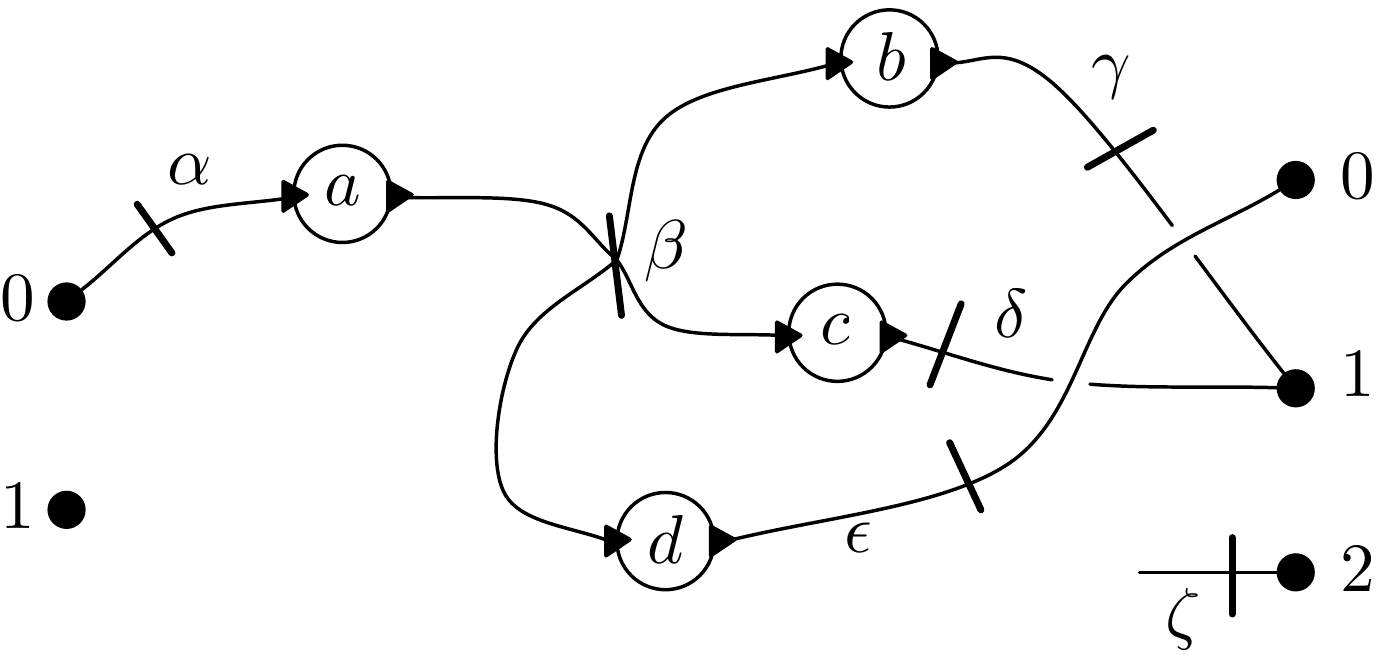}
\]
\caption{Representation of a net with boundaries $N_\emptyset:2\to 3$.
Here $T=\{\alpha,\,\beta,\,\gamma,\,\delta,\,\epsilon,\,\zeta\}$ and $P=\{a,\,b,\,c,\,d\}$.
The non-empty values of $\pre{-}$ and $\post{-}$ are:
$\post{\alpha}=\{a\}$, $\pre{\beta}=\{a\}$, $\post{\beta}=\{b,c,d\}$, 
$\pre{\gamma}=\{b\}$, $\pre{\delta}=\{c\}$, $\pre{\epsilon}=\{d\}$. The non-empty values of $\source{-}$ and $\target{-}$ 
are: $\source{\alpha}=\{0\}$, $\target{\epsilon}=\{0\}$, $\target{\gamma}=\{1\}$, $\target{\delta}=\{1\}$, $\target{\zeta}=\{2\}$. Of course when the same port name appears in the left and right boundaries (e.g., $0$) it denotes different nodes. 
\label{fig:boundedNet}}
\end{figure}

The notion of independence of transitions extends to C/E nets with boundaries: $t,u\in T$ are said to be \emph{independent} when 
$\neg (t\# u)$. We say that a set $U$ of transitions is \emph{mutually independent} if $\forall u,v\in U.\; \neg(u\# v)$.

%Henceforward we shall usually refer to marked C/E nets with boundaries as simply C/E nets.
The obvious notion of homomorphism between two C/E nets 
extends that of ordinary nets: given nets $\marking{N}{X},\marking{M}{Y}\from m\to n$,
$f\from \marking{N}{X}\to \marking{M}{Y}$ is a pair of functions $f_T\from T_N\to T_M$, $f_P\from P_N\to P_M$
such that $f_P(X)=Y$, $f_T(t)\# f_T(u)$ implies $t \# u$, 
$\pre{-}_N \seqComp 2^{f_P}=f_T \seqComp \pre{-}_M$,
$\post{-}_N\seqComp 2^{f_P}=f_T \seqComp \post{-}_M$, 
$\source{-}_N=f_T\seqComp \source{-}_M$
and $\target{-}_N=f_T\seqComp \target{-}_M$. 
A homomorphism is an isomorphism iff its two components
are bijections; we write $\marking{N}{X} \cong \marking{M}{Y}$ when there is an isomorphism
 from $\marking{N}{X}$ to $\marking{M}{Y}$. 

%COMPOSING NETS, SYNCHRONISATION
The main operation on nets with boundaries is \emph{composition along a common boundary}.
That is, given nets $\marking{M}{X}\from l \to m$, $\marking{N}{Y}\from m\to n$ we 
will define a net $\marking{M}{X};\marking{N}{Y}\from l\to n$.
Roughly, the transitions of the composed net $\marking{M}{X};\marking{N}{Y}$ are
certain sets of transitions of the two underlying nets that synchronise on
the common boundary. Thus in order to define the composition of nets along
a shared boundary, we must first introduce the concept of 
\emph{synchronisation}.

%STRONG SYNCHRONISATION
\begin{defi}[Synchronisation of C/E nets]\label{defn:CEsynchronisation}
Let $\marking{M}{X}\from l\to m$ and $\marking{N}{Y}\from m\to n$ be C/E nets.
A \emph{synchronisation} is a pair $(U,V)$, with 
$U\subseteq T_M$ and $V\subseteq T_N$ 
mutually independent sets of transitions such that:
\begin{iteMize}{$-$}
\item $U + V\neq\varnothing$;
\item $\target{U}=\source{V}$.
\end{iteMize}
The set of synchronisations inherits an ordering pointwise from the subset order, 
\ie\  we let $(U',\,V')\subseteq (U,\,V)$ when $U'\subseteq U$ and $V'\subseteq V$. 
A synchronisation 
is said to be \emph{minimal} when it is minimal with respect to this order. 
Let $Synch(M,N)$ denote the set of minimal synchronisations.
\end{defi}
Note that synchronisations do not depend on the markings of the underlying nets, but 
on the  sets of transitions $T_M$ and $T_N$. Consequently, $Synch(M,N)$ is 
finite because $T_M$ and $T_N$ are so. It could be also the case that $Synch(M,N)$ is 
the empty set .
Notice that any transition in $M$ or $N$ not connected to the shared 
boundary $m$ (trivially) induces a minimal synchronisation---for instance
if $t\in T_M$ with $\target{t}=\varnothing$, then $(\{t\},\,\varnothing)$
is a minimal synchronisation.  

The following result shows that any synchronisation can be decomposed into a set of minimal
synchronisations.
\begin{lem}\label{lem:strongTransitionDecomposition}
Suppose that $\marking{M}{X}: k\to n$ and 
$\marking{N}{Y}: n\to m$ are C/E nets with boundaries
and $(U,V)$ is a synchronisation. 
Then there exists a finite set of minimal synchronisations $\{(U_i,V_i)\}_{i\in I}$
such that (i) $U_i\cap U_j=V_i\cap V_j=\varnothing$ whenever
$i\neq j$, (ii) $\bigcup_i U_i = U$ and 
(iii) $\bigcup_i V_i = V$.
\end{lem}
\begin{proof}
See Appendix~\ref{app:proof-netswithboundaries}.
\end{proof}

Minimal synchronisations serve as the transitions of the composition of two nets along a common boundary. Thus, given $(U,\,V)\in Synch(M,N)$ let $\pre{(U,\,V)}\Defeq \pre{U}+\pre{V}$, $\post{(U,\,V)}\Defeq\post{U}+\post{V}$, $\source{(U,\,V)}\Defeq \source{U}$ and $\target{(U,\,V)}\Defeq\target{V}$. For $(U,\,V)$, $(U',\,V')\in Synch(M,N)$, $(U,V)\#(U',V')$ iff $(U,V)\neq (U',V')$ and 
\begin{iteMize}{$-$}
\item $U\cap U'\neq \varnothing$ or $\exists u\in U, u'\in U'$ such that
$u\# u'$ (as transitions of $M$), or 
\item $V\cap V'\neq \varnothing$ $\exists v\in V, v'\in V'$ such that $v\# v'$ (as transitions of $N$);
\end{iteMize}

%COMPOSING STRONG NETS
Having introduced minimal synchronisations
we may now define the composition of two C/E nets that share a common 
boundary. 

\begin{defi}[Composition of C/E nets with boundaries\label{defn:compositionce}]
When $\marking{M}{X}\from l\to m$ and $\marking{N}{Y}\from m\to n$ are C/E nets, 
define their composition, $\marking{M;N}{X+Y}\from l\to n$, as follows:
\begin{iteMize}{$-$}
\item places are  $P_M + P_N$, the ``enforced'' disjoint union of places of $M$ and $N$;
\item transitions are obtained from the set of minimal synchronisations $Synch(M,N)$, after removing any redundant transitions with equal footprint\footnote{It is possible that two or more minimal synchronisations share the same footprint and in that case only one is retained. The precise identity of the transition that is kept is irrelevant.};
\item the marking is $X+Y$.
\end{iteMize}
\end{defi}
We must verify that
$\#$ as defined on $Synch(M,N)$ above satisfies the conditions on the contention relation given in 
Definition~\ref{defn:boundedcenets}.
%Indeed, we have
%$\preandpost{(U,V)}\cap\preandpost{(U',V')}
%%= (\preandpost{U}\cup \preandpost{V})\cap 
%%(\preandpost{U'}\cup \preandpost{V'})
%=(\preandpost{U}\cap\preandpost{U'})\cup
% (\preandpost{V}\cap\preandpost{V'})$; hence
% $\preandpost{(U,V)}\cap\preandpost{(U',V')} \neq \varnothing$
%implies that at least one of $\preandpost{U}\cap\preandpost{U'}$
%and $\preandpost{V}\cap\preandpost{V'}$ are nonempty. Without
%loss of generality assume the first of these is nonempty, thus there exist 
%$u\in U$, $u'\in U'$ with 
%$\preandpost{u}\cap\preandpost{u'}\neq\emptyset$,
%so $u\# u'$ in $M$. The remaining two conditions 
%of Definition~\ref{defn:boundedcenets} are simpler to check.
Indeed if $\pre{(U,V)}\cap\pre{(U',V')} \neq \emptyset$
then one of $\pre{U}\cap\pre{U'}$ and $\pre{V}\cap\pre{V'}$
must be non-empty. Without loss of generality, if the first is nonempty
then there exist $u\in U$, $u'\in U'$ with $\pre{u}\cap\pre{u'}\neq\emptyset$, thus either $u=u'$, in which case $U\cap U'\neq \varnothing$, or
$u\# u'$ in $M$---thus by construction $\pre{(U,V)}\#\pre{(U',V')}$ in the composition,
as required. The remaining conditions are similarly easily shown to hold.
\begin{figure}[t]
\[
\includegraphics[height=3cm]{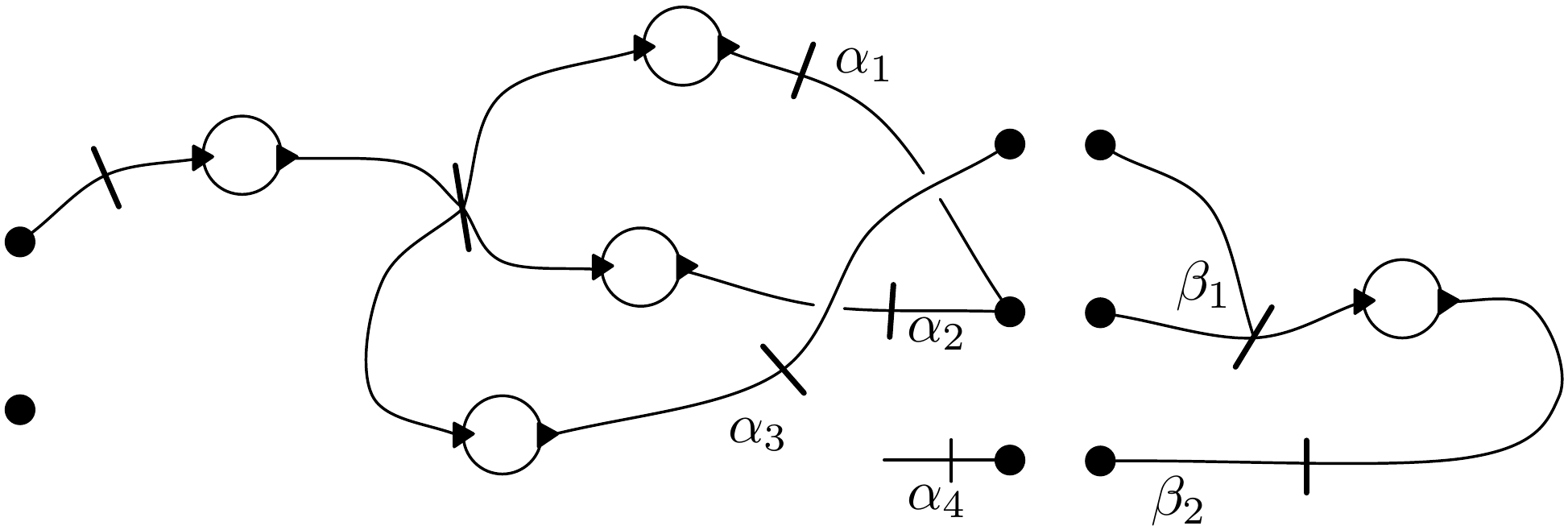}
\]
\[
\includegraphics[height=3cm]{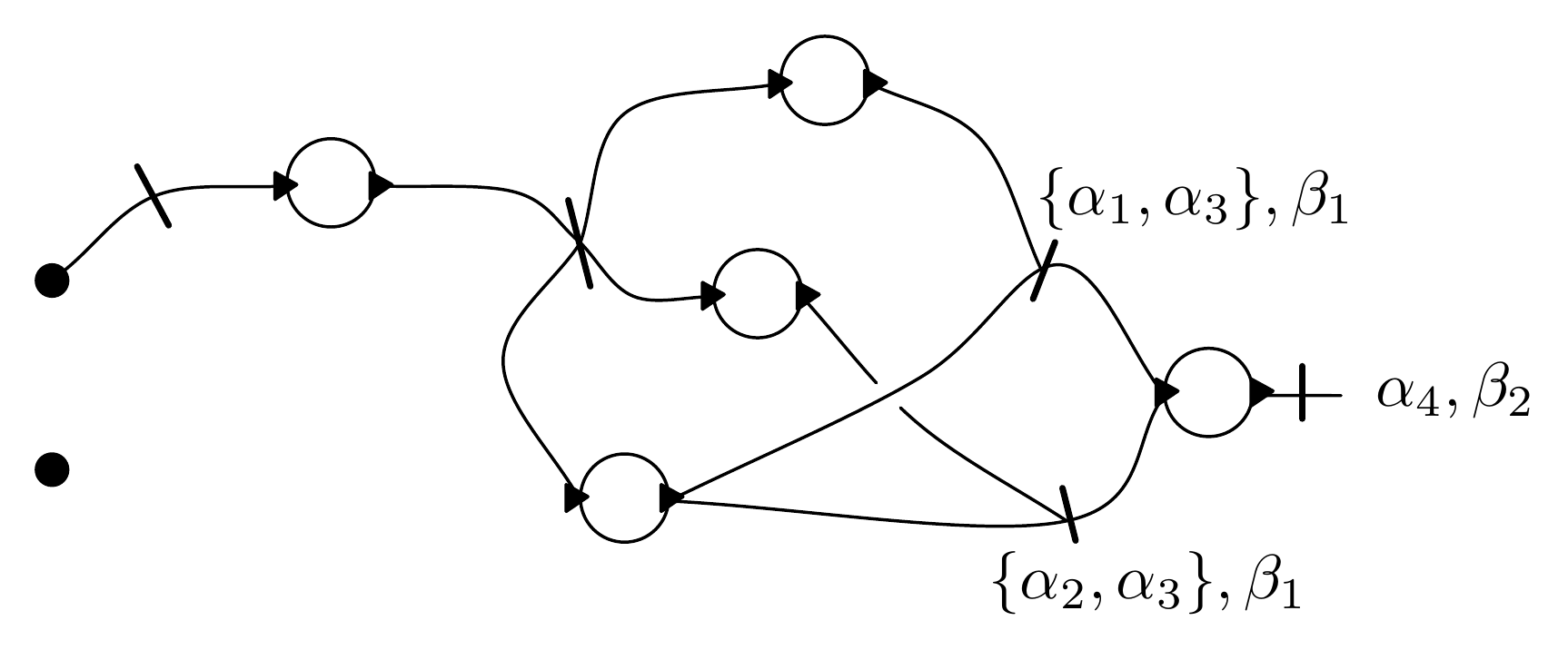}
\]
\caption{Composition of two C/E nets.
\label{fig:compositionBoundedNets}}
\end{figure}
An example of a composition of two C/E nets is illustrated in \fig{\ref{fig:compositionBoundedNets}}. 

\begin{rem} Two transitions in the composition of two C/E nets
may be in contention even though they are mutually independent in the underlying C/E net, as illustrated by   \fig{\ref{fig:compositionBoundedNets-contention}}. \qed
\end{rem}

\begin{figure}[t]
\subfigure{
\includegraphics[height=3cm]{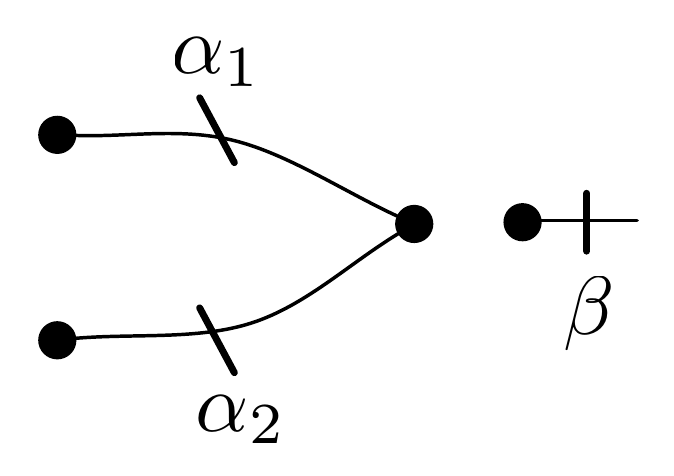}
}
\hspace{2cm}
\subfigure{\includegraphics[height=3cm]{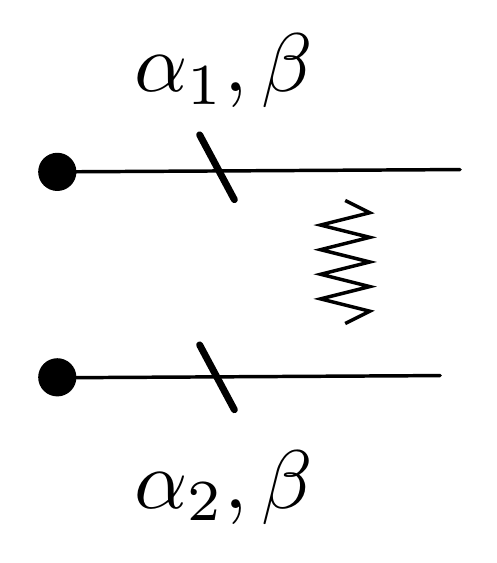}
}
\caption{Composition of two nets with boundaries. Note that $\alpha_1\#\alpha_2$ implicitly in the leftmost net, and $(\alpha_1,\beta)\#(\alpha_2, \beta)$ in the composition. This is emphasised graphically with the jagged line in the rightmost diagram.}
\label{fig:compositionBoundedNets-contention}
\end{figure}

\begin{rem}\label{rmk:ordinaryNets}
Any ordinary C/E net $N$ (Definition~\ref{defn:net}) can be considered as a net with boundaries $N:0\to 0$ as there is exactly one choice for functions $\source{-},\target{-}:T\to 2^{\underline{0}}$ and the contention relation consists of all pairs of transitions that are not independent in $N$. Composition of two nets $N:0\to 0$ and $M:0\to 0$ is then just the disjoint union of the two nets: the set of places is $P_N+P_M$, the minimal synchronisations are precisely $(\{t\},\varnothing)$, $t\in T_N$ and $(\varnothing,\{t'\})$, $t'\in T_M$, and the contention relation is the union of the contention relations of $N$ and $M$. \qed
\end{rem}

%********************
%TODO: Explain how a normal net can be seen as a net with boundaries while preserving 
%their semantics, and what is the interpretation of sequential composition. 
%%%%%%%%%%%%%%%%%%%%%%%%%%%%%%%%%%%%%%%%%%%%%%%%%%%%%%%%%%%%%%%%%%%%%%%%%%%%%%%%%%%%%%

\subsection{Labelled semantics of C/E nets with boundaries}
\label{sec:labelledSemantics}

%%%%%%%%%%%%%%%%%%%%%%%%%%%%%%%%%%%%%%%%%%%%%%%%%%%%%%%%%%%%%%%%%%%%%%%%%%%%%%%%%%%%%%
\begin{figure}[t]
\includegraphics[height=13cm]{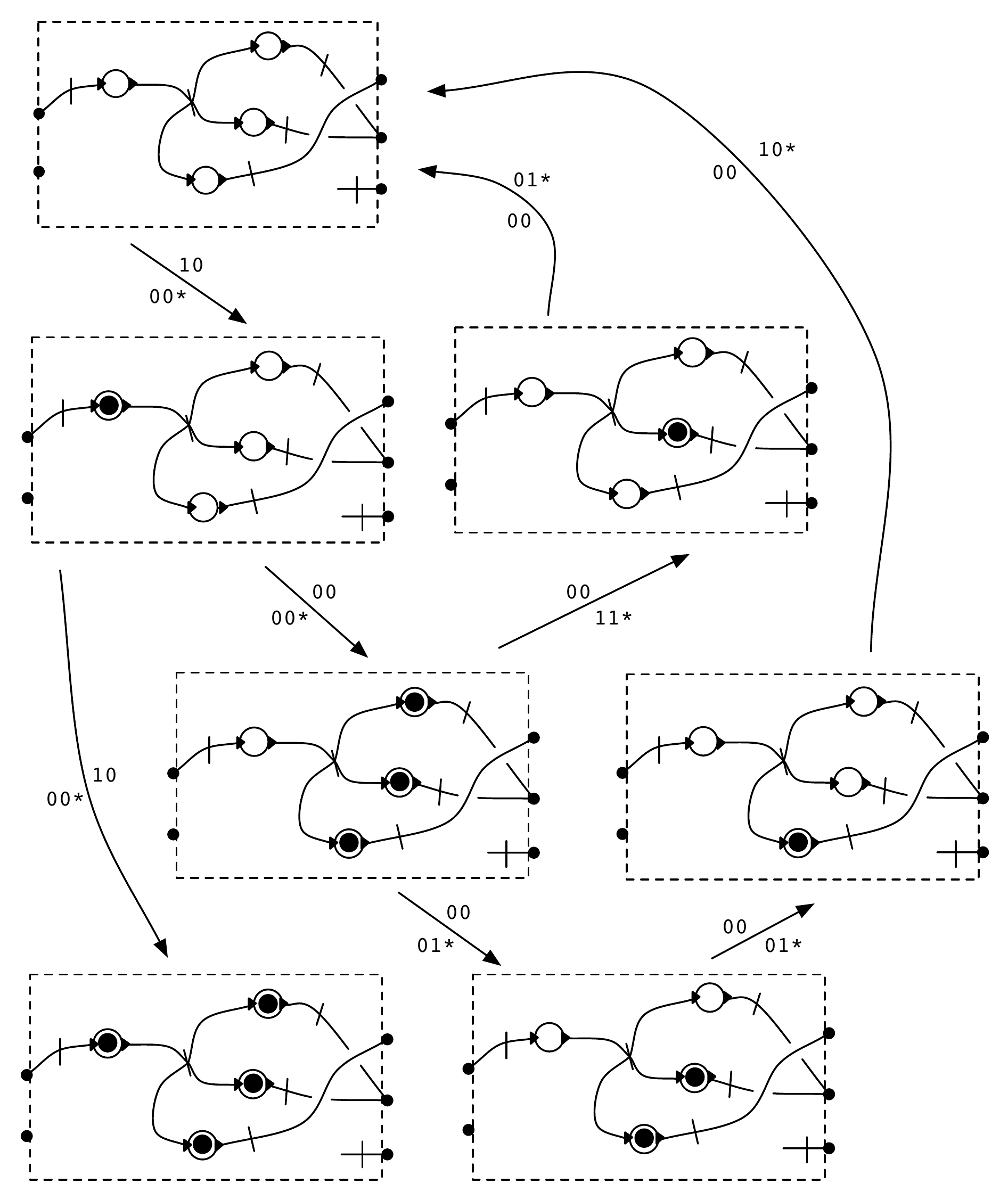}
\caption{Part of a labelled transition system for simple C/E net $2\to 3$. 
The symbol $*$ is used as shorthand for any label in $\{0,1\}$. \label{fig:ts}}
\end{figure}
%INDEPENDENCE - where should this paragraph be??

%Independence is an important
%property when composing nets, as discussed below, and in the definition of labelled
%semantics of nets. 
%A set of transitions $U\subseteq T$ is said to be \emph{mutually independent} 
%when $\forall t,u\in U$, $t\neq u$ implies $\neg (t\# u)$.

For any $k\in\N$, there is a bijection $\characteristic{-}: 2^{\underline{k}} \to \{0,1\}^k$ 
with
\[
\characteristic{U}_i \Defeq \begin{cases} 1 & \text{if }i\in U \\ 0 & \text{otherwise.} \end{cases}	
\] 
Similarly, with slight abuse of notation, we define
 $\characteristic{-}: \multiset{\underline{k}}\to \N^k$ by
\[
\characteristic{\mathcal{U}}_i \Defeq \mathcal{U}(i)
\]

% STRONG LABELLED SEMANTICS
\begin{defi}[C/E Net Labelled Semantics]\label{defn:stronglabels}
Let $N:m\to n$ be a C/E net with boundaries and $X, Y\subseteq P_N$. Write:
\begin{multline}%\label{strongNetTran}
\marking{N}{X} \dtrans{\alpha}{\beta} \marking{N}{Y} \quad\Defeq\quad 
\exists \text{ mutually independent } U\subseteq T_N \text{ s.t. } \\
\marking{N}{X} \rightarrow_U \marking{N}{Y},\  
\alpha = \characteristic{\source{U}} \ \&\ \beta = \characteristic{\target{U}}
\end{multline}
%\begin{multline}\label{weakNetTran}
%(N,\,X) \dTrans{\alpha}{\beta} (N,\,Y) \quad \Defeq \\
%\exists X',Y'.\ (N,\,X) (\dtrans{0^m}{0^n})^* (N,\,X'),\ 
%(N,\,X')\dtrans{\alpha}{\beta} (N,\,Y')\ \\
%\&\ (N,\,Y') (\dtrans{0^m}{0^n})^* (N,\,Y).
%\end{multline}
\end{defi}
It is worth emphasising that no information about precisely which set $U$ of transitions has been fired is carried by transition labels, merely the effect of the firing on the boundaries.
Notice that we always have $\marking{N}{X}\dtrans{0^m}{0^n}\marking{N}{X}$, as the empty
set of transitions is vacuously mutually independent.
%We shall refer to a transition of the form \eqref{strongNetTran} as a \emph{strong} transition
%and a transition of the form \eqref{weakNetTran} as a \emph{weak} transition.

A transition $\marking{N}{X} \dtrans{\alpha}{\beta}\marking{N}{Y}$ indicates  
that the C/E net $N$ evolves from marking $X$ to marking $Y$ by firing a set of transitions whose connections are recorded by $\alpha$ on the left interface and  $\beta$ on the right interface.
We give an example in Fig.~\ref{fig:ts}.

%WEAK LABELLED SEMANTICS
%\begin{defi}[Weak Labelled Semantics]
%Let $N:m\to n$ be a net and $X, Y\subseteq P_N$. Write:
%\begin{multline}\label{weakNetTran}
%\marking{N}{X} \dtransw{\alpha}{\beta} \marking{N}{Y} \quad\Defeq\quad 
%\exists\, \mathcal{U}\in \multiset{T_N} \text{ s.t. }
%\marking{N}{X} \Rightarrow_\mathcal{U} \marking{N}{Y},\  
%\alpha = \characteristic{\source{\mathcal{U}}} \ \&\ \beta = \characteristic{\target{\mathcal{U}}}
%\end{multline}
%\end{defi}

Labelled semantics is compatible with composition in the following sense.

\begin{thm}\label{thm:netdecomposition}
%\begin{enumerate}[(i)]
%\item 
Suppose that $M: k\to n$ and $N: n\to m$ are C/E nets with boundaries, and $X,X'\subseteq P_M$ and $Y, Y'\subseteq P_N$ markings.
Then $\marking{M;N}{X+Y} \dtrans{\alpha}{\beta} \marking{M;N}{X'+Y'}$
iff there exists $\gamma\in\{0,1\}^n$ such that
\[
 \marking{M}{X} \dtrans{\alpha}{\gamma} \marking{M}{X'}
 \text{ and }
 \marking{N}{Y} \dtrans{\gamma}{\beta} \marking{N}{Y'}.
\]
%\item Suppose that $M\from k\to n$ and $N\from n\to m$ are weak C/E nets with boundaries.
%Then
%$\marking{M;N}{X+Y} \dtransw{\alpha}{\beta} \marking{M;N}{X'+Y'}$ iff
%there exists $\gamma\in \N^n$ such that
%\[
%\marking{M}{X} \dtransw{\alpha}{\gamma} \marking{M}{X'} \text{ and }
%\marking{N}{Y} \dtransw{\gamma}{\beta}  \marking{N}{Y'}.	
%\]
%\end{enumerate}
\end{thm}
\begin{proof}
See Appendix~\ref{app:proof-netswithboundaries}.
\end{proof}

%The ``only-if'' component fails for strong semantics. 
%
%\begin{exa}\label{ex:NN0N1}\rm
%In particular, consider the following net $N$, expressed
%as a composition of two nets $N_0$ and $N_1$.
%\begin{equation}\label{eq:decompositionProblem}
%\lowerPic{1pc}{height=1cm}{problemRHS} = \lowerPic{1pc}{height=1cm}{problemLHS} 
%\end{equation}
%It is easy to see that $N \dtrans{11}{11} N$,
%but this transition does not decompose into transitions of $N_0$ and $N_1$.
%Instead, under the weak semantics, we have
%$N_0 \dtransw{11}{2} N_0$ and $N_1 \dtransw{2}{11} N_1$.
%\end{exa}
%

The above result is enough to show that bisimilarity is a congruence
with respect to the composition of nets over a common boundary. 
\begin{prop}\label{pro:netcongruence}
Bisimilarity of C/E nets is a congruence w.r.t. `\;$\comp$'.% and $\ten$. 
\end{prop}
\begin{proof}
See Appendix~\ref{app:proof-netswithboundaries}.
\end{proof}

\begin{rem}\label{rmk:loops}
Consider the composition of the three nets with boundaries below.
\[
\includegraphics[height=2cm]{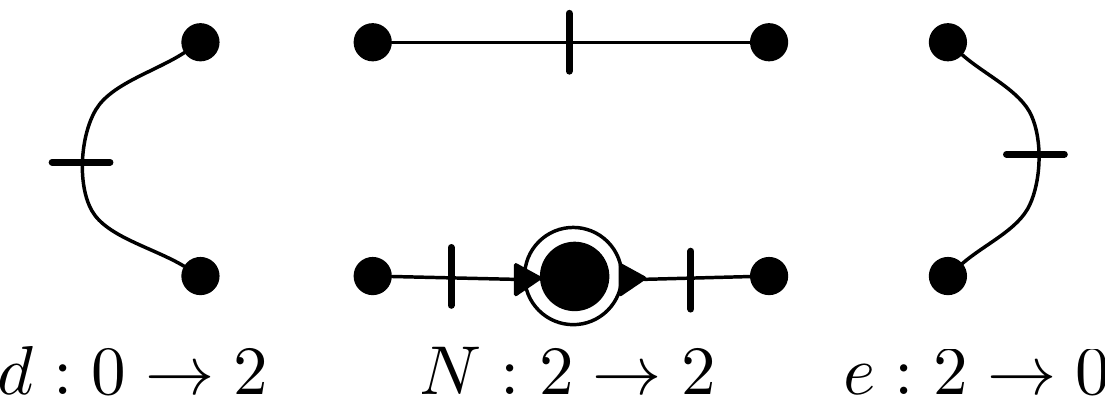}
\]
The result is a net with boundaries $0\to 0$ with a single place and a single consume/produce loop transition. As we have observed in Remark~\ref{rmk:consumeproduceloops}, this transition cannot fire with the semantics of nets that we have considered so far. Globally, the transition cannot fire because its postset is included in the original marking. The fact that the transition cannot fire is also reflected locally, in light of Theorem~\ref{thm:netdecomposition}: indeed, locally, for the transition to be able to fire, there would need to be a transition $\marking{N}{\{\star\}}\dtrans{11}{11}\marking{N}{\{\star\}}$, but this is not possible because there is a token present in the postset of the transition connected to the lower left boundary. It is possible to relax the semantics of nets in order to allow such transitions to fire, as we will explain in Remark~\ref{rmk:otherCEsemantics}. \qed
\end{rem}

\begin{rem}
In Remark~\ref{rmk:ordinaryNets} we noted that any ordinary net $N$ can be considered as a net with boundaries $N:0\to 0$. For such nets, the transition system of Definition~\ref{defn:stronglabels} has transitions with only one label (since there is nothing to observe on the boundaries) and thus corresponds to an unlabelled step-firing semantics transition system.
In particular, it follows that, while the transition systems generated for nets $N:0\to 0$ are different, they are all bisimilar; we feel that this is compatible with the usual view on labelled equivalences in that they capture behaviour that is observable from the outside: a net $N:0\to 0$ does not have a boundary and thus there is no way of interacting with it and therefore no way of telling apart two such nets. One can, of course, allow the possibility of observing the firing of certain transitions (possibly all) by connecting them to ports on the boundary. Let $N$ be a net with   $n= \#T_N$ transitions.  A corresponding  net  with boundaries that makes  transitions  observable over the right interface is as follows:  $N\from 0\to n$ with $\intsourcearg{t} = \emptyset$ for all $t\in T_N$,  $\inttargetarg{\_} \from T_N \to \ord{n}$ any  injective function, and the contention relation containing only those pairs of transitions that are in contention in the underlying C/E net  $N$.
\qed
\end{rem}

\section{{P/T} nets with boundaries}
\label{sec:ptboundaries}

This section extends the notion of nets with boundaries to 
P/T nets. The contention relation no longer plays a role, and 
connections of transitions to boundary ports are weighted.

\begin{defi}[{P/T} net with boundaries]\label{def:ptnetwithbound} 
Let $m, n \in \mathbb{N}$. A (finite, marked) \emph{P/T net with 
boundaries} $N_{\mathcal{X}}: m \rightarrow n$ is a tuple
  $N = (P, T, \pre{-}, \post{-}, \intsource, \inttarget)$ where:
  \begin{iteMize}{$-$}
  	\item $(P, T, \pre{-}, \post{-})$ is a finite P/T net;
  	\item $\intsource: T \rightarrow \mathcal{M}_{\ord{m}}$ and $\inttarget: T \rightarrow
  \mathcal{M}_{\ord{n}}$ are functions that assign transitions to the left and right boundaries of $N$;
   \item $\mathcal{X}\in\mathcal{M}_{P}$. 
  \end{iteMize}
   %Let $\mathcal{X} \in \mathcal{M}_{P}$, we write $N_{\mathcal{X}}$ for the (P/T net) with boundaries $N$ marked with  ${\mathcal{X}}$.
As in Definition~\ref{defn:boundedcenets} we assume that transitions have distinct footprints. 
\end{defi}

\begin{rem}\label{rmk:weakCENets}
For reasons that will become clear when we study the process algebraic account, 
we will sometimes refer to
P/T nets with boundaries that have markings which are \emph{subsets}  ($X\subseteq P$) 
of places instead of a multiset ($\mathcal{X}\in\mathcal{M}_{P}$) 
of places as \emph{weak C/E nets with boundaries}. \qed
\end{rem} 

The notion of net homomorphism extends to marked P/T nets with the same boundaries:
given $N_{\mathcal{X}},M_{\mathcal{Y}}: m\rightarrow n$,  $f:N_{\mathcal{X}}\rightarrow M_{\mathcal{Y}}$ is a pair of functions $f_T\from T_N\to T_M$, $f_P\from P_N\to P_M$
such that $f_P(\mathcal{X})=\mathcal{Y}$, 
$\pre{-}_N \seqComp 2^{f_P}=f_T \seqComp \pre{-}_M$,
$\post{-}_N\seqComp 2^{f_P}=f_T \seqComp \post{-}_M$, 
$\source{-}_N=f_T\seqComp \source{-}_M$
and $\target{-}_N=f_T\seqComp \target{-}_M$.
 A homomorphism is an isomorphism  if its two components are bijections. We write $N_{\mathcal{X}}\ptiso M_{\mathcal{Y}}$ if
 there is an isomorphism from $N_{\mathcal{X}}$ to $M_{\mathcal{Y}}$.

%WEAK SYNCHRONISATION
In order to compose P/T nets with boundaries we need to consider
a more general notion of synchronisation. This is because 
synchronisations of P/T involve multisets of transitions and 
there is no requirement of mutual independence.
The definitions of $\source{-}$ and $\target{-}$ extend for multisets in the obvious way
by letting $\source{\mathcal{U}}\Defeq \bigcup_{t\in T}\mathcal{U}(t) \cdot \source{t}$
and $\target{\mathcal{U}}\Defeq \bigcup_{t\in T} \mathcal{U}(t) \cdot \target{t}$.

\begin{defi}[Synchronisation of P/T nets]\label{defn:PTsynchronisation}
 A  synchronization between nets $M_{\mathcal{X}}:l\rightarrow m$ and $N_{\mathcal{Y}}:m\rightarrow n$ is a pair $({\mathcal{U}},{\mathcal{V}})$, with
${\mathcal{U}}\in\mathcal{M}_{T_M}$ and ${\mathcal{V}}\in\mathcal{M}_{T_N}$ 
multisets of transitions
such that:
\begin{iteMize}{$-$}
\item ${\mathcal{U}} + {\mathcal{V}}\neq\varnothing$;
\item $\target{{\mathcal{U}}}=\source{{\mathcal{V}}}$.
\end{iteMize}
The set of synchronisations inherits an ordering from the subset relation, 
\ie\ $({\mathcal{U}}',\,{\mathcal{V}}')\subseteq ({\mathcal{U}},\,{\mathcal{V}})$ when ${\mathcal{U}}'\subseteq {\mathcal{U}}$ and ${\mathcal{V}}'\subseteq {\mathcal{V}}$. A synchronisation
is said to be \emph{minimal} when it is minimal with respect to this order. 
\end{defi}

Let $Synch(M,N)$ denote the set of minimal synchronisations, an 
unordered set. This set
is always finite---this is an easy consequence of Dickson's 
Lemma~\cite[Lemma A]{Dickson1913}. 
%Any synchronisation can be written as a linear combination 
%of minimal synchronisations: this fact will be useful when we deal with
%labelled semantics of weak C/E nets.
\begin{lem}
The set of minimal synchronisations $Synch(M,N)$ is finite. 
\end{lem}
%\begin{proof}
%A synchronisation is a pair of multisets of transitions $(\mathcal{U},\mathcal{V})$
%and $Synch_w(M,N)$ is the set of minimal elements of the set of all weak 
%synchronisations. The proof proceeds by induction on the total
%set of transitions $T_M+T_N$ of $M$ and $N$. It is trivially true if the total
%set of transitions is empty. Otherwise fix any one transition $t$. 
%\end{proof}

The following result is comparable to Lemma~\ref{lem:strongTransitionDecomposition}
in the P/T net setting---any synchronisation can
be written as a linear combination of minimal synchronisations.
\begin{lem}\label{lem:weakTransitionDecomposition}
Suppose that $\marking{M}{\mathcal{X}}\from l\to m$ and $\marking{N}{\mathcal{Y}}\from m\to n$ are P/T nets with boundaries
and $(\mathcal{U},\mathcal{V})$ is a synchronisation.
Then
there exist a finite family $\{(b_i,(\mathcal{U}_i,\mathcal{V}_i))\}_{i \in I}$
where each $b_i\in\N_+$, $(\mathcal{U}_i,\mathcal{V}_i)\in Synch(M,N)$ where
for any $i,j\in I$, $(\mathcal{U}_i,\mathcal{V}_i)=(\mathcal{U}_j,\mathcal{V}_j)$
implies that $i=j$,
such that $\bigcup_{i\in I} b_i\cdot \mathcal{U}_i = \mathcal{U}$ 
and $\bigcup_{i\in I} b_i\cdot \mathcal{V}_i = \mathcal{V}$.
\end{lem}
\begin{proof} See Appendix~\ref{app:proof-ptboundaries}.
\end{proof}

Given $(\mathcal{U},\,\mathcal{V})\in Synch(M,N)$, let $\pre{(\mathcal{U},\,\mathcal{V})}\Defeq \pre{\mathcal{U}} + \pre{\mathcal{V}}\in\multiset{P_M+P_N}$, $\post{(\mathcal{U},\,\mathcal{V})}\Defeq\post{\mathcal{U}}+\post{\mathcal{V}}\in\multiset{P_M+P_N}$, $\source{(\mathcal{U},\,\mathcal{V})}\Defeq \source{\mathcal{U}}\in\multiset{\underline{l}}$ and 
$\target{(\mathcal{U},\,\mathcal{V})}\Defeq\target{\mathcal{V}}\in\multiset{\underline{n}}$.

%COMPOSING P/T Nets
\begin{defi}[Composition of P/T nets with boundaries]
If $\marking{M}{\mathcal{X}}\from l\to m$, $\marking{N}{\mathcal{Y}}\from m\to n$ are P/T nets with boundaries, 
define their composition, $\marking{M}{\mathcal{X}};\marking{N}{\mathcal{Y}}\from l\to n$, as follows:
\begin{iteMize}{$-$}
\item places are $P_M + P_N$;
\item transitions are obtained from set $Synch(M,N)$, after removing any redundant transitions with equal footprints (c.f.\  Definition~\ref{defn:compositionce});
%As before any transition in $M$ or $N$ not connected to the shared 
%boundary $m$ can be considered as a minimal synchronisation.
%\item $\forall (\mathcal{U},\,\mathcal{V})\in Synch(M,N)$, $\pre{(\mathcal{U},\,\mathcal{V})}\Defeq \pre{\mathcal{U}}\cup\pre{\mathcal{V}}$ and 
%$\post{(\mathcal{U},\,\mathcal{V})}\Defeq\post{\mathcal{U}}\cup\post{\mathcal{V}}$;
%\item $\forall (\mathcal{U},\,\mathcal{V})\in Synch(M,N)$, $\source{(\mathcal{U},\,\mathcal{V})}\Defeq \source{\mathcal{U}}$ and 
%$\target{(\mathcal{U},\,\mathcal{V})}\Defeq\target{\mathcal{V}}$;
\item the marking is $\mathcal{X}+\mathcal{Y}$.
\end{iteMize}
\end{defi}
\begin{figure*}[t]
\subfigure[Two {P/T} nets with boundaries $M$ (left) and $N$(right).]{
\hspace{.8cm}
\includegraphics[height=3.5cm]{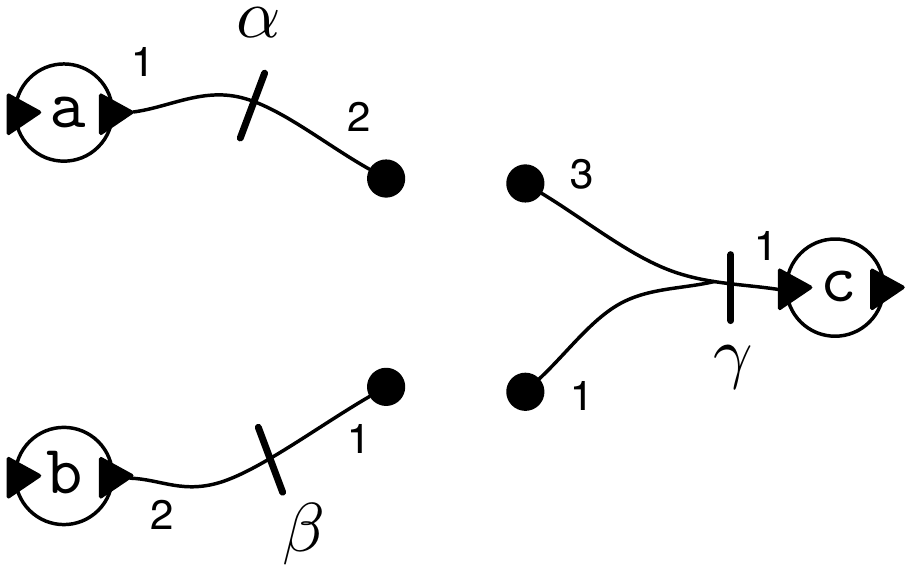}
\hspace{.3cm}
\protect\label{fig:ex-simple-pts}}
\subfigure[Composition $M;N$.]{
\includegraphics[height=3cm]{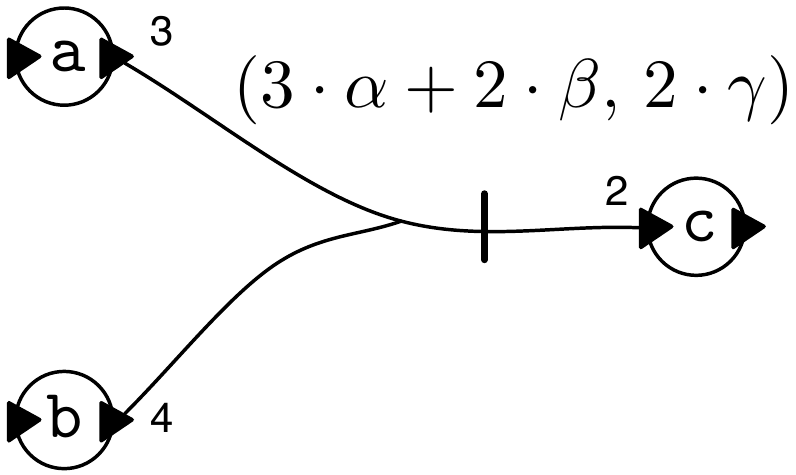}
\protect\label{fig:ex-composition-of-two-nets}}
\caption{Composition of {P/T} nets with boundaries.}
\protect\label{fig:composition}
\end{figure*}

\medskip
\begin{figure}[t]
\[
\includegraphics[height=4cm]{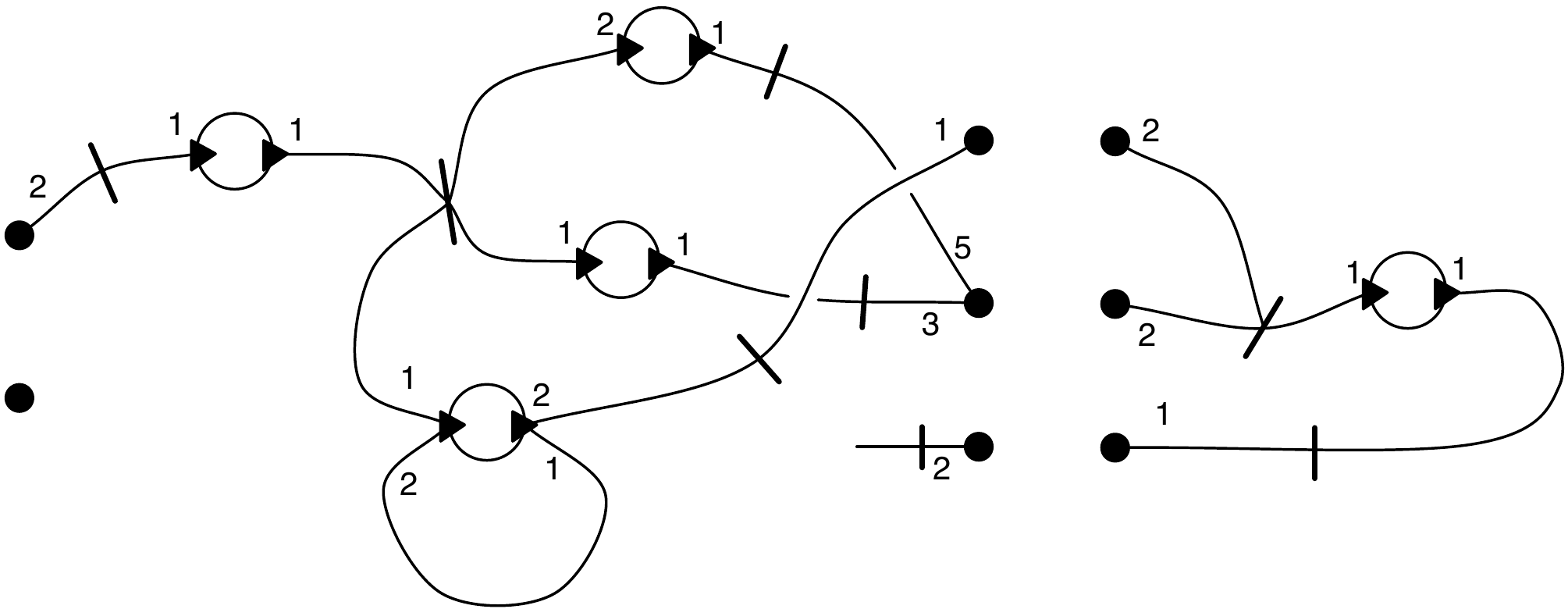}
\]
\[
\includegraphics[height=4cm]{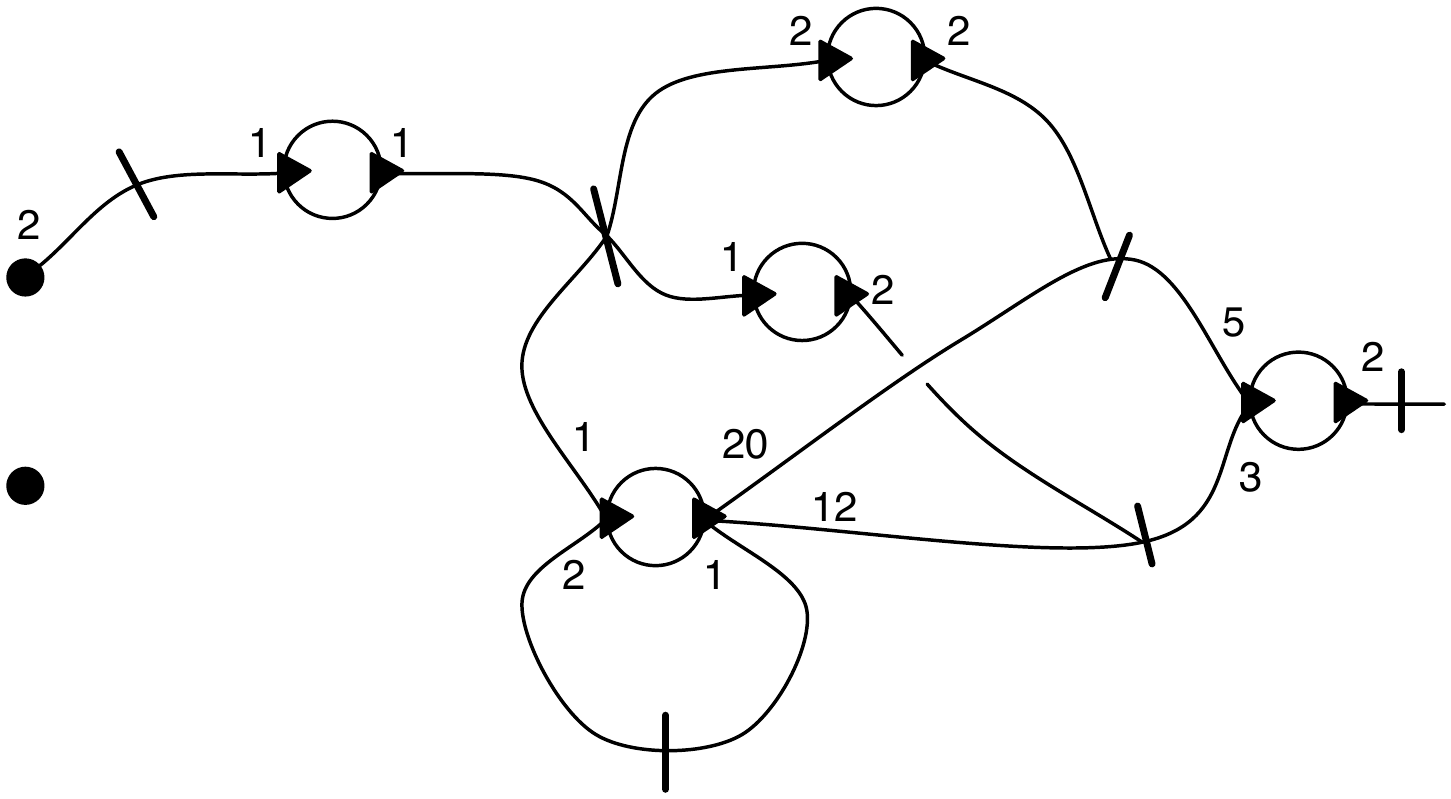}
\]
\caption{Composition of P/T nets with boundaries.
\label{fig:compositionWeakBoundedNets}}
\end{figure}
Figure~\ref{fig:ex-composition-of-two-nets} shows the sequential composition
 of the nets $M$ and $N$ depicted in Fig.~\ref{fig:ex-simple-pts}. The set of minimal 
 synchronization between $M$ and $N$ consists  just in the pair $(\mathcal{U},\mathcal{V})$ with 
  $\mathcal{U} = \{\alpha, \alpha, \alpha, \beta,\beta\}$ and $\mathcal{V} = \{\gamma, \gamma\}$. 
  In other words, to synchronise over the shared interface $M$  should fire transition  $\alpha$ three times (which consumes three 
  tokens from $\mathtt{a}$) and
   $\beta$ twice (which consumes four tokens from $\mathtt{b}$) and  $N$ should fire $\gamma$ twice
   (which produces two  tokens in $\mathtt{c}$). The equivalent net describing the synchronised composition of $M$ and $N$ over their common interface is a net that contains 
   exactly one transition, which consumes three tokens from $\mathtt{a}$, four tokens from $\mathtt{b}$ and 
   produces two tokens in $\mathtt{c}$, as illustrated in \fig{\ref{fig:ex-composition-of-two-nets}}.
 A more complex example of composition of P/T nets is 
given in \fig{\ref{fig:compositionWeakBoundedNets}}.

\subsection{Labelled semantics of P/T nets with boundaries}

We give two versions of labelled semantics, one corresponding to the standard semantics
and one to the banking semantics.
 
\begin{defi}[Strong Labelled Semantics]
\label{def:strong-semantics-pt-bound} Let $N:m\rightarrow n$ be a P/T net and $\mathcal{X},\mathcal{Y} \in\mathcal{M}_{P_N}$. We write
\begin{multline}
N_{\mathcal{X}} \dtrans{\alpha}{\beta} N_{\mathcal{Y}} \quad\Defeq\quad 
\exists {\mathcal{U}}\in\mathcal{M}_{T_N}\text{ s.t. } 
N_{\mathcal{X}}\rightarrow_{\mathcal{U}} N_{\mathcal{Y}},\  
\alpha = \characteristic{\source{\mathcal{U}}} \ \&\ \beta = \characteristic{\target{\mathcal{U}}}.
\end{multline}

\end{defi}

\begin{defi}[Weak Labelled Semantics]
\label{def:weak-semantics-pt-bound} Let $N:m\rightarrow n$ be a P/T net and $\mathcal{X},\mathcal{Y} \in\mathcal{M}_{P_N}$. We write
\begin{multline}
N_{\mathcal{X}} \dtransw{\alpha}{\beta} N_{\mathcal{Y}} \quad\Defeq\quad 
\exists {\mathcal{U}}\in\mathcal{M}_{T_N}\text{ s.t. } 
N_{\mathcal{X}}\Rightarrow_{\mathcal{U}} N_{\mathcal{Y}},\  
\alpha = \characteristic{\source{\mathcal{U}}} \ \&\ \beta = \characteristic{\target{\mathcal{U}}}.
\end{multline}
\end{defi}

%\begin{lem} 
%\label{lemma:seq-composition-boundaries}
%Let $M:l\rightarrow m$ and $N:m\rightarrow n$. $\cell{(M,m_{0})}{\alpha}{\gamma}{(M,m_0')}$ and 
%$\cell{(N,m_{1})}{\gamma}{\beta}{(N,m_{1}')}$ for some $\gamma$ iff $\cell{(M,m_{0});(N,m_{1})}{\alpha}{\beta}{(M,m_{0}');(N;m_{1}')}$.
%\end{lem}

\begin{thm}\label{thm:ptnetdecomposition}
Suppose that $M\from k\to n$ and $N\from n\to m$ are P/T nets with boundaries, and 
       $\mathcal{X}, \mathcal{X}'\in\mathcal{M}_{P_M}$ and  $\mathcal{Y}, \mathcal{Y}'\in\mathcal{M}_{P_N}$
       markings.
	Then
\begin{enumerate}[\em(i)]
       \item 
	$M;N_{\mathcal{X}+\mathcal{Y}} \dtrans{\alpha}{\beta} M;N_{\mathcal{X}'+\mathcal{Y}'}$ iff
	there exists $\gamma\in \N^n$ such that
	\[
	M_{\mathcal{X}}\dtrans{\alpha}{\gamma} M_{\mathcal{X}'} \text{ and }
	N_{\mathcal{Y}}\dtrans{\gamma}{\beta} N_{\mathcal{Y}'}.
	\]
      \item
	$M;N_{\mathcal{X}+\mathcal{Y}}\dtransw{\alpha}{\beta} M;N_{\mathcal{X}'+\mathcal{Y}'}$ iff
	there exists $\gamma\in \N^n$ such that
	\[
	M_{\mathcal{X}}\dtransw{\alpha}{\gamma} M_{\mathcal{X}'} \text{ and }
	N_{\mathcal{Y}}\dtransw{\gamma}{\beta} N_{\mathcal{Y}'}.
	\]	
   \end{enumerate}
\end{thm}
\begin{proof}  See Appendix~\ref{app:proof-ptboundaries}.
\end{proof}

\section{Properties of nets with boundaries}
\label{sec:properties}

%WELL DEFINEDNESS OF ;, ASSOCIATIVITY OF COMPOSITION, IDENTITIES
For each finite ordinal $m$ there is a C/E net 
$\mathit{id}_m:m\to m$ with no places and $m$ transitions, each connecting
the consecutive ports on the boundaries, i.e., for each transition $t_i$ with $0\leq i < m$,
$\source{t_i} = \target{t_i} = \{i\}$.
%components 
%$(\underline{m},\,\emptyset,\,!,\,!,\,\mathit{id}_{\underline{m}},\,\mathit{id}_{\underline{m}})$,
%where $!$ are the unique functions to the singleton set $2^\varnothing$. 
Similarly, there is a P/T net, 
which by abuse of notation we shall also refer to as $\mathit{id}_m:m\to m$.
\begin{prop} The following hold for both
C/E and P/T nets:
\begin{enumerate}[\em(i)]\label{pro:netscategories} 
\item Let $\marking{M}{X},\marking{M'}{X'}\from k\to n$ and 
$\marking{N}{Y},\marking{N'}{Y'}\from n\to m$ with 
$\marking{M}{X}\cong \marking{M'}{X'}$ and 
$\marking{N}{Y}\cong \marking{N'}{Y'}$. Then 
$\marking{M}{X}\seqComp \marking{N}{Y} \cong \marking{M'}{X'}\seqComp \marking{N'}{Y'}$.
\item Let $\marking{L}{W}\from k\to l$, $\marking{M}{X}\from l\to m$ and $\marking{N}{Y}\from m\to n$. Then
$(\marking{L}{W}\seqComp \marking{M}{X})\seqComp \marking{N}{Y} 
	\cong \marking{L}{W} \seqComp (\marking{M}{X} \seqComp \marking{N}{Y})$.
%\item Let $N\from m\to n$ be a net. Then ${(\mathit{id}_m\seqComp N)} \cong N\cong {(N\seqComp\mathit{id}_n)}$.
\item Let $\marking{M}{X}\from k \to n$. Then 
$\mathit{id}_k;\marking{M}{X}\cong \marking{M}{X}\cong \marking{M}{X};\mathit{id}_n$.
\end{enumerate}
\end{prop}

\begin{proof}
The proof are straightforward, exploiting (the composition of) isomorphisms to rename places and transitions.
\end{proof}

Nets taken up to isomorphism, therefore, form the arrows of 
a category with objects the natural numbers. 
Indeed, part (i) of Proposition~\ref{pro:netscategories} ensures
that composition is a well-defined operation on isomorphism equivalence classes of
nets, part (ii) shows that composition is associative and (iii) shows that
composition has identities.  
Let $\mathbf{CENet}$ and
$\mathbf{PTNet}$ denote the
categories with arrows the isomorphism classes of, respectively, 
C/E and P/T nets.

We need to define one other binary operation on nets. Given (C/E or P/T) nets 
$\marking{M}{\mathcal{X}}\from k\to l$ and $\marking{N}{\mathcal{Y}}\from m\to n$, their \emph{tensor product}
is, intuitively,  the net that results from putting the two nets side-by-side.
Concretely, $\marking{M\ten N}{\mathcal{X}+\mathcal{Y}}\from k+m\to l+n$ has:
\begin{iteMize}{$-$}
\item set of transitions $T_M+T_N$;
\item set of places $P_M+P_N$;
\item $\pre{-},\post{-}$ are defined in the obvious way;
\item $\source{-},\target{-}$ are defined by: 
\[
 \source{t} = \left\{
    \begin{array}{ll}
       \source{t}_M & \mbox{if}\ t\in T_M \\
       \{k + p\ |\ p \in \source{t}_N \} & \mbox{if}\ t\in T_N \\
%       k + \source{t}_N & \mbox{if}\ t\in T_N \\
      \end{array}\right.
      \quad
 \target{t} = \left\{
    \begin{array}{ll}
       \target{t}_M & \mbox{if}\ t\in T_M \\
       \{k + p\ |\ p \in \target{t}_N \} & \mbox{if}\ t\in T_N \\
%       l + \target{t}_N & \mbox{if}\ t\in T_N \\
      \end{array}\right.       
\] 
\end{iteMize}

\begin{prop}
The following hold for both C/E nets and P/T nets:
\begin{enumerate}[\em(i)]
\item Let $\marking{M}{\mathcal{X}},\marking{M'}{\mathcal{X}'}\from k\to n$ and 
$\marking{N}{\mathcal{Y}},\marking{N'}{\mathcal{Y}'}\from l\to m$ with 
$\marking{M}{\mathcal{X}}\cong \marking{M'}{\mathcal{X}'}$ and 
$\marking{N}{\mathcal{Y}}\cong \marking{N'}{\mathcal{Y}'}$. Then 
$\marking{M\ten N}{\mathcal{X}+\mathcal{Y}} \cong \marking{M'\ten N'}{\mathcal{X}'+\mathcal{Y}'}\from k+l\to n+m$.
\item $\mathit{id}_{m+n} \cong \mathit{id}_{m} \ten \mathit{id}_{n}$.
\item Let $\marking{M^1}{\mathcal{X}_1}\from m_1\to m_2$, $\marking{M^2}{\mathcal{X}_2}\from m_2\to m_3$,
$\marking{N^1}{\mathcal{Y}_1}\from n_1\to n_2$ and $\marking{N^2}{\mathcal{Y}_2}\from n_2\to n_3$. Then,
letting $\mathcal{Z}\Defeq \mathcal{X}_1+\mathcal{X}_2+\mathcal{Y}_1+\mathcal{Y}_2$ we have
$\marking{(M^1 \comp M^2)\ten (N^1\comp N^2)}{\mathcal{Z}} \cong 
\marking{(M^1 \ten N^1) \comp (M^2\ten N^2)}{\mathcal{Z}}$.
\end{enumerate}
\end{prop}
\begin{proof}
It follows straightforwardly along the proof of Proposition~\ref{sec:properties}.
\end{proof}

The above demonstrates that the categories $\mathbf{CENet}$ and $\mathbf{PTNet}$ 
are, in fact, monoidal.

\begin{prop}\label{pro:ptnetcongruence}
Bisimilarity of C/E nets is a congruence w.r.t. $\ten$. 
Bisimilarity of P/T nets is a congruence w.r.t. `\;$\comp$' and $\ten$. 
\end{prop}
\begin{proof}
The proof is analogous to that of Proposition~\ref{pro:netcongruence}.
\end{proof}

In particular, we obtain categories
$\mathbf{CENet}_{|\sim}$, $\mathbf{PTNet}_{|\sim}, \mathbf{PTNetProc}_{|\approx}$ 
with objects the natural numbers and arrows the 
bisimilarity equivalence classes
of, respectively C/E and P/T nets, the latter with either the
strong or the weak semantics.
Moreover, there are monoidal functors 
\[ [-]:\mathbf{CENet} \to \mathbf{CENet}_{|\sim} \] 
\[ [-]:\mathbf{PTNet} \to\mathbf{PTNet}_{|\sim} \]
\[ [-]_w:\mathbf{PTNet} \to\mathbf{PTNet}_{|\approx} \]
that are identity on objects and sends  
nets to their underlying equivalence classes.

\section{Petri calculus}
\label{sec:syntax}

%% INTRO TO PETRI CALCULUS
The Petri Calculus~\cite{DBLP:conf/concur/Sobocinski10} extends the calculus 
of stateless connectors~\cite{DBLP:journals/tcs/BruniLM06} with one-place buffers.
Here we recall its syntax, sorting rules and structural operational semantics. 
In addition to the rules presented in~\cite{DBLP:conf/concur/Sobocinski10} here
we additionally introduce a \emph{weak} semantics. The connection between
this semantics with some traditional weak semantics in process calculi is clarified in Remark~\ref{remark:weak}.

%Petri calculus with weak semantics enjoys a very close
%operational correspondence to a class of Petri nets with boundaries with
%weak semantics, introduced in ??.

% BNF OF PETRI CALCULUS
We give the BNF for the syntax of the Petri Calculus in~\eqref{eq:PetriCalculusSyntax}
below. The syntax features twelve constants 
\{\emptyplace,\,\tokenplace,\,\id,\,\tw,\,\diag,\codiag,
\,\rightEnd,\,\leftEnd,\,\ldiag,\,\lcodiag,\,\rzero,\,\lzero\},
to which we shall refer to as \emph{basic connectors},
and two binary operations $(\ten,\,\comp)$. 
Elements of the subset $\{\id,\,\tw,\,\diag,\codiag,
\,\rightEnd,\,\leftEnd,\,\ldiag,\,\lcodiag,\,\rzero,\,\lzero\}$ of basic connectors 
will sometimes be referred to as the \emph{stateless} basic connectors.
The syntax does not feature any operations with binding, 
primitives for recursion nor axiomatics for structural congruence.
\begin{equation}\label{eq:PetriCalculusSyntax}
P \bnfEq \emptyplace \bnfSep \tokenplace 
\bnfSep \id \bnfSep \tw
\bnfSep \diag \bnfSep \codiag 
\bnfSep \rightEnd \bnfSep \leftEnd 
\bnfSep \ldiag \bnfSep \lcodiag
\bnfSep \rzero \bnfSep \lzero \bnfSep P\ten P \bnfSep P \comp P
\end{equation}

%INTUITION about connectors

Constant  $ \emptyplace$ represents  an empty 1-place buffer while $\tokenplace$ denotes a full 1-place buffer.  The remaining basic connectors stands for the identity $\id$,  the symmetry $\tw$, synchronisations ($\diag$ and $\codiag$ ), mutual exclusive choices  ($\ldiag$ and $\lcodiag$), hiding  ($\rightEnd$  and $\leftEnd$) and inaction ($\rzero$ and $\lzero$). Complex connectors are obtained by composing simpler connector in parallel ($\ten$) or sequentially ($\comp$).

\begin{figure}[t]
\[
\reductionRule{}{ \typeJudgment{}{\emptyplace}{\sort{1}{1}} }\quad 
\reductionRule{}{ \typeJudgment{}{\tokenplace}{\sort{1}{1}} }\quad
\reductionRule{}{ \typeJudgment{}{\id}{\sort{1}{1}} }\quad 
\reductionRule{}{ \typeJudgment{}{\tw}{\sort{2}{2}} }\quad 
\reductionRule{}{ \typeJudgment{}{\diag}{\sort{1}{2}} }\quad 
\reductionRule{}{ \typeJudgment{}{\codiag}{\sort{2}{1}} }\quad 
\reductionRule{}{ \typeJudgment{}{\rightEnd}{\sort{1}{0}} }\quad 
\reductionRule{}{ \typeJudgment{}{\leftEnd}{\sort{0}{1}} }
\]
\[
\reductionRule{}{ \typeJudgment{}{\ldiag}{\sort{1}{2}} }\quad 
\reductionRule{}{ \typeJudgment{}{\lcodiag}{\sort{2}{1}} }\quad 
\reductionRule{}{ \typeJudgment{}{\rzero}{\sort{1}{0}} }\quad 
\reductionRule{}{ \typeJudgment{}{\lzero}{\sort{0}{1}} }\quad 
\reductionRule{ \typeJudgment{}{P}{\sort{k}{l}} \quad \typeJudgment{}{R}{\sort{m}{n}} }
{ \typeJudgment{}{P\ten R}{\sort{k+m}{l+n}} } \quad
\reductionRule{ \typeJudgment{}{P}{\sort{k}{n}} \quad \typeJudgment{}{R}{\sort{n}{l}} }
{ \typeJudgment{}{P\comp R}{\sort{k}{l}} }
\]
\caption{Sort inference rules.\label{fig:sortInferenceRules}}
\end{figure}

% SORTS
The syntax is augmented with a simple discipline of sorts.
The intuitive idea is that a well-formed term of the Petri calculus 
describes a kind of black box with a number of ordered wires on the left
and the right. Then, following this intuition, the operation $\comp$ 
connects such boxes by connecting wires on a shared boundary, and the 
operation $\ten$ places two boxes on top of each other.
A sort indicates the number of wiring ports of a term, it is thus 
a pair $\sort{k}{l}$, where $k,l\in\N$. The syntax-directed sort inference rules 
are given in~\fig{\ref{fig:sortInferenceRules}}.
Due to their simplicity, a trivial induction confirms 
uniqueness of sorting:
%\begin{prop}
%Let $P$ be a term generated from~\eqref{eq:PetriCalculusSyntax}. 
 if $\typeJudgment{}{P}{\sort{k}{l}}$ and $\typeJudgment{}{P}{\sort{k'}{l'}}$
then $k=k'$ and $l=l'$. 
%\qed
%\end{prop}

% TERMS WITHOUT SORT
As evident from the rules in~\fig{\ref{fig:sortInferenceRules}},
a term generated from~\eqref{eq:PetriCalculusSyntax} fails to have a sort
iff it contains a subterm of the form $P\comp R$ with 
$\typeJudgment{}{P}{\sort{k}{l}}$ and
$\typeJudgment{}{R}{\sort{m}{n}}$ such that $l\neq m$. Coming back to our
intuition, this means that $P\comp R$ refers to a system in which
box $P$ is connected to box $R$, yet they do not have a compatible common boundary;
we consider such an operation undefined and we shall not consider it further. Consequently
in the remainder of the paper we shall only consider those terms that have a sort.

% OPERATIONAL SEMANTICS
\begin{figure}[t]
\[
\derivationRule{}{\emptyplace \dtrans{1}{0} \tokenplace}{TkI} \quad
\derivationRule{}{\tokenplace \dtrans{0}{1} \emptyplace}{TkO} 
\]
\[
\derivationRule{}{\id  \dtrans{1}{1} \id }{Id} \quad
\derivationRule{}{\tw \dtrans{ab}{ba} \tw}{Tw}  \quad
\derivationRule{}{\rightEnd \dtrans{1}{} \rightEnd}{$\rightEnd$} \quad
\derivationRule{}{\leftEnd \dtrans{}{1} \leftEnd}{$\leftEnd$}
\]
\[
\derivationRule{}{\diag \dtrans{1}{1 \labelSep 1} \diag}{$\diag$} \quad
\derivationRule{}{\codiag \dtrans{1 \labelSep 1}{1} \codiag}{$\codiag$} \quad 
\derivationRule{}{\ldiag \dtrans{\ 1\ }{ (1-a)\labelSep a} \ldiag}{$\ldiag a$} \quad
\derivationRule{}{\lcodiag \dtrans{\!\!(1-a)\labelSep a\!\!}{1} \lcodiag}{$\lcodiag a$}\quad
\derivationRule{\mathsf{C}\typ\sort{k}{l}\mbox{ \footnotesize a basic connector}}{\mathsf{C}\dtrans{0^k}{0^l} \mathsf{C}}{Refl}\quad
\]
\[
\derivationRule{P\dtrans{\alpha}{\gamma} Q \quad R\dtrans{\gamma}{\beta} S}
{P\comp R \dtrans {\alpha}{\beta} Q\comp S}{Cut} \quad
\derivationRule{P\dtrans{\alpha_1}{\beta_1} Q\quad R\dtrans{\alpha_2}{\beta_2} S} 
{P\ten R \dtrans{\alpha_1 \labelSep \alpha_2}{\beta_1 \labelSep \beta_2} Q\ten S}{Ten}\qquad\qquad
\derivationRule{P\dtrans{\alpha_1}{\beta_1} R\qquad R\dtrans{\alpha_2}{\beta_2}Q }
{P\dtrans{\alpha_1+\alpha_2}{\beta_1+\beta_2} Q}{Weak*}
\]
\caption{Structural rules for operational semantics\label{fig:operationalSemantics}.
Assume that $a,b\in\{0,1\}$ and $\alpha,\beta,\gamma\in \{0,1\}^*$ (strong variant)
and $\alpha,\beta,\gamma\in\N^*$ (weak variant).}
\end{figure}

%INTRO TO OPERATIONAL SEMANTICS OF PETRI CALCULUS
The structural inference rules for the operational semantics 
of the Petri Calculus are given in \fig{\ref{fig:operationalSemantics}}. 
Actually, 
two variants of the operational semantics are considered, to which we 
shall refer to as the strong and weak operational semantics. 
%Both the semantics feature 
%two-labelled transitions, one label written above a transition and one below. 
%
%STRONG SEMANTICS
The strong variant is obtained by considering all the rules in 
\fig{\ref{fig:operationalSemantics}} apart from the rule \ruleLabel{Weak*}.

The labels on transitions in the strong variant are pairs of binary vectors; i.e., $P\dtrans{\alpha}{\beta}Q$ 
with $\alpha,\beta\in \{0,1\}^*$. The transition $P\dtrans{\alpha}{\beta} Q$ describes  the 
evolution of  $P$ that exhibits the behavior 
$\alpha$ over its left boundary and $\beta$ over its right boundary.
It is easy
to check that whenever $P\typ\sort{n}{m}$ and $P\dtrans{\alpha}{\beta}Q$
then $\alpha\in \{0,1\}^n$, $\beta\in \{0,1\}^m$ and $Q\typ\sort{n}{m}$.   Intuitively, $\alpha$ and 
$\beta$ describe the observation on each wire of the boundaries. 

For instance, \ruleLabel{TkI} states that the empty place $\emptyplace$ becomes a full place $\tokenplace$ when one token is received over its left boundary and no token is produced  on its right boundary.  Rule \ruleLabel{TkO} describes the transition of a full place that becomes an empty place and releases a token over its right boundary. Rule  \ruleLabel{Id} states that connector $\id$ replicates the same observation on its two boundaries. Rule $ \ruleLabel{Tw}$  shows that the connector $\tw$ exchanges the order of the wires on its two interfaces. Rules  \ruleLabel{$\leftEnd$}  and \ruleLabel{$\rightEnd$} say  that both $\leftEnd$ and $\rightEnd$ hide to one of its boundaries the observation that  takes  over the other. By rule \ruleLabel{$\diag$}, the connector $\diag$ duplicates the observation on its left wire to the two wires on its right boundary.
Each of the rules \ruleLabel{$\ldiag a$} and \ruleLabel{$\lcodiag a$} actually 
represent two rules, one for $a=0$ and one for $a=1$. The rule \ruleLabel{Refl} 
guarantees that any basic connector (and, therefore, any term) 
is always capable of ``doing nothing''; we will refer to transitions
in which the labels consist only of $0$s as \emph{trivial} behaviour. 
\ruleLabel{Refl} is the only rule that applies to basic connectors $\rzero$ and $\lzero$,
which consequently only exhibit trivial behaviour. Rule \ruleLabel{Cut} states that 
two connectors composed sequentially can compute if the observations over their shared interfaces 
coincide. Differently, components composed in parallel can evolve independently (as defined by rule \ruleLabel{Ten}.

%WEAK SEMANTICS
The weak variant is obtained by additionally allowing the unrestricted
use of rule \ruleLabel{Weak*} in any derivation of a transition.
This rule deserves further explanation: the addition operation that features
in \ruleLabel{Weak*} is simply point-wise addition of vectors of natural numbers (as 
defined in Section~\ref{sec:background});
the labels in weak transitions will thus, in general, be natural number vectors
instead of mere binary vectors.
In order to distinguish the two variants we shall write weak transitions
with a thick transition arrow: $P\dtransw{\alpha}{\beta}Q$.
Analogously to the strong variant, if $P\typ\sort{n}{m}$ and $P\dtransw{\alpha}{\beta}Q$
then $\alpha\in \N^n$, $\beta\in \N^m$ and $Q\typ\sort{n}{m}$.

\begin{exa}
Let $P \bydef \emptyplace\comp \diag$ and $Q \bydef \tokenplace\comp \diag$. It is easy to check that $P \typ\sort{1}{2}$ and $Q \typ\sort{1}{2}$. The unique  non trivial behavior of $P$ under the strong semantics  is $P \dtrans{1}{0\labelSep 0} Q$ and can be derived as follows 
\[
\derivationRule{
  \derivationRule{}{\emptyplace \dtrans{1}{0} \tokenplace}{TkI} \quad 
  \derivationRule{}{\diag \dtrans{0}{0 \labelSep 0} \diag}{Refl} 
  }
  {
  \emptyplace\comp \diag \dtrans{1}{0\labelSep 0} \tokenplace\comp\diag
  }
  {Cut}
\]
 We can also show that the non-trivial behaviours of $Q$ are $Q \dtrans{0}{1\labelSep 1} P$ and 
$Q  \dtrans{1}{1\labelSep 1} Q$.  By using rule \ruleLabel{Weak*} with the premises  $P \dtrans{1}{0\labelSep 0} Q$ and $Q \dtrans{0}{1\labelSep 1} P$, we can obtain $P \dtransw{1}{1\labelSep 1} P$. This weak  transition denotes a computation in which a token received over the left interface is immediately available on the right interface. This kind of behavior is not derivable when considering the strong semantics.  Finally, note that we can build the following derivation  

\[
\derivationRule{
 {P \dtransw{1}{1\labelSep 1} P}{} \quad 
 {P \dtransw{1}{1\labelSep 1} P} {} 
  }
  {
 P \dtransw{2}{2\labelSep 2} P  
 }
  {Weak*}
\]
and, in general, for any $n$ we can build $P \dtransw{n}{n\labelSep n} P$, i.e., a transition in which $P$ receives $n$ tokens over the wire on its left boundary and sends $n$ tokens over each wire on its right boundary. 
\end{exa}
 
\begin{rem}\label{remark:weak}
There is a strong analogy between the weak semantics of the Petri Calculus 
and the
weak semantics of traditional process calculi, say CCS. Given the standard LTS of 
CCS, one can generate in an obvious way
 a new LTS with the same states but in which the actions are labelled with traces of non-$\tau$ CCS actions, 
 where any $\tau$-action of the original LTS
 is considered to be an empty trace in the new LTS---i.e. the identity for the free monoid
of non-$\tau$ actions. Bisimilarity
on this LTS corresponds to weak bisimilarity, in the sense of 
Milner, on the original LTS.

On the other hand, the labels of the strong version of the Petri calculus
are pairs of strings of $0$s and $1$. A useful intuition is that $0$ means ``absence of signal''
while $1$ means ``presence of signal.'' The free monoid on this set, taking $0$ to be identity
is nothing but the natural numbers with addition---in this sense the rule \ruleLabel{Weak*}
generates a labelled transition system that is analogous to the aforementioned
 ``weak'' labelled transition system for CCS. See~\cite{Sobocinski2012} for further details.
\qed
\end{rem}

% The "Other C/E" semantics remark
\begin{rem}\label{rmk:otherCEsemantics}
Consider the additional rules below, not included in the set of
operational rules for the Petri calculus in Fig.~\ref{fig:operationalSemantics}.
\[ 
\derivationRule{}{\emptyplace \dtrans{1}{1} \emptyplace}{TkI2} \qquad
\derivationRule{}{\tokenplace \dtrans{1}{1} \tokenplace}{TkO2} 
\]
 Recall that the semantics of C/E nets, given in Definition~\ref{defn:CENetStrongSemantics} is as follows:
\[
N_X \rightarrow_{U} N_Y \quad\Defeq\quad
\pre{U} \subseteq X,\ 
\post{U} \cap X = \varnothing\ \&\ 
Y = (X \backslash \pre{U}) \cup \post{U}.
\]
where $U$ is a set of mutually independent transitions.

 Including the rule~\ruleLabel{TkI2} would allow an empty place to receive a token, and simultaneously release it, in one operation. Similarly, rule \ruleLabel{TkO2} allows computations in which a marked place simultaneously receives and releases a token.

 While we will not give all the details here, the system with \ruleLabel{TkI2} would correspond to the semantics where, for $U$ a set of mutually independent transitions:
 \[
 \marking{N}{X} \rightarrow_{U} \marking{N}{Y} 
 \quad \Defeq \quad 
 \post {U} \cap X  = \varnothing,\, 
 \pre{U} \cap Y = \varnothing\ \&\ 
 X \cup \post{U} = Y \cup \pre{U}.
 \]
 Using this semantics,
 in the example below, the two transitions can fire simultaneously to move from the marking illustrated on the left to the marking illustrated on the right.
  \[
  \includegraphics[height=.75cm]{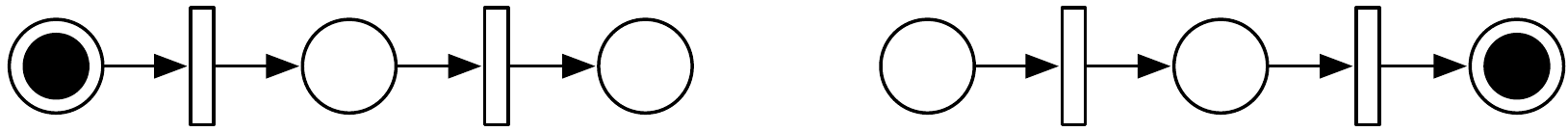}
  \]
 Note that this definition also allows intuitively less correct behaviour, in particular, a transition that has an \emph{unmarked} place in both its pre and post sets is able to fire, assuming that it is otherwise enabled.
 
 Instead, including the rule~\ruleLabel{TKO2} allows a marked place to receive a token and simultaneously release it, in one operation. Here, we would need to change the underlying semantics of nets to:
 \[
 N_X \rightarrow_{U} N_Y \quad\Defeq\quad
 \pre{U} \subseteq X,\,\ \post{U}\subseteq Y\ \&\  
 X\backslash \pre{U} =
 Y\backslash \post{U}.
 \]
 This was the semantics of nets originally considered in~\cite{DBLP:conf/concur/Sobocinski10}. For example, in the net below, the two transitions can again fire independently. 
 \[
   \includegraphics[height=.75cm]{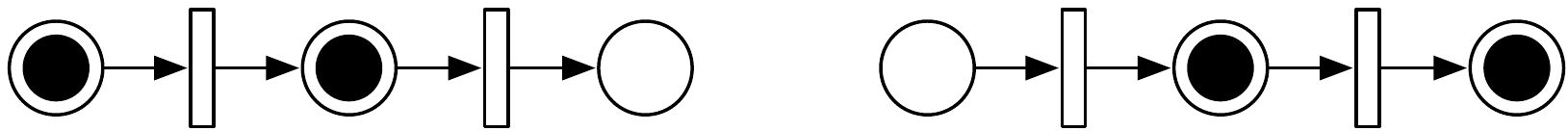}
 \]
 Note that this semantics allows transitions that intersect non-trivially in their pre and post sets to fire (see Remarks~\ref{rmk:consumeproduceloops} and~\ref{rmk:loops}).

 Including both rules~\ruleLabel{TKI2} and~\ruleLabel{TKO2} allows both of the behaviours described above, with the underlying net semantics:
  \[
 N_X \rightarrow_{U} N_Y \quad\Defeq\quad
 Y + \pre{U} = X + \post{U}
 \]
 where $X$, $Y$, $U$ are sets but the operations are those of multisets.
 For example, in the net below left, all the transitions can fire together to produce the marking on the right.
 \[
   \includegraphics[height=.75cm]{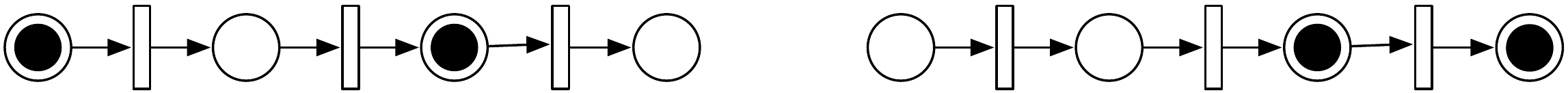}
 \]
 The full weak semantics (Definition~\ref{defn:PTWeakSemantics}) that we will consider here corresponds to considering the unrestricted use of the rule \ruleLabel{Weak*} in the Petri calculus. This semantics is even more permissive: we do not keep track of independence of transitions and allow the firing of \emph{multisets} of transitions. Notice that \ruleLabel{Weak*} subsumes the rules \ruleLabel{TkI2} and \ruleLabel{TkO2} discussed above, in the sense that they can be derived from \ruleLabel{TkI}, \ruleLabel{TkO} and \ruleLabel{Weak*}. An example computation is illustrated below.
  \[
   \includegraphics[height=1.75cm]{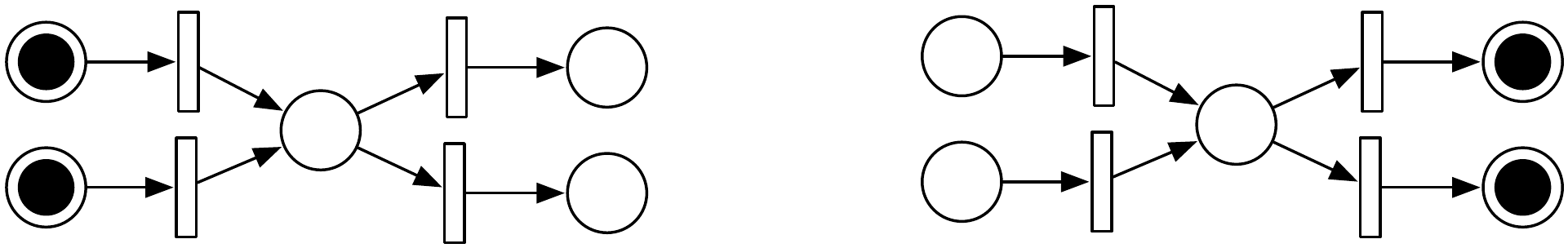}
 \]
 Here the four transitions can fire together. \qed
\end{rem}

We let $\subterms{P}$ denote the set of basic subterms of $P$, inductively defined as
\[
\subterms{P}  = 
\left\{\begin{array}{ll}
 \subterms{Q}\cup \subterms{R} & \quad\mbox{if $P = Q \ten R$ or $P = Q \comp R$}\\
\{\, P\, \} & \quad\mbox{otherwise}
\end{array}\right.
\]
A term $P$ is said to be
\emph{stateless} when
 $\sigma(P) \cap \{\,\emptyplace,\,\tokenplace\,\} = \emptyset$.
The next result  follows by  trivial induction on derivations.
\begin{lem}\label{lem:stateless}
Let $P:(k,l)$ be a stateless term. Then, for any $\alpha,\beta,Q$ such that 
$P\dtrans{\alpha}{\beta} Q$ or
$P\dtransw{\alpha}{\beta} Q$  
we have $P=Q$. \qed
\end{lem}

%BASIC CONNECTORS WEAK CHARACTERISATION
It is useful to characterise the behaviour of the basic connectors under the weak 
semantics. The proofs of the following are straightforward. 
\begin{prop}\label{pro:weakConstants} 
In the following let $a,b,c,d\in\N$.
\begin{enumerate}[\em(i)]
\item $\emptyplace \dtransw{a}{b} P$ iff $P=\emptyplace$ and 
	$a=b$, or $P=\tokenplace$ and $a=b+1$. \vspace{.3em}
\item $\tokenplace \dtransw{a}{b} P$ iff $P=\tokenplace$ and 
    $a=b$, or $P=\emptyplace$ and $b=a+1$. \vspace{.3em}
\item $\id \dtransw{a}{b} \id$ iff $a=b$. \vspace{.3em}
\item $\tw \dtransw{ab}{cd} \tw$ iff $a=d,b=c$.\vspace{.3em}
\item $\diag \dtransw{a}{bc} \diag$ iff $a=b=c$.\vspace{.3em}
\item $\codiag \dtransw{ab} c\codiag$ iff $a=b=c$.\vspace{.3em}
\item $\rightEnd \dtransw{a}{} \rightEnd$.\vspace{.3em}
\item $\leftEnd \dtransw{}{a} \leftEnd$.\vspace{.3em}
\item $\ldiag \dtransw{a}{bc} \lcodiag$ iff $a=b+c$.\vspace{.3em}
\item $\lcodiag \dtransw{ab}{c} \lcodiag$ iff $c=a+b$.\vspace{.3em}
\item $\rzero \dtransw{a}{} \rzero$ iff  $a=0$.\vspace{.3em}
\item $\lzero \dtransw{}{a} \lzero$ iff  $a=0$.
\end{enumerate}
\end{prop}

\medskip
%WEAK SYNTAX DECOMPOSITION LEMMA
The following useful technical lemma shows that, in any derivation of a weak
transition for a composite term of the form $P\comp Q$ or $P\ten Q$ one can
assume without loss of generality that the last rule applied was, respectively,
\ruleLabel{Cut} and \ruleLabel{Ten}.
\begin{lem}\label{lem:syntaxdecomposition}~
\begin{enumerate}[\em(i)]
\item
If $P\comp R \dtransw{\alpha}{\beta} Q$ then there exist 
	$P'$, $R'$, $\gamma$ such that $Q=P'\comp R'$,
	 $P\dtransw{\alpha}{\gamma} P'$ and $R\dtransw{\gamma}{\beta} R'$.
\item
If $P\ten R\dtransw{\alpha}{\beta} Q$ then there exist 
	$P'$, $R'$ such that $Q=P'\ten R'$,
		$P\dtransw{\alpha_1}{\beta_1} P'$, $R\dtransw{\alpha_2}{\beta_2} R'$ 
		  with $\alpha=\alpha_1\alpha_2$ and $\beta=\beta_1\beta_2$.
%\item FALSE!
%If $P\dtransw{a}{b}P'$ then there exist $P_1,\dots,P_n$ with
%$P\dtrans{a_0}{b_0}P_1\dtrans{a_1}{b_1}P_2\cdots P_n\dtrans{a_n}{b_n} P'$
%such that $a=\sum_{0\leq i\leq n}{a_i}$ and $b=\sum_{0\leq i\leq n}{b_i}$.
\end{enumerate}
\end{lem}
\begin{proof}
(i) If the last rule used in the derivation was \ruleLabel{Cut} then we are finished.
By examination of the rules in \fig{\ref{fig:operationalSemantics}} the only other
possibility is \ruleLabel{Weak*}. We can collapse all the instances of 
\ruleLabel{Weak*} at the root of the derivation into a subderivation tree of the form:
\begin{equation}\label{eq:lastStep}
\prooftree
P \comp R \dtransw{\alpha_0}{\beta_0} Q_1\quad Q_1 
\dtransw{\alpha_1}{\beta_1} Q_2 \quad \cdots\quad  Q_n \dtransw{\alpha_n}{\beta_n} Q 
\justifies
P \comp R \dtransw{\alpha}{\beta} Q
\endprooftree
\end{equation}
where $\alpha=\sum_{i} \alpha_i$ and $\beta=\sum_{i} \beta_i$. 
We now proceed by induction on $n$.
The last rule in the derivation of 
$P \comp R\dtransw{\alpha_0}{\beta_0} Q_1$ must have been $\ruleLabel{Cut}$,
whence we obtain some 
$\gamma_0$, $P_1$ and $R_1$ such that $P\dtransw{\alpha_0}{\gamma_0}P_1$, 
$R\dtransw{\gamma_0}{\beta_0}R_1$ and
$Q_1=P_1\comp R_1$. If $n=1$ then 
$\alpha=\alpha_0$, $\beta=\beta_0$, $Q_1=Q$ and we are finished. Otherwise,
let $\alpha'=\sum_{1\leq i\leq n}\alpha_i$, $\beta'=\sum_{1\leq i\leq n}\beta_i$,
we have $P_1\comp R_1 \dtransw{\alpha'}{\beta'} Q$ and by the 
 inductive hypothesis, there exists $\gamma'$ such that
$P_1 \dtransw{\alpha'}{\gamma'} P'$, $R_1 \dtransw{\gamma'}{\beta'} R'$ with
and $Q=P'\comp R'$. We can now apply \ruleLabel{Weak*} twice to obtain $P \dtransw{\alpha}{\gamma_0+\gamma'} P'$ and
$R\dtransw{\gamma_0+\gamma'}{\beta} R'$.

The proof of (ii) is similar.
\end{proof}

%CONGRUENCE OF STRONG AND WEAK BISIMILARITY
We shall denote bisimilarity on the strong semantics by $\sim$ and bisimilarity on the weak semantics by
$\approx$. Both equivalence relations are congruences.
This fact is important, because it allows us to replace subterms with bisimilar ones without
affecting the behaviour of the overall term.
\begin{prop}[Congruence]\label{pro:petricongruence}
For $\bowtie\in\{\sim,\approx\}$, if $P\bowtie Q$ then, for any $R$ :
\begin{enumerate}[\em(i)] 
\item $(P\comp R) \bowtie (Q\comp R)$.
\item $(R\comp P) \bowtie (R\comp Q)$.
\item $(P\ten R)  \bowtie (Q\ten R)$.
\item $(R\ten P)  \bowtie (R\ten Q)$.
\end{enumerate}
\end{prop}
\begin{proof}
The proof follows the standard format; we shall only treat case (i) as the others are similar.
For (i) we show that $\{(P\comp R, Q\comp R)\,|\,P\bowtie Q\}$ for $\bowtie\in\{\sim,\approx\}$ 
are bisimulations w.r.t. respectively the strong and weak semantics.
Suppose that $P\comp R \dtrans{a}{b} S$. 
The only possibility is that this transition was derived using \ruleLabel{Cut}. 
So $P\dtrans{a}{c} P'$, $R\dtrans{c}{b} R'$ for some $c,P',R'$ with $S=P' \comp R'$. 
Similarly, if $P\comp R \dtransw{a}{b} S$ then using part (i) of Lemma~\ref{lem:syntaxdecomposition} 
gives $P\dtransw{a}{c} P'$, $R\dtransw{c}{b} R'$ with $S=P' \comp R'$.
Using the fact that $P\bowtie Q$ we obtain corresponding matching transitions from $Q$ 
to $Q'$ where $P'\bowtie Q'$, and finally apply
\ruleLabel{Cut} to obtain matching transitions from $Q\comp R$ to $Q'\comp R$; 
the transition is thus matched and the targets stay in their respective relations.
\end{proof}

\subsection{Circuit diagrams}

%%%%%%%%%%%%%%%%%%%%%%%%%%%%%%%%%%%%%%%%%%%%%%%%%%%%%%%%%%%%%%%%%%%%%%%%%%%%%%%%%%%%

\label{subsec:graphicalRepresentation}

In subsequent sections it will often be convenient to use a graphical language
for Petri calculus terms. Diagrams in the language will
be referred to as \emph{circuit diagrams}. We shall be careful, when drawing diagrams,
to make sure that each diagram can be converted to a syntactic expression by ``scanning'' the 
diagram from left to right.

The following result, which confirms the associativity of $\comp$
and $\ten$ justifies the use of circuit diagrams to represent terms. 
\begin{lem}
Suppose that $\bowtie\in\{\sim,\approx\}$.
\begin{enumerate}[\em(i)] \label{lem:syntaxBisimilarities} 
\item Let $P\typ\sort{k}{l}$, $Q\typ\sort{l}{m}$, $R\typ\sort{m}{n}$. Then
\[
(P\comp Q)\comp R \bowtie P\comp (Q\comp R).
\]
\item Let $P\typ\sort{k}{l}$, $Q\typ\sort{m}{n}$, $R\typ\sort{t}{u}$. Then
\[
(P\ten Q) \ten R \bowtie P\ten (Q\ten R).
\]
\item Let $P\typ\sort{k}{l}$, $Q\typ\sort{l}{m}$, $R\typ\sort{n}{t}$, $S\typ\sort{t}{u}$.
Then 
\[
(P\comp Q) \ten (R\comp S) \bowtie (P\ten R)\comp (Q\ten S).
\]
\end{enumerate}
\end{lem}
\begin{proof}
Straightforward, using the inductive presentation of the operational semantics in
the case of $\sim$ and the conclusions of Lemma~\ref{lem:syntaxdecomposition} in
the case of $\approx$.
\end{proof}

\begin{figure}[t]
\[
\begin{tabular}{c @{\hspace{1cm}} c @{\hspace{2cm}} c @{\hspace{1cm}} c}
\emptyplace & \lowerPic{.6pc}{height=.7cm}{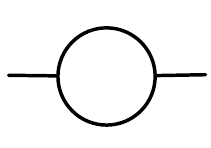} &
\tokenplace & \lowerPic{.6pc}{height=.7cm}{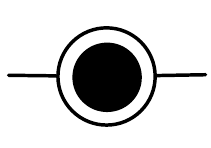} \\[8pt]
\id & \lowerPic{.6pc}{height=.7cm}{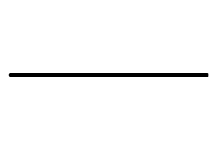} & \tw & \lowerPic{.6pc}{height=.7cm}{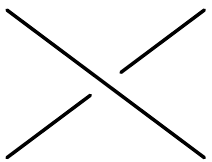} \\[8pt]
\diag & \lowerPic{.6pc}{height=.7cm}{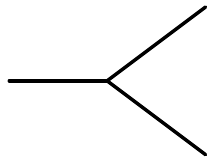}  & \codiag & \lowerPic{.6pc}{height=.7cm}{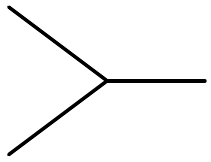} \\[8pt]
\rightEnd & \lowerPic{.6pc}{height=.7cm}{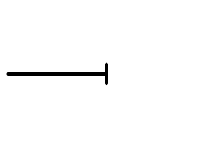} & \leftEnd & \lowerPic{.6pc}{height=.7cm}{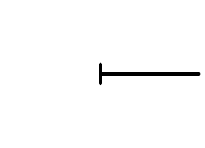} \\[8pt]
\ldiag & \lowerPic{.6pc}{height=.7cm}{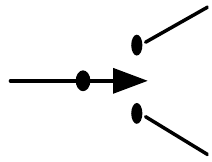} & \lcodiag & \lowerPic{.6pc}{height=.7cm}{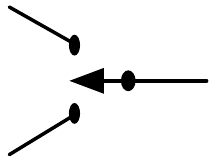} \\[8pt]
\rzero & \lowerPic{.6pc}{height=.7cm}{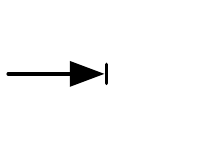} & \lzero & \lowerPic{.6pc}{height=.7cm}{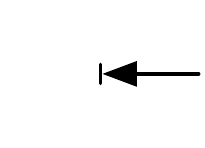} \\[8pt]
\end{tabular}
\]
\caption{Circuit diagram components.\label{fig:circDiagComponents}}
\end{figure}

Each of the language constants is represented by a circuit component
listed in \fig{\ref{fig:circDiagComponents}}.

% COMPOUND TERMS
For the translations of Section~\ref{sec:netstosyntax} we shall need additional families
of compound terms, indexed by $n\in\N_+$:
\[
\id_n\typ\sort{n}{n} \ \ 
\bang_n\typ\sort{n}{0} \ \  
\zero_n \typ\sort{n}{0} \ \ 
d_n\typ\sort{0}{2n} \ \ 
e_n\typ\sort{2n}{0} \ \ 
\Delta_n\typ\sort{n}{2n} \ \  
\nabla_n\typ\sort{2n}{n}. 
\]
Their definitions, given below, are less intuitive than their behaviour,
which we state first. 
Under the strong semantics, it is characterised
in each case by the following rules:
\begin{equation}\label{eq:specification}
\reductionRule{\alpha\in\{0,1\}^n}{\id_n\dtrans{\alpha}{\alpha}\id_n}
\quad
\reductionRule{\alpha\in\{0,1\}^n}{\bang_n\dtrans{\alpha}{}\bang_n}
\quad
\reductionRule{\phantom{\alpha\in\{0,1\}^n}}{\zero_n\dtrans{0^n}{}\zero_n}
\quad
\reductionRule{\alpha\in\{0,1\}^n}{d_n\dtrans{}{\alpha\alpha}d_n}
\quad
\reductionRule{\alpha\in\{0,1\}^n}{e_n\dtrans{\alpha\alpha}{}e_n}
\quad
\reductionRule{\alpha\in\{0,1\}^n}{\Delta_n \dtrans{\alpha}{\alpha\alpha} \Delta_n}
\quad
\reductionRule{\alpha\in\{0,1\}^n}{\nabla_n \dtrans{\alpha\alpha}{\alpha} \nabla_n}
\end{equation}
and their weak semantics is characterised by:
\begin{equation}\label{eq:wspecification}
\reductionRule{\alpha\in\N^n}{\id_n\dtransw{\alpha}{\alpha}\id_n}
\quad
\reductionRule{\alpha\in\N^n}{\bang_n\dtransw{\alpha}{}\bang_n}
\quad
\reductionRule{\phantom{\alpha\in\{0,1\}^n}}{\zero_n\dtransw{0^n}{}\zero_n}
\quad
\reductionRule{\alpha\in\N^n}{d_n\dtransw{}{\alpha\alpha}d_n}
\quad
\reductionRule{\alpha\in\N^n}{e_n\dtransw{\alpha\alpha}{}e_n}
\quad
\reductionRule{\alpha\in\N^n}{\Delta_n \dtransw{\alpha}{\alpha\alpha} \Delta_n}
\quad
\reductionRule{\alpha\in\N^n}{\nabla_n \dtransw{\alpha\alpha}{\alpha} \nabla_n}
\end{equation}

Intuitively, $\id_n$, $\bang_n$ and $\zero_n$ correspond to $n$ parallel copies 
 of $\id$, $\bang$ and $\zero$, respectively. Connector $d_n$ (and its dual 
$e_n$) stands for the synchronisation of $n$ pairs of wires. For $n=2$, the only allowed transitions
 under the strong semantics are  
 $d_2\dtrans{\phantom{0000}}{0000}d_2$, $d_2\dtrans{\phantom{0000}}{0101}d_2$,  
 $d_2\dtrans{\phantom{0000}}{1010}d_2$ and 
 $d_2\dtrans{\phantom{0000}}{1111}d_2$, i.e., all labels 
 that are concatenations of two identical strings of length 2. 
Connector $\Delta_n$ (and its dual $\nabla_n$) is similar but duplicates 
any label $\alpha$ in the other interface.

We now give the definitions: first we let 
$\id_n\Defeq\bigotimes_{n}\id$,
$\bang_n\Defeq\bigotimes_{n}\bang$ and
$\zero_n\Defeq\bigotimes_{n}\zero$.
In order to define the remaining terms we first
define $\tw_{n}\typ\sort{n+1}{n+1}$ recursively as follows:
\[
\tw_{1}\Defeq \tw \qquad \tw_{n+1} \Defeq (\tw_{n} \ten \id)\comp (\id_n\ten \tw).
\]
A simple induction confirms that the semantics of $\tw_n$ is characterised
as follows:
\begin{equation*}
\reductionRule{a\in\{0,1\},\,\alpha\in\{0,1\}^n}{\tw_n\dtrans{a\alpha}{\alpha a}\tw_n}
\qquad
\reductionRule{a\in\N,\,\alpha\in\N^n}{\tw_n\dtransw{a\alpha}{\alpha a}\tw_n}
\end{equation*}
Now because
$d_n$ and $e_n$, as well as $\Delta_n$ and $\nabla_n$ are symmetric, here we only 
give the constructions of $d_n$ and $\Delta_n$.
We define $\Delta_n$ recursively:
%\[
%d_1 \Defeq \leftEnd\comp\diag \qquad
%d_{n+1} \Defeq d_n \comp(\id_n\ten d_1 \ten \id_n)\comp(\id_{n+1}\ten \tw_{n})
%\]
\[
\Delta_1 \Defeq \diag \qquad
\Delta_{n+1} \Defeq (\diag \ten \Delta_{n})\comp (\id \ten \tw_{n} \ten \id_{n})
\]
Then, we let  
$d_n \Defeq \leftEnd_{n} \comp \Delta_{n}$ for $\leftEnd_{n} \Defeq \bigotimes_{n}\leftEnd$.

An easy induction on the derivation of a transition confirms that
these constructions produce terms whose semantics is characterised by 
\eqref{eq:specification} and \eqref{eq:wspecification}. 

\subsection{Relationship between strong and weak semantics}
%OBSERVATIONS ON RELATIONSHIP BETWEEN STRONG AND WEAK
%The following observation relates strong and weak semantics
%of the Petri Calculus.
%\begin{prop}~
%\begin{enumerate}[(i)]
%\item If $P\dtrans{a}{b} Q$ then $P\dtransw{a}{b}Q$.
%\item If $P\sim Q$ then $P\approx Q$.
%\end{enumerate}
%\end{prop}
%\begin{proof}
%(i) is immediate from the definition.
%
%(ii) it suffices to show that $\sim$ is a bisimulation with respect to the weak
%semantics.
%Indeed if $P\sim Q$ and $P\dtransw{a}{b} P'$.
%\end{proof}

It is immediate from the definition that if $P\dtrans{a}{b} Q$ then
$P\dtransw{a}{b} Q$. Perhaps surprisingly (cf.\ Remark~\ref{remark:weak}), it is \emph{not} true that 
$P\sim Q$ implies $P\approx Q$. Indeed, consider the term $\diag\comp\lcodiag$
with circuit diagram shown below.
\[
\includegraphics[height=1cm]{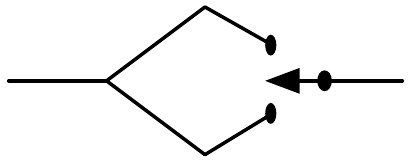}
\]
It is not difficult to verify that the only transition derivable using the strong
semantics is the trivial $\diag\comp\lcodiag \dtrans{0}{0} \diag\comp\lcodiag$.
Hence, $\diag\comp\lcodiag \mathrel{\sim} \rzero\comp\lzero$.
Instead, in the weak semantics we have the following derivation:
\[
\begin{prooftree}
\[
\justifies
\diag \dtransw{1}{11} \diag
\]
\[
\[
\justifies
\lcodiag \dtransw{01}{1} \lcodiag
\]
\quad
\[
\justifies
\lcodiag \dtransw{10}{1} \lcodiag
\]
\justifies
\lcodiag \dtransw{11}{2} \lcodiag
\using \ruleLabel{Weak*}
\]
\justifies
\diag \comp \lcodiag
\dtransw{1}{2} \diag \comp \lcodiag
\using \ruleLabel{Cut}
\end{prooftree}
\]
and, indeed, it is not difficult to show that 
$\diag\comp\lcodiag \dtransw{m}{n} P$ iff
$P=\diag\comp\lcodiag$ and $n=2m$. It follows 
that $\diag\comp\lcodiag \not\approx \rzero\comp\lzero$ (the only transition derivable with the weak semantics is the trivial $\rzero\comp\lzero\dtransw{0}{0} \rzero\comp\lzero$).
 In Section~\ref{sec:translating-weak-ce} we shall study
these terms further and refer to them as (right) amplifiers.

\section{P/T Calculus}
\label{sec:petritile}

This section introduces an extension of the Petri calculus 
with buffers that may contain an 
unbounded number of tokens.  We replace the terms $\emptyplace$ and $\tokenplace$
of the Petri calculus by a denumerable set of constants $\tokensplace{n}$ (one for 
any $n\in\nat$), each of them representing
a buffer containing $n$ tokens. In particular, $\tokensplace{0}$ stands for $\emptyplace$ and $\tokensplace{1}$ for $\tokenplace$. All remaining terms have analogous meaning.
We give the BNF for the syntax of the P/T Calculus below, with $n\in \nat$ and $\tokensplace{n}:(1,1)$.
\begin{equation*}
P \bnfEq \tokensplace{n} 
\bnfSep \id \bnfSep \tw
\bnfSep \diag \bnfSep \codiag 
\bnfSep \rightEnd \bnfSep \leftEnd 
\bnfSep \ldiag \bnfSep \lcodiag
\bnfSep \rzero \bnfSep \lzero 
\bnfSep P\ten P \bnfSep P \comp P
\end{equation*}

We rely on a  sorting discipline analogous to the one of the Petri calculus. The inference rules for the all terms but 
$\tokensplace{n}$ are those of  Fig.~\ref{fig:sortInferenceRules}, where $h,k,n\in\nat$ and $\alpha,\beta\in \nat^{*}$. For $\tokensplace{n}$ we add the following:
\[
\reductionRule{}{ \typeJudgment{}{\tokensplace{n}}{\sort{1}{1}} }\quad 
\]

% OPERATIONAL SEMANTICS
The operational semantics is shown in Fig.~\ref{fig:ptcalcopsem}. We remark that rules are now 
schemes. For instance, there is one particular instance of Rule \ruleLabel{TkIO$_{n,h,k}$} for any possible choice of $n$, $h$ and $k$. We  have just one scheme for buffers. In fact,  rules \ruleLabel{TkI} and \ruleLabel{TkO} %and \ruleLabel{TkO2} 
of the Petri calculus   (Fig.~\ref{fig:operationalSemantics}) are 
obtained as particular instances of  \ruleLabel{TkIO$_{n,h,k}$}, namely
\ruleLabel{TkIO$_{0,1,0}$} and \ruleLabel{TkIO$_{1,0,1}$}. The semantics 
for all stateless connectors are defined so that they agree with their 
corresponding \emph{weak}
semantics in the Petri calculus (see Proposition~\ref{pro:weakConstants}).

We say that $P\dtrans{\alpha}{\beta} Q$ \emph{strongly} if we can prove that $P\dtrans{\alpha}{\beta} Q$ without using rule \ruleLabel{Weak*}.
As for the Petri calculus, the weak variant is obtained by additionally allowing the unrestricted
use of rule \ruleLabel{Weak*} and we write weak transitions
 with a thick transition arrow: $P\dtransw{\alpha}{\beta}Q$.

\begin{figure}[t]
\[
\derivationRule{n,h,k\in \nat\quad k\leq n}{\tokensplace{n} \dtrans{h}{k} \tokensplace{n+h-k}}{TkIO$_{n,h,k}$} \hspace{2.3cm}
%\derivationRule{}{\tokenplace \dtrans{0}{1} \emptyplace}{TkO} \quad
%\derivationRule{}{\tokenplace \dtrans{1}{1} \tokenplace}{TkO2} \quad
\derivationRule{k\in \nat}{\id  \dtrans{k}{k} \id }{Id$_{k}$} \hspace{2.3cm}
\derivationRule{h,k\in\nat}{\tw \dtrans{hk}{kh} \tw}{Tw$_{h,k}$}
\]
\[
\derivationRule{k\in \nat}{\rightEnd \dtrans{k}{} \rightEnd}{$\rightEnd_{k}$} \hspace{1.7cm}
\derivationRule{k\in \nat}{\leftEnd \dtrans{}{k} \leftEnd}{$\leftEnd_{k}$} \hspace{1.7cm}
\derivationRule{k\in \nat}{\diag \dtrans{k}{k \labelSep k} \diag}{$\diag_{k}$} \hspace{1.7cm}
\derivationRule{k\in \nat}{\codiag \dtrans{k \labelSep k}{k} \codiag}{$\codiag_{k}$} 
\]
\[
\derivationRule{h,k\in \nat}{\ldiag \dtrans{h+k}{ h\labelSep k} \ldiag}{$\ldiag_{h,k}$} \hspace{1.6cm}
\derivationRule{h,k\in \nat}{\lcodiag \dtrans{h\labelSep k}{h+k} \lcodiag}{$\lcodiag_{h,k}$} \hspace{1.6cm}
\derivationRule{\mathsf{C}\typ\sort{k}{l}\mbox{ \footnotesize a basic connector}}{\mathsf{C}\dtrans{0^k}{0^l} \mathsf{C}}{Refl}
\]
\[
\derivationRule{P\dtrans{\alpha}{\gamma} Q \quad R\dtrans{\gamma}{\beta} S}
{P\comp R \dtrans {\alpha}{\beta} Q\comp S}{Cut} \qquad
\derivationRule{P\dtrans{\alpha_1}{\beta_1} Q\quad R\dtrans{\alpha_2}{\beta_2} S} 
{P\ten R \dtrans{\alpha_1\labelSep\alpha_2}{\beta_1\labelSep\beta_2} Q\ten S}{Ten} \qquad
\derivationRule{P\dtrans{\alpha_1}{\beta_1} P'\quad P'\dtrans{\alpha_2}{\beta_2} Q }
{P\dtrans{\alpha_1+\alpha_2}{\beta_1+\beta_2} Q}{Weak*}
\]
\caption{Structural rules for operational semantics of P/T calculus, where $\alpha,\beta,\gamma\in\N^*$.}
\label{fig:ptcalcopsem}
\end{figure}
 
As for the Petri calculus, we refer to   
$\{\tokensplace{n},\,\id,\,\tw,\,\diag,\codiag,
\,\rightEnd,\,\leftEnd,\,\ldiag,\,\lcodiag,\,\rzero,\,\lzero\}$ as the
\emph{basic connectors},
%Moreover, we call a connector \emph{stateless}
%if it is defined without using the constant $\tokensplace{n}$.
%
%Similarly to the situation in the Petri calculus
and a term $P$ is \emph{stateless} if $\sigma(P) \cap \{\, \tokensplace{n} \,\mid\, n\in \nat\,\} = \emptyset$, i.e.,
if $P$ does not contain any subterm of the form $\tokensplace{n}$. 
It is easy to show that the conclusion
of Lemma~\ref{lem:stateless} also holds in the P/T calculus.
\begin{lem}
Let $P$ be a stateless P/T calculus term. If $P\dtrans{\alpha}{\beta} Q$  and  $P\dtrans{\alpha'}{\beta'} Q'$ , then $Q = Q' = P$ and $P\dtrans{\alpha+\alpha'}{\beta+\beta'} P$.
\end{lem}
\begin{proof} The proof follows by induction on the structure of $P$. Since $P$ is stateless it cannot be of the form $\tokensplace{m}$. The cases corresponding to the 
remaining basic connectors are straightforward. Cases for sequential ($;$) and parallel ($\oplus$) composition
 follow by using the inductive hypothesis on both subterms. 
\end{proof}

\begin{cor}
\label{lem:only-strong-stateless}
Let $P$ be a stateless term. For any $\alpha,\beta$, $P\dtransw{\alpha}{\beta} Q$ if and only if $P\dtrans{\alpha}{\beta} Q$ and $Q = P$.
\end{cor}

\begin{lem}
\label{lemma:place-inv-strong}
Let $n,h,k\in\nat$. Then, $\tokensplace{n} \dtrans{h}{k} Q$  strongly iff   $k\leq n$ and  $Q =  \tokensplace{n+h-k}$.
\end{lem}

\begin{proof} Straightforward since the only possible strong derivations for $\tokensplace{n}$ are obtained by using \ruleLabel{TkIO$_{n,h,k}$}.
\end{proof}

Note that, $\tokensplace{n} \dtrans{h}{k} \tokensplace{m}$ does not imply $k\leq n$ for  weak transitions. For instance, the  transitions 
$\tokensplace{0} \dtrans{1}{0} \tokensplace{1}$ and  $\tokensplace{1} \dtrans{0}{1} \tokensplace{0}$ can obtained from rule \ruleLabel{TkIO$_{n,h,k}$}.
Then, we can derive $\tokensplace{0} \dtransw{1}{1} \tokensplace{0}$ by using rule \ruleLabel{Weak*}.
This example makes it evident that the weak transitions account for the banking semantics.

\begin{lem}
\label{lemma:place-inv-weak}
Let $n,h,k\in\nat$. Then,  $\tokensplace{n} \dtransw{h}{k} Q$ if and only if $k\leq n+h$  and $Q=\tokensplace{n+h-k}$.
\end{lem}
\begin{proof}
See Appendix~\ref{app:proof-petritile}. 
\end{proof}

The following example shows that any buffer containing $n+m$ tokens can be seen as the combination of two buffers containing, respectively, $n$ and $m$ 
tokens. This idea will be reprised in Section~\ref{sec:petri-tile-calculus}
 to show that P/T nets can be represented with a finite set of constants  (instead of using the 
infinite set  presented in this section). 

\begin{exa}
Given $n,m\in \nat$, it is easy to check that $P = \ldiag ; (\tokensplace{n} \otimes \tokensplace{m}) ; \lcodiag : (1,1)$ is (strong and weak)  bisimilar to $\tokensplace{n+m}$. For the strong case, the only non-trivial behaviour of $P$ is obtained as follows. By Lemma~\ref{lemma:place-inv-strong},  ${\tokensplace{n} \dtrans{h_1}{k_1} Q_1}$ and $Q_1= \tokensplace{n+h_1-k_1}$ with ${k_1\leq n}$, and similarly, ${\tokensplace{m} \dtrans{h_2}{k_2} Q_2}$ and $Q_2= \tokensplace{m+h_2-k_2}$ with ${k_2\leq m}$. By using  rules \ruleLabel{$\ldiag$}, \ruleLabel{Ten} and  \ruleLabel{$\lcodiag$} we derive
  $P  \dtrans {h_1+h_2}{k_1+k_2}  \ldiag ;  (\tokensplace{n+h_1-k_1} \otimes \tokensplace{m+h_2-k_2});\lcodiag$.  
  From $k_1\leq n$ and $k_2\leq m$ we get   $k_1+k_2 \leq n_1+ n_2$. Then,  $\tokensplace{n+m} \dtrans {h_1+h_2}{k_1+k_2}   \tokensplace{n+m+h_1+h_2-k_1-k_{2}} = \tokensplace{(n+h_1-k_1)+(m+h_2-k_2)}$  by Lemma~\ref{lemma:place-inv-strong}. Conversely,  by Lemma~\ref{lemma:place-inv-strong} the non trivial behaviours of $\tokensplace{n+m}$ are $\tokensplace{n+m}\dtrans{h}{k}\tokensplace{n+m+h-k}$ with $k\leq n+m$. As done before, we can derive  $P  \dtrans {h_1+h_2}{k_1+k_2}  \ldiag ;  (\tokensplace{n+h_1-k_1} \otimes \tokensplace{m+h_2-k_2});\lcodiag$ for any  $k_1$, $k_2$, $h_1$, $h_2\in\nat$ s.t.  ${k_1\leq n}$, ${k_2\leq m}$, $k = k_1+k_2$ and $h = h_1+h_2$ by using Lemma~\ref{lemma:place-inv-strong} and rules \ruleLabel{$\ldiag$}, \ruleLabel{Ten} and  \ruleLabel{$\lcodiag$}.
The weak case follows analogously by using Lemma~\ref{lemma:place-inv-weak} instead of Lemma~\ref{lemma:place-inv-strong}.
\end{exa}

The following technical result is similar to Lemma~\ref{lem:syntaxdecomposition} and shows that we can assume without loss of generality that the last applied rule 
in the derivation of a transition for a term of the form $P\comp Q$ or $P\ten Q$  is, respectively, \ruleLabel{Cut} and \ruleLabel{Ten}.

\begin{lem}\label{lem:syntaxdecomposition-pt}~
\begin{enumerate}[\em(i)]
\item
If $P\comp R \dtransw{\alpha}{\beta} Q$ then there exist 
	$P'$, $R'$, $\gamma$ such that $Q=P'\comp R'$,
	 $P\dtransw{\alpha}{\gamma} P'$ and $R\dtransw{\gamma}{\beta} R'$.
\item
If $P\ten R\dtransw{\alpha}{\beta} Q$ then there exist 
	$P'$, $R'$ such that $Q=P'\ten R'$,
		$P\dtransw{\alpha_1}{\beta_1} P'$, $R\dtransw{\alpha_2}{\beta_2} R'$ 
		  with $\alpha=\alpha_1\alpha_2$ and $\beta=\beta_1\beta_2$.
\end{enumerate}
\end{lem}
\begin{proof}
(i) $\Rightarrow$) We proceed by induction on the structure of the derivation. 
If the last rule used in the derivation was \ruleLabel{Cut} then we are finished.
By examination of the rules in \fig{\ref{fig:operationalSemantics}} the only other
possibility is \ruleLabel{Weak*}. Then, the derivation has the following shape:
\begin{equation}\label{eq:lastStep-pt}
\prooftree
P \comp R \dtransw{\alpha_0}{\beta_0} Q_1\quad Q_1 
\dtransw{\alpha_1}{\beta_1} Q
\justifies
P \comp R \dtransw{\alpha}{\beta} Q
\endprooftree
\end{equation}
where $\alpha=\alpha_0+\alpha_1$ and $\beta=\beta_0+\beta_1$. 
By inductive hypothesis on the first premise
\begin{equation}\label{eq:premise-one}
Q_1 = P_1 \comp R_1 \qquad P \dtransw{\alpha_0}{\gamma_0} P_1 \qquad 
R \dtransw{\gamma_0}{\beta_0} R_1
\end{equation}
Since $Q = P_1 \comp R_1$, by inductive hypothesis on the second premise
of~(\ref{eq:lastStep-pt}) 
\begin{equation}\label{eq:premise-two}
Q = P_2 \comp R_2 \qquad P_1 \dtransw{\alpha_1}{\gamma_1} P_2 \qquad 
R \dtransw{\gamma_1}{\beta_2} R_2
\end{equation}
From~(\ref{eq:premise-one}) and~(\ref{eq:premise-two}), we can build the following proof 
in which last applied rule is  \ruleLabel{Weak*}:

\[
\derivationRule{
\derivationRule
   { P \dtransw{\alpha_0}{\gamma_0} P_1 \quad 
      P_1 \dtransw{\alpha_1}{\gamma_1} P_2}
   { P \dtransw{\alpha_0+\alpha_1}{\gamma_0+\gamma_1} P_2}
   {Weak*}
\quad
\derivationRule
   { R \dtransw{\alpha_0}{\gamma_0} R_1 \quad 
      R_1 \dtransw{\alpha_1}{\gamma_1} R_2}
   { R \dtransw{\gamma_0+\gamma_1}{\beta_0+\beta_1} R_2}
   {Weak*}
}
{ P \comp R \dtransw{\alpha}{\beta} P_2;R_2}
{Cut}
\]

$\Leftarrow$) Immediate by using rule  \ruleLabel{Weak*}.

The proof of (ii) is similar.
\end{proof}

\noindent As in the Petri calculus we denote bisimilarity on the strong semantics by $\sim$ and bisimilarity on the weak semantics by $\approx$. The following result shows that both equivalence relations are congruences also for P/T nets.

\begin{prop}[Congruence]\label{pro:petricongruence-pt}
For $\bowtie\in\{\sim,\approx\}$, if $P\bowtie Q$ then, for any $R$ :
\begin{enumerate}[\em(i)] 
\item $(P\comp R) \bowtie (Q\comp R)$.
\item $(R\comp P) \bowtie (R\comp Q)$.
\item $(P\ten R)  \bowtie (Q\ten R)$.
\item $(R\ten P)  \bowtie (R\ten Q)$.
\end{enumerate}
\end{prop}
\begin{proof}
The proof follows as the one for Proposition~\ref{pro:petricongruence}, but 
we use  Lemma~\ref{lem:syntaxdecomposition-pt} instead of 
Lemma~\ref{lem:syntaxdecomposition}.
\end{proof}

\section{Translating terms to nets}
\label{sec:syntaxtonets}

In this section we give several
straightforward translations from the process algebras studied in
Sections~\ref{sec:syntax} and~\ref{sec:petritile} to 
the nets with boundaries studied in~\ref{sec:nets} and~\ref{sec:ptboundaries}.
In particular, this section contains translations:
\begin{enumerate}[(i)]
\item from C/E calculus terms with strong semantics to C/E nets (Theorem~\ref{thm:petritonets});
\item from C/E calculus terms with weak semantics to weak C/E nets\footnote{See Remark~\ref{rmk:weakCENets}.} (Proposition~\ref{pro:petritoweaknets});
\item from P/T calculus terms with strong semantics to P/T nets with standard semantics (Theorem~\ref{th:correspondence-ptfortiles});
\item from P/T calculus terms with weak semantics to P/T nets with banking semantics
(Theorem~\ref{th:correspondence-ptfortiles-weak}).
\end{enumerate}

\noindent The translations rely on the facts that (1) there is a simple net %that corresponds
for each basic connector and (2) the two operations on syntax agree with the
corresponding operations on nets with boundaries. In each case the translations both preserve and 
reflect semantics.
%%%%%%%%%%%%%%%%%%%%%%%%%%%%%%%%%%%%%%%%%%%%%%%%%%%%%%%%%%%%%%%%%%%%%%%%%%%%%%%%%%%

\subsection{Translating Petri calculus terms to C/E nets}
\label{sec:syntaxToNets}

%%%%%%%%%%%%%%%%%%%%%%%%%%%%%%%%%%%%%%%%%%%%%%%%%%%%%%%%%%%%%%%%%%%%%%%%%%%%%%%%%%%

\begin{figure*}[t]
\subfigure[C/E buffers.]{\label{fig:syntaxToNets-ce-buffers}
\begin{math}
\begin{tabular}{r@{$\quad\Defeq\quad$}l @{\hspace{2cm}} r@{$\quad\Defeq\quad$}l }
$\semanticsOf{\emptyplace}$ & \lowerPic{.5pc}{height=.6cm}{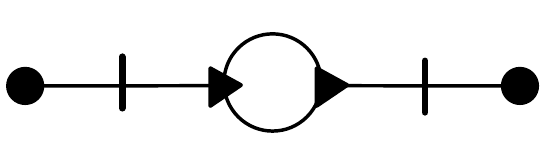} &
$\semanticsOf{\tokenplace}$ & \lowerPic{.5pc}{height=.6cm}{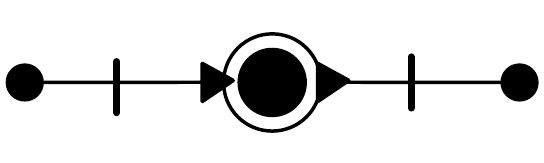} \\[8pt]
\end{tabular}
\end{math}
}
\subfigure[Stateless connectors.]{\label{fig:syntaxToNets-connectors}
\begin{math}
\begin{tabular}{r@{$\quad\Defeq\quad$}l @{\hspace{2cm}} r@{$\quad\Defeq\quad$}l }
$\semanticsOf{\id}$ & \lowerPic{.5pc}{height=.6cm}{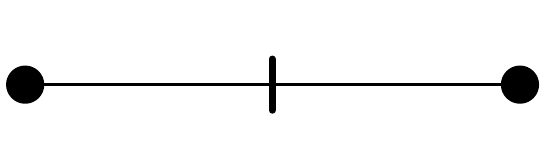} &
$\semanticsOf{\tw}$ & \lowerPic{.5pc}{height=.6cm}{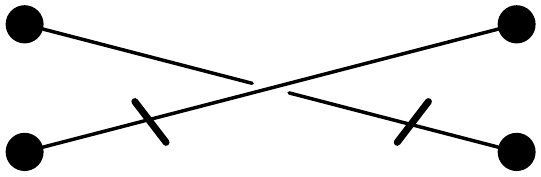} \\[8pt]
$\semanticsOf{\diag}$  & \lowerPic{.5pc}{height=.6cm}{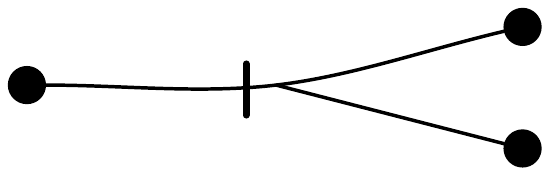} &
$\semanticsOf{\codiag}$ & \lowerPic{.5pc}{height=.6cm}{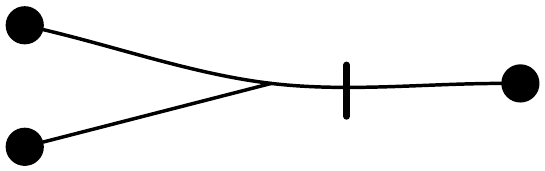} \\[8pt] 
$\semanticsOf{\leftEnd}$ & \lowerPic{.5pc}{height=.6cm}{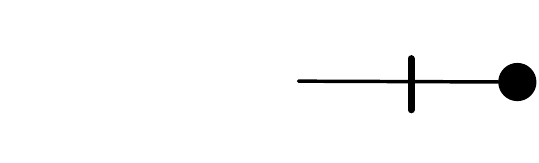} &
$\semanticsOf{\rightEnd}$ & \lowerPic{.5pc}{height=.6cm}{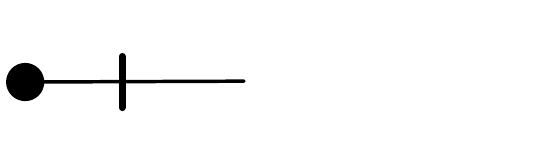} \\[8pt]
$\semanticsOf{\ldiag}$  & \lowerPic{.5pc}{height=.6cm}{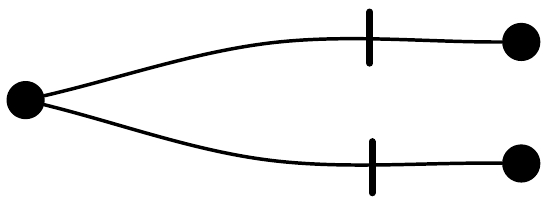} &
$\semanticsOf{\lcodiag}$ & \lowerPic{.5pc}{height=.6cm}{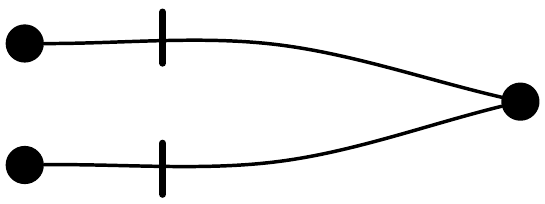} \\[8pt]
$\semanticsOf{\lzero}$ & \lowerPic{.5pc}{height=.6cm}{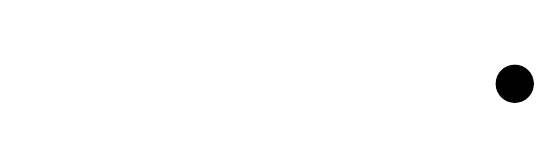} &
$\semanticsOf{\rzero}$ & \lowerPic{.5pc}{height=.6cm}{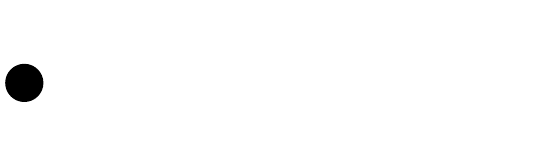} \\[8pt]
\end{tabular}
\end{math}}
\subfigure[P/T buffers.]{\label{fig:syntaxToNets-pt-buffers}
\hspace{2cm}\begin{math}
\encodeinv{\tokensplace{n}} \Defeq \lowerPic{.5pc}{height=.6cm}{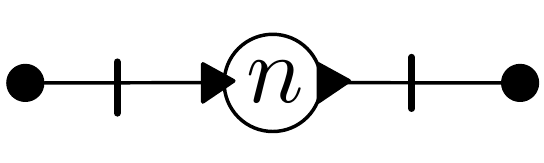} 
\end{math}\hspace{2cm}
}
\caption{Translation from basic connectors to nets.}\label{fig:syntaxToNets}
\end{figure*}

We start by giving a compositional translation from Petri calculus terms
to C/E nets with boundaries that preserves the strong semantics in 
a tight manner.

%Here we give two compositional translations from the Petri calculus to 
% marked C/E nets, one that gives tight correspondence between the strong 
%semantics and strong C/E nets another between the 
%weak semantics and weak C/E nets.

%The two cases are  similar: 
Each of the basic connectors of the Petri calculus has a corresponding C/E net
with the same semantics: this  
translation ($\semanticsOf{-}$) is given in
 \fig{\ref{fig:syntaxToNets-ce-buffers}} and \fig{\ref{fig:syntaxToNets-connectors}}, where we leave implicit that the contention relation is just the smallest relation induced by the sharing of ports.
The translation extends compositionally by letting
\[
\semanticsOf{T_1\comp T_2}\Defeq\semanticsOf{T_1}\comp\semanticsOf{T_2}
\mbox{ and }\semanticsOf{T_1\ten T_2}\Defeq\semanticsOf{T_1}\ten\semanticsOf{T_2}.
\]
%note that in the strong case we composition and tensor refers
%to the corresponding operations of strong C/E nets, while
%in the weak case to those of weak C/E nets--we will write $\semanticsOf{-}_s$
%and $\semanticsOf{-}_w$ to distinguish the two cases.
We obtain a very close operational correspondence between
Petri calculus terms and their translations to C/E nets, as stated by the following result.
\begin{thm}\label{thm:petritonets}
Let $T$ be a term of the Petri calculus.
\begin{enumerate}[\em(i)]
\item 
if $T \dtrans{\alpha}{\beta} T'$ 
then $\semanticsOf{T} \dtrans{\alpha}{\beta} \semanticsOf{T'}$.
\item
if $\semanticsOf{T} \dtrans{\alpha}{\beta} \marking{N}{X}$ then
there exists $T'$ such that
$T\dtrans{\alpha}{\beta}T'$ and $\semanticsOf{T'} = \marking{N}{X}$.
%\item 
%if $T \dtransw{\alpha}{\beta} T'$ 
%then $\semanticsOf{T}_w \dtransw{\alpha}{\beta} \semanticsOf{T'}_w$.
%\item
%if $\semanticsOf{T}_w\dtransw{\alpha}{\beta} \marking{N}{X}$ then
%there exists $T'$ such that
%$T\dtransw{\alpha}{\beta}T'$ and $\semanticsOf{T'}_w = \marking{N}{X}$.
\end{enumerate}
\end{thm}
\begin{proof}
(i) 
 We proceed by structural induction on $T$. When $T$ is a constant, 
each case can be shown to hold easily by inspection of the translations
illustrated in \fig{\ref{fig:syntaxToNets}}. 
Now if $T=P\comp Q$ and $T\dtrans{\alpha}{\beta} T'$ then
we have $P \dtrans{\alpha}{\gamma} P'$, $Q \dtrans{\gamma}{\beta} Q'$
and $T'=P';Q'$. Using the inductive hypothesis we obtain
$\semanticsOf{P}\dtrans{\alpha}{\gamma} \semanticsOf{P'}$
and 
$\semanticsOf{Q}\dtrans{\gamma}{\beta} \semanticsOf{Q'}$.
Then by Theorem~\ref{thm:netdecomposition} we obtain  
$\semanticsOf{P\comp Q}=\semanticsOf{P}\seqComp \semanticsOf{Q}
\dtrans{\alpha}{\beta}\semanticsOf{P'};\semanticsOf{Q'}
=\semanticsOf{P';Q'}$.
The case of $\ten$ is straightforward.
%similar and uses the conclusions of Lemma~\ref{lem:nettensordecomposition}.

(ii) Again we proceed by structural induction on $T$ and again for constants it is a matter
of examination. Suppose that $T=P\comp Q$ and
$\semanticsOf{T}\dtrans{\alpha}{\beta}\marking{N}{X}$.
Since $\semanticsOf{T}=\semanticsOf{P}\seqComp\semanticsOf{Q}$
then by Theorem~\ref{thm:netdecomposition} there exists $\gamma$
such that
$\semanticsOf{P}\dtrans{\alpha}{\gamma}\marking{N'}{X'}$,
$\semanticsOf{Q}\dtrans{\gamma}{\beta}\marking{N''}{X''}$ where
$\marking{N}{X} = \marking{N'}{X'} \seqComp \marking{N''}{X''}$. By the inductive hypothesis
$P \dtrans{\alpha}{\gamma} P'$,
$Q \dtrans{\gamma}{\beta} Q'$ 
with $\semanticsOf{P'} = \marking{N'}{X'}$
and $\semanticsOf{Q'} = \marking{N''}{X''}$.
Using \ruleLabel{Cut} we obtain $P\comp Q\dtrans{\alpha}{\beta} P'\comp Q'$ 
and clearly $\semanticsOf{P'\comp Q'}=\semanticsOf{P'}\seqComp \semanticsOf{Q'}= 
\marking{N}{X}$.
The case of $T=P\ten Q$ is again straightforward.
\end{proof}

\subsection{Translating P/T calculus terms to {P/T} nets.}
\label{sec:ptfortiles}

%%%%%%%%%%%%%%%%%%%%%%%%%%%%%%%%%%%%%%%%%%%%%%%%%%%%%%%%%%%%%%%%%%%%%%%%%%%%%%%%%%%

%\begin{figure}[t]
%\[
%\encodeinv{\tokensplace{n}} \Defeq \lowerPic{.5pc}{height=.6cm}{netTokensPlace} 
%\]
%\[
%\begin{tabular}{r@{$\quad\Defeq\quad$}l @{\hspace{2cm}} r@{$\quad\Defeq\quad$}l }
%$\encodeinv{\id}$ & \lowerPic{.5pc}{height=.6cm}{netIdentity} &
%$\encodeinv{\simmetry}$ & \lowerPic{.5pc}{height=.6cm}{netTwist} \\[8pt]
%$\encodeinv{\dup}$  & \lowerPic{.5pc}{height=.6cm}{netDiag} &
%$\encodeinv{\codup}$ & \lowerPic{.5pc}{height=.6cm}{netCodiag} \\[8pt] 
%$\encodeinv{\cobang}$ & \lowerPic{.5pc}{height=.6cm}{netLeftEnd} &
%$\encodeinv{\bang}$ & \lowerPic{.5pc}{height=.6cm}{netRightEnd} \\[8pt]
%$\encodeinv{\me}$  & \lowerPic{.5pc}{height=.6cm}{netLdiag} &
%$\encodeinv{\come}$ & \lowerPic{.5pc}{height=.6cm}{netLcodiag} \\[8pt]
%$\encodeinv{\cozero}$ & \lowerPic{.5pc}{height=.6cm}{netLzero} &
%$\encodeinv{\zero}$ & \lowerPic{.5pc}{height=.6cm}{netRzero} \\[8pt]
%\end{tabular}
%\]
%\caption{Encoding of basic connectors into {P/T} nets}
%\label{fig:encoding-basic-tiles}
%\end{figure}

The translation from P/T calculus to P/T nets is similar to the 
translation that we have already considered. We will
use the notation $\encodeinv{-}$ to emphasise that the codomain
of the translation is P/T nets, where
composition of nets is defined differently.
For $C$ a stateless connector, let $\encodeinv{C}\Defeq \semanticsOf{C}$
(considered as a P/T net)
as given in \fig{\ref{fig:syntaxToNets-connectors}} and the translation of the buffers of the P/T calculus is
in \fig{\ref{fig:syntaxToNets-pt-buffers}}.

As for C/E nets,  the encoding  is  homomorphic w.r.t. $\comp$ and $\oplus$: 
\[
\encodeinv{T_1 \comp T_2} \Defeq \encodeinv{T_1} \comp \encodeinv{T_2} \mbox{ and } 
\encodeinv{T_1 \oplus T_2} \Defeq \encodeinv{T_1} \oplus\encodeinv{T_2}. 
\]

We first consider P/T calculus with strong semantics and P/T nets with the standard semantics. 
\begin{thm} ~\label{th:correspondence-ptfortiles}
Let $T$ be a term of P/T calculus.
\begin{enumerate}[\em(i)]
\item 
if $T \dtrans{\alpha}{\beta} T'$
then $\encodeinv{T} \dtrans{\alpha}{\beta} \encodeinv{T'}$.
\item
if $\encodeinv{T}\dtrans{\alpha}{\beta} N_{\mathcal{X}}$ then
there exists a term $T'$ such that
$T\dtrans{\alpha}{\beta}T'$ and $\encodeinv{T'}=N_{\mathcal{X}}$.
\end{enumerate}
\end{thm}

\begin{proof}
(i) We proceed by structural induction on $T$. If $T=\tokensplace{n}$
then
%
%\[ \encodeinv{T} \bydef  \lowerPic{.5pc}{height=.6cm}{netTokensPlace}\]
   $\cell{\tokensplace{n}}{h}{k}{Q}$ implies $k\leq n$ and $Q = \tokensplace{n+h-k}$   by Lemma~\ref{lemma:place-inv-strong}. Consider 
   the net corresponding to the term $\tokensplace{n}$ given in \fig{\ref{fig:syntaxToNets-pt-buffers}} and let $\alpha$ be the transition on the left and $\beta$ 
   the transition on the right.
   Take $\mathcal{U} = h\alpha + k\beta$. It is immediate to check that $\encodeinv{T}\rightarrow_{\mathcal{U}}\encodeinv{\tokensplace{n+h-k}}$. 
   The cases corresponding 
   to the remaining constants can be shown to hold easily by inspection of the translations
   illustrated in \fig{\ref{fig:syntaxToNets}}. 
Now if $T=T_1;T_2$ and $T\dtrans{\alpha}{\beta} T'$ then we
have $T_1 \dtrans{\alpha}{\gamma} T_1'$, $T_2 \dtrans{\gamma}{\beta} T_2'$
and $T'=T_1';T_2'$. Using the inductive hypothesis we obtain
$\encodeinv{T_1}\dtrans{\alpha}{\gamma} \encodeinv{T_1'}$
and 
$\encodeinv{T_2}\dtrans{\gamma}{\beta} \encodeinv{T_2'}$.
 By Theorem~\ref{thm:ptnetdecomposition}($i$), 
\[\encodeinv{T_1;T_2}=\encodeinv{T_1};\encodeinv{T_2}\dtrans{\alpha}{\beta}\encodeinv{T_1'};\encodeinv{T_2'}
=\encodeinv{T_1';T_2'}.\]
The case for $\ten$ follows by using rule~\ruleLabel{ten}, inductive hypothesis on both premises and 
then parallel composition of nets. 

(ii) Again we proceed by structural induction on $T$ and again for constants it is a matter
of examination. Suppose that $T=T_1;T_2$ and
$\encodeinv{T}\dtrans{\alpha}{\beta}N_{\mathcal{X}}$.
Then, $\encodeinv{T}=\encodeinv{T_1};\encodeinv{T_2}$ by definition of the encoding. 
By Theorem~\ref{thm:ptnetdecomposition}($i$), there exists $\gamma$
such that
$\encodeinv{T_1}\dtrans{\alpha}{\gamma}{N_{1\mathcal{X}_1}}$,
$\encodeinv{T_2}\dtrans{\gamma}{\beta}{N_{2\mathcal{X}_2}}$.
By the inductive hypothesis
$T_1\dtrans{\alpha}{\gamma}T_1'$,
$T_2\dtrans{\gamma}{\beta}T_2'$ 
with $\encodeinv{T_1'} = {N_{1\mathcal{X}_1}}$
and $\encodeinv{T_2'} = {N_{2\mathcal{X}_2}}$.
Using \ruleLabel{Cut} we obtain $T_1;T_2\dtrans{\alpha}{\beta} T_1';T_2'$ 
and clearly $\encodeinv{T_1';T_2'}=\encodeinv{T_1'};\encodeinv{T_2'} = (N_1;N_2)_{\mathcal{X}}$ 
where $\mathcal{X} = \mathcal{X}_1 + \mathcal{X}_2$.
\end{proof}

Then, we extend the result to P/T calculus and P/T nets with weak semantics.
\newpage

\begin{thm} ~\label{th:correspondence-ptfortiles-weak}
Let $T$ be a term of P/T calculus.
\begin{enumerate}[\em(i)]
\item 
if $T \dtransw{\alpha}{\beta} T'$ 
then $\encodeinv{T} \dtransw{\alpha}{\beta} \encodeinv{T'}$.
\item
if $\encodeinv{T}\dtransw{\alpha}{\beta} {N_{\mathcal{X}}}$ then
there exists a term $T'$ such that
$T\dtransw{\alpha}{\beta}T'$  and $\encodeinv{T'}={N_{\mathcal{X}}}$.
\end{enumerate}
\end{thm}
\begin{proof} The proof follows analogously to  the one of Theorem~\ref{th:correspondence-ptfortiles} (for this 
case we rely on 
Lemma~\ref{lemma:place-inv-weak} and Theorem~\ref{thm:ptnetdecomposition} ($ii$)).
\end{proof}

To complete the picture, we also give a translation of Petri calculus terms with weak
semantics to weak C/E nets, that is P/T nets with
banking semantics where the marking is a subset (instead
of a multiset) of places. Again, the translation of basic connectors is defined as
in \fig{\ref{fig:syntaxToNets}}, and the translation of compound terms is homomorphic.
The proof is similar to the proof of Theorem~\ref{th:correspondence-ptfortiles-weak} (this is in 
particular due to the fact that in Theorem~\ref{thm:ptnetdecomposition} if $\mathcal{X}$ and $\mathcal{Y}$ are 
sets, so are $\mathcal{X}_M$, $\mathcal{X}_N$, $\mathcal{Y}_M$, and $\mathcal{Y}_N$).

\begin{prop}\label{pro:petritoweaknets}
Let $T$ be a term of the Petri calculus. 
\begin{enumerate}[\em(i)]
\item 
if $T \dtransw{\alpha}{\beta} T'$ 
then $\encodeinv{T} \dtransw{\alpha}{\beta} \encodeinv{T'}$.
\item
if $\encodeinv{T} \dtransw{\alpha}{\beta} \marking{N}{X}$ then
there exists $T'$ such that
$T\dtransw{\alpha}{\beta}T'$ and $\encodeinv{T'} = \marking{N}{X}$. 
\end{enumerate} 
\end{prop}

\section{Translating nets to terms}
\label{sec:netstosyntax}

In this section we exhibit translations from
the net models to process algebra terms. As with the translations
in Section~\ref{sec:syntaxtonets} all the translations preserve
and reflect semantics.
Concretely, we will define translations:
\begin{enumerate}[(i)]
\item from C/E nets to Petri calculus terms with strong semantics (Theorem~\ref{thm:strongcenetstosyntax});
\item from weak C/E nets (see Remark~\ref{rmk:weakCENets}) to Petri calculus terms with weak semantics (Theorem~\ref{thm:weakcenetstosyntax});
\item from P/T nets with standard semantics to P/T calculus terms with strong semantics (Theorem~\ref{theorem:correspondence-pt-into-tiles-strong});
\item from P/T nets with the banking semantics to P/T calculus terms with weak semantics (Theorem~\ref{theorem:correspondence-pt-into-tiles-weak}).
\end{enumerate}

\smallskip
First we treat the translation from C/E nets to Petri calculus terms. 
In order to do
this we shall need to first introduce and study particular kinds of Petri calculus
terms: \emph{relational forms} (Definition~\ref{defn:relationalForms}).
In order to translate P/T nets, these will be later generalised to
\emph{multirelational forms} (Definition~\ref{defn:multirelationalForms}), which
are relevant in the Petri calculus with weak semantics and the two
variants of the P/T calculus. These building blocks allow us to
translate any net with boundary to a corresponding process algebra term with
the same labelled semantics.

Relational and multirelational forms are built from more basic syntactic
building blocks: inverse functional forms (Definition~\ref{defn:inverseFunctionalForms}), 
direct functional forms (Definition~\ref{defn:directFunctionalForms}), and 
additionally for multirelational forms, amplifiers (Definition~\ref{defn:amplifiers}).
For $\Theta$ a set of Petri calculus terms,
%$ \in \{\tw,\,\diag,\,\codiag,\,\rightEnd,\,%
%\leftEnd,\,\ldiag,\,\lcodiag,\,\rzero,\,\lzero\}$ 
let
$T_\Theta$ denote the set of terms generated by the following grammar: 
\[
T_\Theta \bnfEq \theta \in \Theta \bnfSep \id \bnfSep T_\Theta \ten T_\Theta 
	\bnfSep T_\Theta \comp T_\Theta.
\]
We shall use $t_\Theta$ to range over terms of $T_\Theta$.

%RELATIONAL FORMS
\subsection{Functional forms} 
We start by introducing functional forms, which are instrumental to the definition of the
relational forms used in the proposed encoding.

%\begin{lem}[Permutations]\label{lem:permutations}
%Let $\varphi:\underline{k}\to\underline{k}$ be a bijection.
%Then there exists a term $\mathrm{per}_\varphi\in T_{\{\tw\}}$ with dynamics
%characterised by the following conditions: 
%\[
%\mathrm{per}_\varphi 
% \dtrans{\characteristic{U}}{\characteristic{V}} \mathrm{per}_\varphi
%\ 
%\Leftrightarrow
%\ 
% U\subseteq \underline{k} \mbox{ and } V=\varphi(U),
%\]
%\[
%\mathrm{per}_\varphi 
% \dtransw{\characteristic{\mathcal{U}}}{\characteristic{\mathcal{V}}} \mathrm{per}_\varphi
%\ 
%\Leftrightarrow
%\ 
% \mathcal{U}\in \multiset{\underline{k}} \mbox{ and } \mathcal{V}=\varphi(\mathcal{U}).
%\]
%\end{lem}

\begin{defi}[Inverse functional form]\label{defn:inverseFunctionalForms}
A term $t\typ\sort{k}{l}$ is said to be in right inverse  functional form
when it is in
$T_{\{\rightEnd\}} \comp T_{\{\diag\}} \comp T_{\{\tw\}}$. 
Dually, $t\typ\sort{k}{l}$ is in left inverse functional form 
when it is in
$T_{\{\tw\}} \comp T_{\{\codiag\}} \comp T_{\{\leftEnd\}}$.
\end{defi}

\begin{lem}\label{lem:inverseFunctionalForms}
For any function $f\from \underline{l} \to \underline{k}$ there
exists a term $\mathrm{riff}_f\typ\sort{k}{l}$ in right inverse functional form, the dynamics
of which are characterised by the following: 
\[
\mathrm{riff}_f \dtrans{\alpha}{\beta} \mathrm{riff}_f
\ 
\Leftrightarrow
\
 \exists U\subseteq \underline{k} \mbox{ s.t. }
  \alpha = \characteristic{U} \mbox{ and } 
   \beta = \characteristic{f^{-1}(U)}
\]
\[
\mathrm{riff}_f \dtransw{\alpha}{\beta} \mathrm{riff}_f
\
\Leftrightarrow
\ 
\exists\mathcal{U} \in \multiset{\underline{k}}\mbox{ s.t. } 
\alpha=\characteristic{\mathcal{U}} \mbox{ and } 
  \beta=\characteristic{f^{-1}(\mathcal{U})}
\]
The symmetric result holds for terms $t\typ\sort{l}{k}$ in left inverse functional 
form. That is, given a function $f\from \underline{l} \to \underline{k}$
there exists a term $\mathrm{liff}_f\typ\sort{l}{k}$ in left inverse functional 
form, the dynamics of which are characterised by the following:
\
\[
\mathrm{liff}_f \dtrans{\alpha}{\beta} \mathrm{liff}_f
\ 
\Leftrightarrow
\
 \exists U\subseteq \underline{k} \mbox{ s.t. }
  \beta = \characteristic{U} \mbox{ and } 
   \alpha = \characteristic{f^{-1}(U)}
\]
\[
\mathrm{liff}_f \dtransw{\alpha}{\beta} \mathrm{liff}_f
\
\Leftrightarrow
\ 
\exists\mathcal{U} \in \multiset{\underline{k}}\mbox{ s.t. } 
\beta=\characteristic{\mathcal{U}} \mbox{ and } 
  \alpha=\characteristic{f^{-1}(\mathcal{U})}
\]
\end{lem}
\begin{proof}
In Appendix~\ref{app:proof-cetranslations}.
\end{proof}

\begin{figure}[t]
\includegraphics[height=2cm]{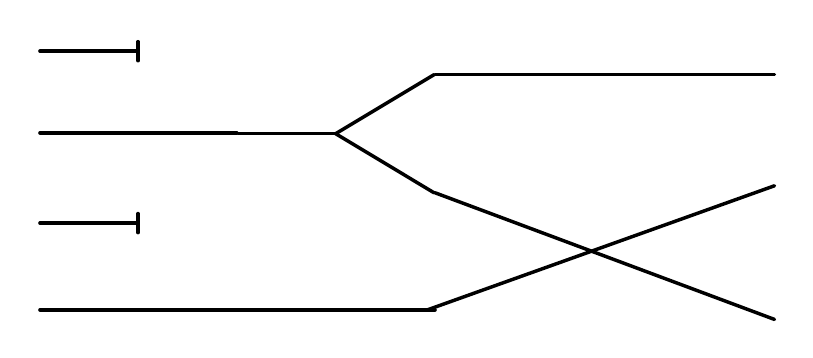}
\caption{Right inverse functional form.}
\label{fig:rifff}
\end{figure}

\begin{exa}
\label{ex:riff-liff}
 Let $f\from \ord{3}\to \ord{4}$ s.t. $f(0) = f(2) = 1$ and 
$f(1) = 3$. Then, $\mathrm{riff}_f : (4,3)$ can be defined as follows (see \fig{\ref{fig:rifff}}):
\[
   \mathrm{riff}_f = (\rightEnd \ten\ \id\ \ten \rightEnd \ten\ \id); (\diag  \ten \id);(\id\   \ten \tw)
\]

The term  $(\rightEnd \ten\ \id\ \ten \rightEnd \ten\ \id)$ captures the fact that  $0$ and $2$ are not in the 
image of $f$ (i.e., we write $\rightEnd$ attached to the corresponding ports), while $1$ and $3$ are (i.e., we write $\id$ for 
those ports). Term  $(\diag  \ten\ \id)$ says that the pre-image of $1$ has two elements (i.e., $\diag$) while the pre-image of 
$3$ has 1(i.e., $\id$). Finally, term $(\id\   \ten \tw)$ sorts the connections to the proper ports. 

Lemma~\ref{lem:inverseFunctionalForms} ensures that  the only
transitions of $\mathrm{riff}_f$ under the strong semantics are those in which $U\subseteq\ord{4}$ is observed over the left interface while its pre-image is 
observed over the right interface. For instance,  $\mathrm{riff}_f \dtrans{1010}{000}\mathrm{riff}_f$ (i.e., $f^{-1}(\{0,2\}) = \emptyset$),
$\mathrm{riff}_f \dtrans{0100}{101}\mathrm{riff}_f$ (i.e., $f^{-1}(\{1\}) = \{0,2\}$), $\mathrm{riff}_f \dtrans{1101}{111}\mathrm{riff}_f$ (i.e., $f^{-1}(\{0,1,3\}) = \{0,1,2\}$), and $\mathrm{riff}_f \dtrans{0000}{000}\mathrm{riff}_f$ (i.e., $f^{-1}(\emptyset) = \emptyset$) among others. Similarly, for the 
weak semantics we can obtain, e.g., $\mathrm{riff}_f \dtranswl{2201}{212}\mathrm{riff}_f$ (i.e., $f^{-1}(\{0,0,1,1,3\}) =\{0,0,1,2,2\}$).
The term $\mathrm{liff}_f : (3,4)$ can be defined analogously by ``mirroring'' $\mathrm{riff}_f : (4,3)$, i.e.,
\[
   \mathrm{liff}_f = (\id\   \ten \tw); (\codiag  \ten\ \id);(\leftEnd \ten\ \id\ \ten \leftEnd \ten\ \id).
\]
\end{exa}

\begin{defi}[Direct functional form]\label{defn:directFunctionalForms}
A term $t\typ\sort{k}{l}$ is said to be in right direct functional form
when it is in $T_{\{\tw\}} \comp T_{\{\lcodiag\}} \comp T_{\{\lzero\}}$
Dually, $t\typ\sort{k}{l}$ is in left direct functional form 
when it is in
$T_{\{\rzero\}} \comp T_{\{\ldiag\}} \comp T_{\{\tw\}}$
\end{defi}

\begin{lem}\label{lem:directFunctionalForms}
For each function $f\from \underline{k} \to \underline{l}$ there
exists a term $\mathrm{rdff}_f\typ\sort{k}{l}$ in right direct 
functional form, the dynamics of which are characterised by the following: 
\[
\mathrm{rdff}_f \dtrans{\alpha}{\beta} \mathrm{rdff}_f
\ 
\Leftrightarrow
\ 
\exists U\subseteq \underline{k} \mbox{ s.t. }
\forall u,v\in U.\; u\neq v \Rightarrow f(u) \neq f(v),\,
\alpha=\characteristic{U} 
\mbox{ and }\beta=\characteristic{f(U)}
\]
\[
\mathrm{rdff}_f \dtransw{\alpha}{\beta} \mathrm{rdff}_f
\ 
\Leftrightarrow
\ 
\exists\mathcal{U} \in \multiset{\underline{k}} \mbox{ s.t. } 
\alpha=\characteristic{\mathcal{U}} \mbox{ and }
\beta=\characteristic{f(\mathcal{U})}
\]
The symmetric result holds for terms
$t\typ\sort{l}{k}$ in left direct functional form.
That is, there exists a term $\mathrm{ldff}_f\typ\sort{l}{k}$ in left direct 
functional form with semantics characterised by the following:
\[
\mathrm{ldff}_f \dtrans{\alpha}{\beta} \mathrm{ldff}_f
\ 
\Leftrightarrow
\ 
\exists U\subseteq \underline{k} \mbox{ s.t. }
\forall u,v\in U.\; u\neq v \Rightarrow f(u) \neq f(v),\,
\beta=\characteristic{U} 
\mbox{ and }
\alpha=\characteristic{f(U)}
\]
\[
\mathrm{ldff}_f \dtransw{\alpha}{\beta} \mathrm{ldff}_f
\ 
\Leftrightarrow
\ 
\exists\mathcal{U} \in \multiset{\underline{k}} \mbox{ s.t. } 
\beta=\characteristic{\mathcal{U}} \mbox{ and }
\alpha=\characteristic{f(\mathcal{U})}.
\]
\end{lem}
\begin{proof}
In Appendix~\ref{app:proof-cetranslations}.
\end{proof}

\begin{figure}[t]
\includegraphics[height=2cm]{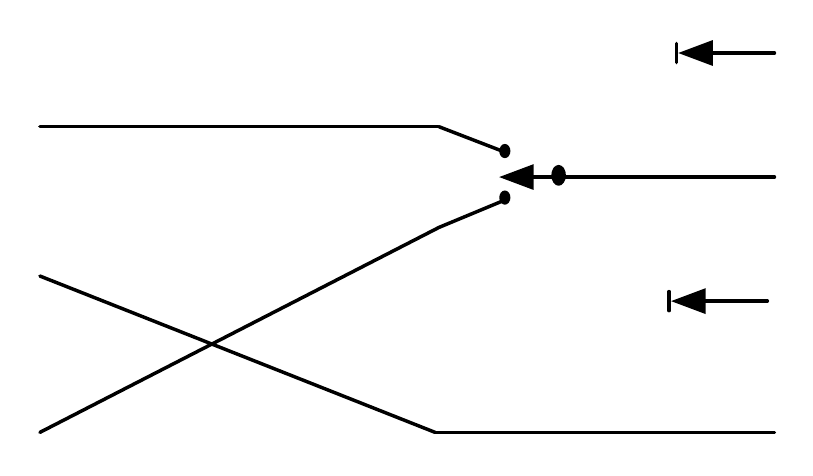}
\caption{Right direct functional form.}
\label{fig:rdfff}
\end{figure}

\begin{exa} The right direct functional form for  $f$ introduced in Example~\ref{ex:riff-liff} is as follow (see \fig{\ref{fig:rdfff}}): 
\[
   \mathrm{rdff}_f = (\id  \ten \tw); (\lcodiag  \ten \id);(\rzero \ten \ \id\  \ten \rzero \ten\ \id).
\]

The construction is analogous to the inverse functional form in Example~\ref{ex:riff-liff}. The term  $(\id  \ten \tw)$ exchange 
the order of  wires appropriately (it switches the ports $1$ and $2$), then the term $(\lcodiag  \ten \id)$ states the values $0$ and $2$ are mutually exclusive because they have the same image (i.e., $f(0)=f(2)= 1$). Finally,   $\rzero$ in $(\rzero \ten\ \id\ \ten \rzero \ten\ \id)$ denotes that  
 $0$ and $2$ (over the right interface) are not part of the image of $f$.  
It is worth noticing that the following transitions $\mathrm{rdff}_f \dtrans{100}{0100}\mathrm{rdff}_f$  (i.e., $f(\{0\})= \{1\}$), $\mathrm{rdff}_f \dtrans{001}{0100}\mathrm{rdff}_f$  (i.e., $f(\{2\})= \{1\}$) and $\mathrm{rdff}_f \dtrans{011}{0101}\mathrm{rdff}_f$ (i.e., $f(\{1,2\})= \{1,3\}$) are derivable under the strong semantics, while the transition $\mathrm{rdff}_f \dtrans{101}{0100}\mathrm{rdff}_f$  (i.e., $f(\{0,2\})= \{1\}$)  cannot be derived because the domain values $0$ and $2$ has the same image and, thus,  are in mutual exclusion. Differently,  the weak semantics allows us to consider multisets  of  domain values that may have the same image, e.g., we can derive 
$\mathrm{rdff}_f \dtranswl{101}{0200}\mathrm{rdff}_f$ (i.e., $f(\{0,2\})= \{1,1\}$).

Similarly, we can define the left direct functional form of $f$ as below:
\[
   \mathrm{ldff}_f = (\rzero \ten\ \id\ \ten \rzero \ten\ \id); (\ldiag  \ten \id);(\id  \ten \tw).
\]
The dynamics of $\mathrm{ldff}_f$ can be interpreted analogously to $\mathrm{rdff}_f$, after swapping the interfaces.
\end{exa}

\subsection{Relational forms}

We now identify two classes of terms of the Petri 
calculus: the left and right \emph{relational forms}. These will
be used in the translation from C/E nets to Petri calculus terms with
strong semantics for representing the functions $\sourcearg{\_}, \targetarg{\_}, \intsourcearg{\_}, \inttargetarg{\_}$. 

\begin{defi}\label{defn:relationalForms}
A term $t\typ\sort{k}{l}$ is in \emph{right relational form} when it is in
\[
T_{\{\rightEnd\}}\comp T_{\{\diag\}}\comp T_{\{\tw\}}\comp 
T_{\{\lcodiag\}}\comp T_{\{\lzero\}}.
\]
Dually, $t$ is said to be in \emph{left relational form} when it is in
\[
T_{\{\rzero\}} \comp T_{\{\ldiag\}}\comp 
T_{\{\tw\}}\comp T_{\{\codiag\}}\comp T_{\{\leftEnd\}}.
\]
\end{defi}

The following result spells out the significance of the relational forms.
\begin{lem}\label{lem:relationalForms}
For each function $f\from \underline{k}\to 2^{\underline{l}}$ there
exists a term $\rho_f\typ\sort{k}{l}$ in right relational form, the dynamics
of which are characterised by the following: 
\[
\rho_f \dtrans{\alpha}{\beta} \rho_f
\ 
\Leftrightarrow
\ 
 \exists U\subseteq \underline{k} \mbox{ s.t. }
\forall u,v\in U.\; u\neq v \Rightarrow f(u)\intersection f(v)=\varnothing,\,
 \alpha=\characteristic{U} \mbox{ and }
 \beta=\characteristic{f(U)} 
\]
The symmetric result holds for functions $f\from \underline{k}\to 2^{\underline{l}}$ and terms
$t\typ\sort{l}{k}$ in left relational form. That is, there exists
$\lambda_f\typ\sort{l}{k}$ in left relational form with semantics 
\[
\lambda_f \dtrans{\alpha}{\beta} \lambda_f
\ 
\Leftrightarrow
\ 
 \exists U\subseteq \underline{k} \mbox{ s.t. }
\forall u,v\in U.\; u\neq v \Rightarrow f(u)\intersection f(v)=\varnothing,\,
 \beta=\characteristic{U} \mbox{ and }
 \alpha=\characteristic{f(U)} 
\]
\end{lem}
\begin{proof}
To give a function $f:\underline{k}\to 2^{\underline{l}}$
is to give functions $f_l:\underline{m}\to\underline{k}$,
$f_r:\underline{m}\to\underline{l}$ such that
$(f_l,f_r):\underline{m}\to \underline{k}\times\underline{l}$
is injective and for any $i<k$, $f(i)=f_r(f_l^{-1}(i))$. 
Let 
\[ 
\rho_f\Defeq \mathrm{riff}_{f_l}\comp \mathrm{rdff}_{f_r}
\mbox{ and }
\lambda_f\Defeq \mathrm{ldff}_{f_r}\comp \mathrm{liff}_{f_l}.
\]
Then the required characterisations follow directly from
the characterisations of inverse and direct functional forms given
in Lemmas~\ref{lem:inverseFunctionalForms} and~\ref{lem:directFunctionalForms}.
\end{proof}

\begin{figure}[t]
\subfigure[$\mathrm{riff}_{f_l}$.]{\label{fig:examplefl}
\includegraphics[height=2cm]{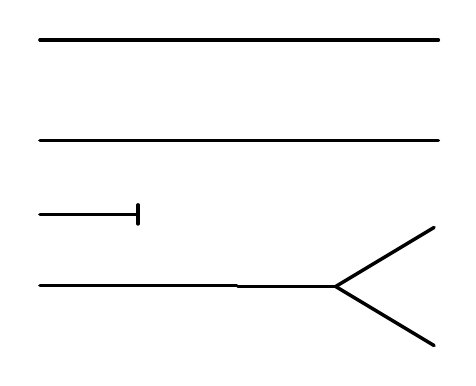}
}
\hspace{1cm}
\subfigure[$\mathrm{rdff}_{f_r}$.]{\label{fig:examplefr}
\includegraphics[height=2cm]{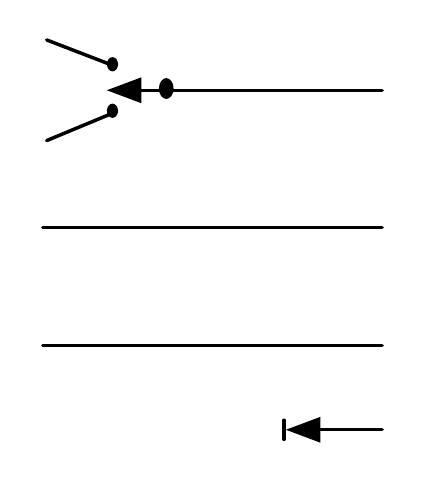}
}
\hspace{1cm}
\subfigure[$\rho_{f}$.]{\label{fig:example}
\includegraphics[height=2cm]{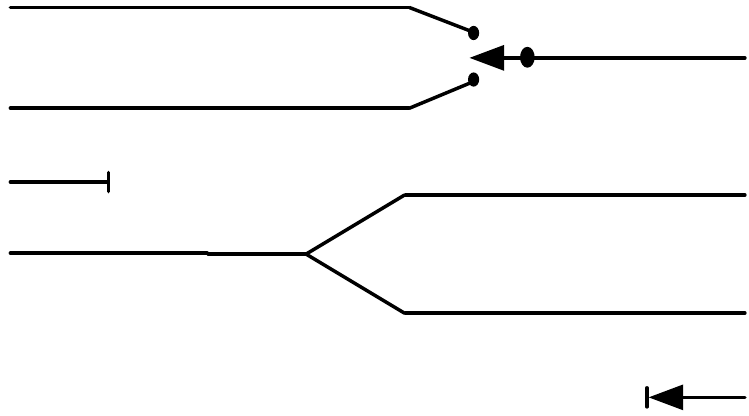}
}\caption{Right relational form.}
\end{figure}

\begin{exa} Let $f\from\underline{4}\to 2^{\underline{4}}$ defined by
$f(0),f(1)=\{0\}$, $f(2)=\varnothing$ and $f(3)=\{1,2\}$. Figure~\ref{fig:example} shows 
the right relational form $\rho_f$ of $f$ that can be obtained, as suggested by proof of Lemma~\ref{lem:relationalForms}, from 
the combination of the functions $f_l\from\ord{4}\to\ord{4}$ and   $f_r\from\ord{4}\to\ord{4}$ defined by
$f_l(0)=0$, $f_l(1)=1$, $f_l(2)=f_l(3) = 3$, $f_r(0)=f_r(1)=0$, $f_r(2) = 1$ and $f_r(3)=2$ (the corresponding 
$\mathrm{riff}_{f_l}$ and $\mathrm{rdff}_{f_r}$ are in \fig{\ref{fig:examplefl}} and \ref{fig:examplefr}, respectively).

Assume now that $f$ above is the postset function of a C/E net consisting on 
four transitions (named $0,1,2,3$) and four places (also named $0,1,2,3$). Intuitively, the term $\rho_f$
accounts for the tokens produced during the execution of a step. The left interface stands for transitions while 
the right interface stands for places. For instance, the 
transition
$\mathrm{\rho}_f \dtranswl{1011}{1110}\mathrm{\rho}_f$ (i.e., $f(\{0,2,3\})= \{0,1,2\}$) stands for tokens 
produced by the simultaneous firing of the transitions $0$, $2$, and $3$ in the 
places $0$, $1$ and $2$. Note that transitions $0$ and $1$ are not independent and, hence, they cannot be fired
simultaneously. This fact is made evident in $\rho_f$ because the ports $0$ and $1$ over the left interface are in mutual exclusion. 

The left relational form $\lambda_f$ can be defined analogously. Dually, $\lambda_f$ can be interpreted 
as a term describing the consumption of tokens during the execution of several mutually independent transitions.
\end{exa}

%A simple example is given in~\fig{\ref{fig:example}}.
Note that not all terms $t\typ\sort{k}{l}$
in right relational form have the behaviour of $\rho_f$ for
some $f\from \underline{k}\to 2^{\underline{l}}$; a simple counterexample is 
$\diag\comp\lcodiag\typ\sort{1}{1}$ whose only reduction 
under the strong semantics 
is $\diag\comp\lcodiag \dtrans{0}{0} \diag\comp\lcodiag$.

\subsection{Contention}
Recall that contention is an irreflexive, symmetric relation on transitions, that restricts the sets of transitions that can be fired together. 

Now consider the term $c$, defined below. 
\[ 
c\Defeq (\diag \otimes \diag)\comp(\id \otimes \tw\otimes \id)\comp(\id\otimes\id\otimes (\lcodiag\comp \rightEnd)).
\] 
It is not difficult to verify that $c$
has its behaviour characterised by the following transitions:
\[
c \dtrans{00}{00} c, \quad c\dtrans{10}{10} c,  \quad c\dtrans{01}{01} c. 
\]

Given a set of transitions $\underline{t}$ and a contention relation $\#\subseteq \underline{t} \times \underline{t}$, here we will define a term
$\#_t : \ord{t}\to\ord{t}$ with semantics: 
\[
{\#}_t \dtrans{\characteristic{U}}{\characteristic{V}} {\#}_t \quad \text{iff} \quad U=V\text{ and }\forall u,v\in U.\; \neg (u\# v)
\]
We define it by induction on the size of $\#$. The base case is when $\#$ is empty, and in this case we let $\#_{t}=\id_{t}$. Otherwise, there exists $(u,v)\in \#$. Let $\#' = \#\backslash\{(u,v),(v,u)\}$.
By the inductive hypothesis we have a term ${\#'}_t:t\to t$ that satisfies
the specification wrt the relation $\#'$. Now the term ${\#'}_t\comp \tw_{u,v} \comp \id_{t-2} \otimes c \comp \tw_{u,v}^{-1}$ has the required behaviour, where $\tw_{u,v}$ is a term in $T_{\{\tw,\id\}}$ that permutes $\underline{t}$, taking $u$ and $v$ to $t-2$ and $t-1$, and $\tw_{u,v}^{-1}$ is its inverse.
\begin{figure}[t]
\[
\includegraphics[height=4cm]{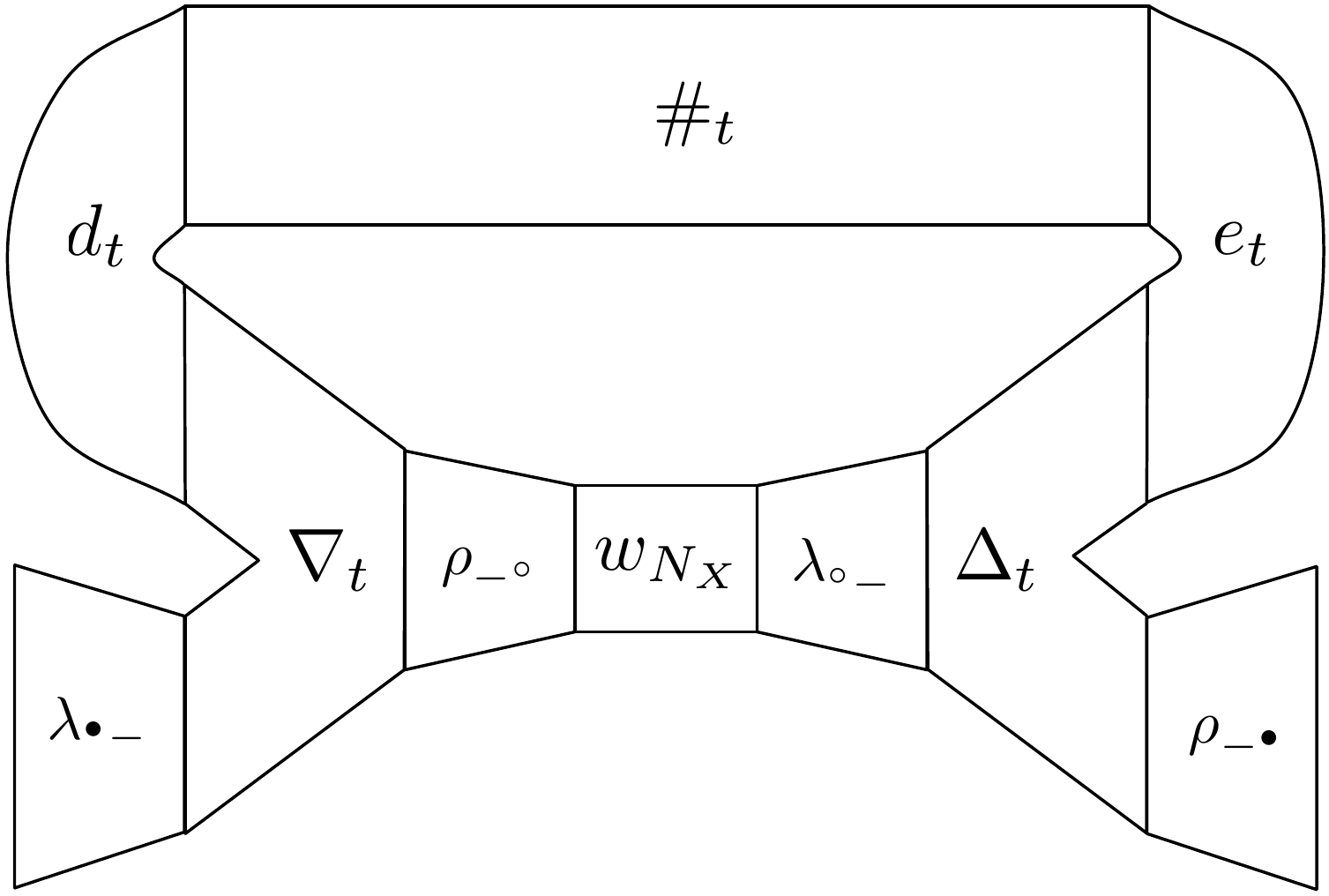}
\]
\caption{Circuit diagrammatic representation of the translation from a C/E net to a Petri calculus term.\label{fig:netsToSyntax}}
\end{figure}

\subsection{Translating C/E nets}
% TRANSLATING NETS TO SYNTAX
Here we 
present a translation from C/E nets with boundaries, defined 
in Section~\ref{sec:nets}, to Petri calculus terms as defined in Section~\ref{sec:syntax}.
%While straightforward, it is  utilising the graphical
%representation of terms introduced in \S\ref{subsec:graphicalRepresentation}.
Let $\marking{N}{X}\from m\to n=(P,\,T,\,X,\,\#,\,\pre{-},\,\post{-},\,\source{-},\,\target{-})$ be a finite 
C/E net with boundary (Definition~\ref{defn:boundedcenets}). Assume, without loss of 
generality, that $P=\underline{p}$ and $T=\underline{t}$ for some $p,t\in\N$.
If $p=0$ we let $w_{N_X}:(0,0)\Defeq \leftEnd \comp \rightEnd$, otherwise
\begin{equation}\label{eq:encodingCEplaces}
w_{\marking{N}{X}}\typ\sort{p}{p}\ \Defeq\  
\bigotimes_{i<p} m_i
\quad
\text{where}
\quad
m_i \Defeq \begin{cases} \tokenplace & \text{if }i\in X \\ \emptyplace & \text{otherwise} \end{cases}
\end{equation}
The following technical result will be useful for showing that the encodings of this
section are correct.

\begin{lem}\label{lem:encodingTech}~
\begin{enumerate}[\em(i)]
\item $w_{\marking{N}{X}}\dtrans{\characteristic{Z}}{\characteristic{W}} Q$ iff
$Q=w_{\marking{N}{Y}}$,
\medskip
$W\subseteq X$, $Z\cap X = \varnothing$ and $Y = (X \backslash W) \cup Z$.
\item $w_{\marking{N}{X}}\dtransw{\alpha}{\beta} Q$ iff 
$Q=w_{\marking{N}{Y}}$ and $X+\alpha = Y + \beta$ as multisets. 
\end{enumerate}
\end{lem}
\begin{proof}
(i) Examination of either 
rules \ruleLabel{$\rightEnd_{1}$} and \ruleLabel{$\leftEnd_{1}$},
together with the rule \ruleLabel{Cut} (when $p=0$) or
rules \ruleLabel{TkI} and \ruleLabel{TkO},
together with the rule \ruleLabel{Ten} (when $p>0$). 

(ii) Combination of (i) with
part (ii) of Lemma~\ref{lem:syntaxdecomposition}.
%
%Suppose that $w_{P,X}\dtrans{\alpha}{\beta} Q$.
%For each $i\in \underline{p}$, if
%$i\in X$ and $\alpha_i=1$ then 
\end{proof}

\noindent The translation of $N$ can now be expressed as:
\begin{equation}\label{eq:strongtranslation}
T_{\marking{N}{X}} \ \Defeq\ 
(d_t\ten \lambda_{\source{-}})\comp
({\#}_t\ten (\codiag_t\comp \rho_{\post{-}} \comp w_{\marking{N}{X}} \comp \lambda_{\pre{-}} \comp \diag_t));
(e_t\ten \rho_{\target{-}}).
\end{equation}
A schematic circuit diagram representation of the above term is illustrated in \fig{\ref{fig:netsToSyntax}}, where 
terms are represented as boxes, sequential composition is the juxtaposition of boxes (to be read from left to right)
and parallel composition is shown vertically (read from top to bottom).

The encoding preserves and reflects semantics in a very tight manner, as shown by the following result.
\begin{thm}\label{thm:strongcenetstosyntax}
Let $N$ be a (finite) C/E net. The following hold:
\begin{enumerate}[\em(i)]
\item if 
$\marking{N}{X} \dtrans{\alpha}{\beta} \marking{N}{Y}$ then 
$T_{\marking{N}{X}} \dtrans{\alpha}{\beta} T_{\marking{N}{Y}}$;
\medskip
\item
 if $T_{\marking{N}{X}}\dtrans{\alpha}{\beta} Q$ then there exists 
$Y$ such that $Q=T_{\marking{N}{Y}}$ and $\marking{N}{X}\dtrans{\alpha}{\beta}\marking{N}{Y}$.
\end{enumerate} 
\end{thm}
\begin{proof}
In Appendix~\ref{app:proof-cetranslations}.
\end{proof}

%TRANSLATING WEAK NETS

\subsection{Translating P/T nets (and weak C/E nets)}
\label{sec:translating-weak-ce}
%AMPLIFICATIONS
We begin by defining right and left amplifiers 
$!k,\,k!\typ\sort{1}{1}$ for any $k\in \N$ that will be necessary in order to 
define multirelational forms, with the latter being needed
to translate P/T nets. 
\begin{defi}[Amplifiers]\label{defn:amplifiers}
Given $k\in \N_+$, the right amplifier $!k\typ\sort{1}{1}$ is defined recursively as follows:
 $!1\Defeq \id$, $!(k+1)\Defeq \diag;(!k\ten \id);\lcodiag$.
Dually, the left amplifier $k!\typ\sort{1}{1}$ is defined:
%$0!\Defeq \leftEnd$, 
$1!\Defeq \id$, $(k+1)!\Defeq \ldiag;(\id\ten k!);\codiag$.
\end{defi}

Notice that under the strong semantics of the Petri calculus,
for any $k>1$, $k!$ and $!k$ have no non-trivial
behaviour (i.e., the only behavior is $k!\dtrans{0}{0}k!$). Instead, the behaviour of a right amplifier $!k$ 
under the weak semantics of the Petri calculus, and in both the strong and weak semantics of the P/T calculus,
intuitively ``amplifies'' a signal $k$ times from left to right.
Symmetrically, a left amplifier $k!$ amplifies a signal $k$ times from right to left.
More formally, their behaviour under the weak semantics of the Petri calculus
and both semantics of the P/T calculus is summarised by the following result.
\begin{lem}\label{lem:amplifiers}
Let $a,b\in\N$. Then
$!k\dtransw{a}{b}!k$ iff $b=ka$. Similarly, $k!\dtransw{a}{b}k!$ iff $a=kb$.
\end{lem}
\begin{proof}
In Appendix~\ref{app:proof-cetranslations}. Note that $k!$ and $!k$ are built from stateless connectors, and thus the weak and strong semantics of the P/T calculus coincide on amplifies. 
\end{proof}

We let $!x\Defeq \{!k\,|\,k\in\N_+\}$ denote the set of right amplifiers and
$x!\Defeq \{k!\,|\,k\in\N_+\}$ the set of left amplifiers.

% Multirelational forms

\subsection{Multirelational forms}
Multirelational forms generalise relational forms and take on their role
in the translation of P/T nets. Whereas relational
forms allow us to encode relations between $\underline{k}$ and $\underline{l}$ 
(or equivalently, functions $\underline{k}\to 2^{\underline{l}}$)
in terms of stateless connectors,
as demonstrated in Lemma~\ref{lem:relationalForms}, multirelational
forms, using the weak semantics of the Petri calculus, allow 
us to encode multirelations between
$\underline{k}$ and $\underline{l}$ (equivalently, functions 
$\underline{k}\to \multiset{\underline{l}}$).  
\begin{defi}\label{defn:multirelationalForms}
A term $t\typ\sort{k}{l}$ is in \emph{right multirelational form} when it
is in
\[
T_{\{\rightEnd\}}\comp T_{\{\diag\}}\comp T_{\{\tw\}} \comp T_{!x} 
	\comp T_{\{\tw\}} \comp T_{\{\lcodiag\}}\comp T_{\{\lzero\}}.
\]
Dually, $t$ is said to be in \emph{left multirelational form} when it is in
\[
 T_{\{\rzero\}} \comp T_{\{\ldiag\}} \comp T_{\{\tw\}} 
	\comp T_{x!}  \comp  T_{\{\tw\}}\comp T_{\{\codiag\}}\comp T_{\{\leftEnd\}}.
\]
\end{defi}

%Multirelational forms lemma
\begin{lem}\label{lem:multirelationalForms}
For each function $f\from \underline{k}\to \multiset{\underline{l}}$ 
there exists a term $\boldsymbol{\rho}_f\typ\sort{k}{l}$ in right 
multirelational form, the dynamics of which are characterised by 
the following: 
\[
\boldsymbol{\rho}_f \dtransw{\alpha}{\beta} \boldsymbol{\rho}_f
\ 
\Leftrightarrow
\
 \exists\, \mathcal{U}\in\multiset{\underline{k}} \text{ s.t. }
 \alpha = \characteristic{\mathcal{U}} \mbox{ and }
 \beta = \characteristic{f(\mathcal{U})}.
\]
The symmetric result holds for functions 
$f\from \underline{k}\to \multiset{\underline{l}}$ and terms
$t\typ\sort{l}{k}$ in left relational form. 
That is, there
exists a term $\boldsymbol{\lambda}_f\typ\sort{l}{k}$ 
in left multirelational form so that 
\[
\boldsymbol{\lambda}_f \dtransw{\alpha}{\beta} \boldsymbol{\lambda}_f
\ 
\Leftrightarrow
\
 \exists\, \mathcal{U}\in\multiset{\underline{k}} \text{ s.t. }
 \beta = \characteristic{\mathcal{U}} \mbox{ and }
 \alpha = \characteristic{f(\mathcal{U})}.
\]
\end{lem}
\begin{proof}
To give a function $f:\underline{k}\to\multiset{\underline{l}}$
is to give functions $f_l:\underline{m}\to\underline{k}$,
$f_r:\underline{m} \to \underline{l}$,
$f_m:\underline{m}\to \N$ 
with $(f_l,f_r):\underline{m}\to\underline{k}\times \underline{l}$
injective,
so that, for all $i<k$, $j<l$
\begin{equation}\label{eq:multirelations}
f(i)_j = \begin{cases} 
f_m(u) & \mbox{ if }\exists u < m.\, f_{l}(u)=i \mbox{ and }f_{r}(u)=j \\ 
0 & \mbox{ otherwise.}
\end{cases} 
\end{equation}
Notice that the above makes sense because $(f_l,f_r)$ is injective, that is,
if there exists $u$ that satisfies the first premise in~\eqref{eq:multirelations}
then it is the unique such element. 
We let 
$\boldsymbol{\rho}_f \Defeq \mathrm{riff}_{f_l} 
\comp (\bigotimes_{i<m}!f_m(i)) \comp \mathrm{rdff}_{f_r}$
and
$\boldsymbol{\lambda}_f \Defeq 
\mathrm{ldff}_{f_r} \comp (\bigotimes_{i<m}f_m(i)!)  \comp
\mathrm{liff}_{f_l}$.
For the Petri calculus with weak semantics, the required characterisation then follows from the
weak cases of Lemmas~\ref{lem:inverseFunctionalForms} 
and~\ref{lem:directFunctionalForms},
together with the conclusion of Lemma~\ref{lem:amplifiers}. For both
the strong and the weak semantics of P/T calculus it follows since in both
cases those semantics agree with the weak semantics of the Petri calculus
on stateless connectors.
\end{proof}

Recall that by restricting P/T nets with the banking semantics 
to markings that are merely sets we obtain
a class of nets that we call weak C/E nets (Remark~\ref{rmk:weakCENets}).
Let $\marking{N}{X}\from m\to n=(P,\,T,\,\pre{-},\,\post{-},\,\source{-},\,\target{-})$ be a 
finite weak C/E net with $X\subseteq P$ a marking.
Recall that 
\[
\pre{-}\from T\to \multiset{P},\, \post{-}\from T\to \multiset{P},\,
\source{-}\from T \to \multiset{\underline{m}} \mbox{ and }
\target{-}\from T \to \multiset{\underline{n}}.
\]
The translation from $N$ to the Petri calculus is as given 
in~\eqref{eq:strongtranslation}, with multirelational forms
replacing relational forms. Let $\weaktranslation{\marking{N}{X}}$ denote
the obtained Petri calculus term. 

We are now ready to state the semantic correspondence of the translation from
weak nets to Petri calculus terms. 
%% Weak translation theorem
\begin{thm}\label{thm:weakcenetstosyntax}
Let $N$ be a finite weak C/E net. The following hold:
\begin{enumerate}[\em(i)]
\item if 
$\marking{N}{X} \dtransw{\alpha}{\beta} \marking{N}{Y}$ then 
$\weaktranslation{\marking{N}{X}} \dtransw{\alpha}{\beta} 
\weaktranslation{\marking{N}{Y}}$.
\item
 if $\weaktranslation{\marking{N}{X}}
\dtransw{\alpha}{\beta} Q$ then there exists 
$Y$ such that $Q=\weaktranslation{\marking{N}{Y}}$ and 
$\marking{N}{X}\dtransw{\alpha}{\beta}\marking{N}{Y}$. 

\end{enumerate} 
\end{thm}

\begin{proof}
The proof closely follows
the proof of Theorem~\ref{thm:strongcenetstosyntax}.
\end{proof}

To translate general P/T nets with both the standard and the banking semantics, we move
to the P/T calculus with, respectively, the strong and the weak semantics.
We start by introducing the encoding of a marking.
Again w.l.o.g. we assume that places and  transitions of a P/T are given by sets of ordinals,  \ie, 
  $P = \ord{p}$ and $T = \ord{t}$ for some $p,t\in \nat$. Let $N = (\ord{p}, \ord{t}, \pre{-}, \post{-}, \intsource, \inttarget)$ be a P/T net with boundaries.
 Then, a marking $\mathcal{X}\in\mathcal{M}_{P}$ 
  is encoded as follows
 \begin{equation}\label{eq:encodingPTplaces}
 \textstyle
  \encodemark{\mathcal{X}}: (p, p)  = \Big\{
  \begin{array}{ll}
     \leftEnd;\rightEnd  & \text{if}\ p = 0\\
     \bigotimes_{i<p} \tokensplace{\mathcal{X}(i)} & \text{otherwise}
  \end{array}
\end{equation}

Clearly the encoding of markings of P/T nets in~\eqref{eq:encodingPTplaces} is
similar to the encoding of markings of C/E nets in~\eqref{eq:encodingCEplaces}. 
The following technical result  is used to prove the correctness of the 
proposed encoding. 

\begin{lem}
 Let $\mathcal{X} \in \mathcal{M}_{\ord{p}}$. 
\begin{enumerate}[\em(i)]
 \item \label{lemma:encoding-marking-strong}
 $\cell{\encodemark{\mathcal{X}}}{\characteristic{\mathcal{Z}}}{\characteristic{\mathcal{W}}} Q$ iff $Q = \encodemark{\mathcal{Y}}$,
 $\mathcal{W}\subseteq\mathcal{X}$, $\mathcal{Z}\subseteq\mathcal{Y}$ and $\mathcal{X}-\mathcal{W} = \mathcal{Y}-\mathcal{Z}$.

\medskip

 \item \label{lemma:encoding-marking-weak}
 ${\encodemark{\mathcal{X}}}\dtranswl{\characteristic{\mathcal{Z}}}{\characteristic{\mathcal{W}}} Q$ iff $Q = \encodemark{\mathcal{Y}}$,
 and $\mathcal{X}+\mathcal{Z} = \mathcal{Y}+\mathcal{W}$.
 \end{enumerate}
\end{lem}

\begin{proof} $(i)$ The proof follows by induction on $p$. Base case (${p = 0}$) follows immediately because $\cell{\leftEnd;\rightEnd}{}{}{\leftEnd;\rightEnd}$ is the 
only allowed reduction for $\leftEnd;\rightEnd$. Inductive step follows by inductive hypothesis, rule~\ruleLabel{Ten},  and Lemma~\ref{lemma:place-inv-strong}.

$(ii)$ Proof follows analogously to case $(i)$ but we use Lemma~\ref{lemma:place-inv-weak} for the inductive step.
\end{proof}

As for the translation of weak C/E nets, the definition of the encoding relies on the definitions of terms $\boldsymbol{\rho}_f\typ\sort{k}{l}$ and $\boldsymbol{\lambda}_f\typ\sort{l}{k}$ (the right and left multirelational forms of function $f\from \underline{k}\to \multiset{\underline{l}}$) as introduced in  Lemma~\ref{lem:multirelationalForms}. The translation of P/T with boundaries 
 to a P/T calculus terms is as given in~\eqref{eq:strongtranslation}, with multirelational forms replacing the corresponding relational forms. 
 We write ${\weaktranslation{\marking{N}{\mathcal{X}}}}$ (instead of $T_{N_{\mathcal{X}}}$) to highlight the usage of multirelational forms.
 
The proofs of next results closely follows the proof of Theorem~\ref{thm:strongcenetstosyntax}
and are omitted.

\begin{thm}[Strong]
\label{theorem:correspondence-pt-into-tiles-strong}
Let $N$ be a finite P/T net with boundaries, then
\begin{enumerate}[\em(i)]
  \item if $\cell{N_\mathcal{X}}{\alpha}{\beta}{N_\mathcal{Y}}$	then  $\cell{\weaktranslation{\marking{N}{\mathcal{X}}}}{\alpha}{\beta}{\weaktranslation{\marking{N}{\mathcal{Y}}}}$.
  \item if $\cell{\weaktranslation{\marking{N}{\mathcal{X}}}}{\alpha}{\beta} Q$  then $\cell{N_{\mathcal{X}}}{\alpha}{\beta}{N_{\mathcal{Y}}}$ and $Q = \weaktranslation{\marking{N}{\mathcal{Y}}}$. 
\end{enumerate}
\end{thm}

\begin{thm} [Weak]
\label{theorem:correspondence-pt-into-tiles-weak}
Let $N$ be a finite P/T net with boundaries, then
\begin{iteMize}{$-$}
  \item if ${\marking{N}{\mathcal{X}}}\dtransw{\alpha}{\beta}{N_{\mathcal{Y}}}$	then  ${\weaktranslation{\marking{N}{\mathcal{X}}}}\dtransw{\alpha}{\beta}{\weaktranslation{\marking{N}{\mathcal{Y}}}}$.
  \item if $\weaktranslation{\marking{N}{\mathcal{X}}}\dtransw{\alpha}{\beta} Q$ then ${N_{\mathcal{X}}}\dtransw{\alpha}{\beta}{N_{\mathcal{Y}}}$ and $Q = \weaktranslation{\marking{N}{\mathcal{Y}}}$. 
\end{iteMize}
\end{thm}
 
\begin{exa}
We now exhibit the P/T terms that encode the behaviour of the nets in Fig.~\ref{fig:ex-simple-pts} (for simplicity we show terms that are bisimilar to the ones generated by the encoding, but simpler). 
The term $T_1 \Defeq \,\cozero ; \tokensplace{0} ; \diag ; \lcodiag \typ \sort{0}{1}$ encodes the behaviour of the place $\mathtt{a}$ and the transition $\alpha$. 
%Note that this part of the %net $M$ has no 
%attachments to the left boundary and one attachment to the right boundary. This %corresponds with the 
%sort of $T_1$, which is $T_1 \typ\sort{0}{1}$. 
Intuitively, the term $\cozero$ in $T_1$ represents the transitions that can produce tokens into the place $\mathtt{a}$, which in this case is the empty set.  Analogously, 
$ \dup ; \come$ describes the transitions that can consume tokens from place $\mathtt{a}$,
in this case, this is the right amplifier $!2$ (see Lemma~\ref{lem:amplifiers}).
%It is easy to check that 
%the unique non-trivial reduction of $ \dup ; \come$ is $\dup ; \come\dtrans{1}{2}\dup ; \come$, which actually describes the behavior of the transition $\alpha$, which fetches one token from place $\mathtt{a}$ and produces two tokens over the right  boundary. 
%
The term 
$T_2 \Defeq \,\cozero ; \tokensplace{0} ; \ldiag ; \codiag$ corresponds to the part of $M$ containing place $\mathtt{b}$ and the transition  $\beta$. Then, the complete net $M$ can be translated into $T \Defeq (T_1 \otimes T_2)$. As for $T_1$ it is easy to check that $T\typ\sort{0}{2}$, which coincides with the boundaries of $M$. 

Similarly, we can obtain the encoding of $N$ as follows 
$U \Defeq (3! \ten \id ) \comp \codiag \comp \tokensplace{0} ; \rzero$.
Finally, the term for $M;N$ in 
Fig.~\ref{fig:ex-composition-of-two-nets} is bisimilar to
$((\cozero ; \tokensplace{0} \comp !2) \ten  (\cozero ; \tokensplace{0} ; 2!))\comp (3! \ten \id ) \comp \codiag \comp \tokensplace{0} ; \rzero$.
%$N = (\me^3_1 \otimes \id) ; \codup^4_1 ; \tokensplace{0} ; \zero$.  
%Here the sub term  $(\me^3_1 \otimes \id) ; \codup^4_1$ stands for  the transition that produces tokens into the buffer (corresponding to the place $\mathtt{c}$. It is easy to check that  $(\me^3_1 \otimes \id) ; \codup^4_1 \dtrans{31}{1} (\me^3_1 \otimes \id) ; \codup^4_1$, which corresponds to the behavior of $\gamma$. 
%Finally, the term for the net $M;N$ is bisimilar to 
%$((\cozero ; \tokensplace{0} ; \me^3_1)
%\otimes
%(\cozero ; \tokensplace{0} ; \me^4_1)) ; \codup^7_1 ; \dup ; \come;\tokensplace{0} ; \zero$, which corresponds to the encoding of the net in 
%Fig.~\ref{fig:ex-composition-of-two-nets}.
\end{exa}

\section{Petri Tile Calculus}
\label{sec:petri-tile-calculus}

%The P/T calculus we have presented captures quite nicely the behaviour of P/T nets with boundaries, but it is not as ``atomic'' as the Petri calculus, because it fails to model tokens as independent resources within a place. In fact, 

While in the Petri calculus we have just two possible states for each place (empty or full), in the P/T calculus we have a denumerable set of constants $\tokensplace{n}$, one for each $n \in \nat$. Correspondingly, the rules of the P/T calculus are actually schemes of rules, parametric to the number of tokens that are observed in one step.

%Moreover, in the strong case, where the rule \ruleLabel{Weak} is not used, the rules for stateless connectors are actually schemes of rules, parametric to the number of tokens that are observed in the step.

In this section we show that we can further decompose the P/T calculus to expose the minimal units of computation while preserving the correspondence to P/T nets with boundaries. Furthermore, the ability to do so provides a technical answer to the long standing quest for the algebra of P/T nets (see Section~\ref{sec:related}), where boundaries are key instruments to achieve compositionally.

%In this section we show that this can be obtained using the tile model. 
Technically, we present the rules of the operational semantics as \emph{tiles} and exploit the monoidality law of the \emph{tile model}~\cite{DBLP:conf/birthday/GadducciM00} to give a finitary presentation of P/T nets with boundaries. In the weak case the tile model arises as the straightforward generalisation of the Petri calculus to account for unbounded buffers. In the strong case some ingenuity is needed to avoid  computations that consume tokens before being produced. 
%In both cases the result relies on the well-known monoidality laws of tiles.
Moreover, in both cases we can exploit standard machinery from the theory of tile systems to prove that bisimilarity is a congruence w.r.t. sequential and parallel composition just by noting that the basic tiles we start from adhere to a simple syntactic format, called basic source.

We start by overviewing the basics of the tile model, then presenting the Petri tile calculus and finally proving the correspondence with the P/T calculus (and, by transitivity, with P/T nets with boundaries).

\subsection{The tile model}
\label{sec:tiles}

Roughly, the  semantics of (concurrent) systems can be expressed via
 tiles when: 
 i)~system configurations $s$ are equipped with input/output interfaces, 
 written $s: w_{i}\to w_{o}$ for $w_{i}$ the input interface of $s$ and $w_{o}$ the output interface of $s$,
 with special configurations $id_{w}: w \to w$ called \emph{identities} for each interface $w$;
 ii)~system configurations are equipped with a notion of sequential composition $s;t$ 
 (defined when the output interface of $s$ matches the input interface of $t$) such that
 $id_{w_{i}};s = s = s;id_{w_{o}}$ for each $s: w_{i} \to w_{o}$;
  iii)~system configurations are equipped with a distinguished unit element $id_{\epsilon}: \epsilon \to \epsilon$
  and with a monoidal tensor product $s \otimes t$ that is associative, has the unit $id_{\epsilon}$ as neutral element
  and distributes over sequential composition 
  (i.e., $(s;t) \otimes (s';t') = (s \otimes s');(t\otimes t')$ whenever both sides are defined); 
  iv)~observations have analogous structure $id_{w}$, $a;b$ and $a\otimes b$;
 v)~the interfaces of configurations and of observations are the same.
 Technically, the above requirements impose that configurations and observations form two monoidal categories
 called, respectively,
 $\mathcal{H}$ (form horizontal) and $\mathcal{V}$ (from vertical) with the same underlying set of objects. 
 
%The \emph{tile model} has been introduced in~\cite{DBLP:conf/birthday/GadducciM00}.
%
% In this paper, we choose the Tile Model for defining the
% operational and abstract semantics of the suitable class of 
% connectors consisting of nets with boundaries.  
%
% In fact, tile configurations are particularly suitable to represent  
% the concept of connector, which includes input and output 
% interfaces where actions can be observed and that can be used 
% to organize connectors in series and parallel and coordinate 
%their local behaviors.
%
 A tile $A:\cell{s}{a}{b}{t}$ is a rewrite rule stating that the
 \emph{initial configuration} $s$ can evolve to the \emph{final
 configuration} $t$ via $A$, producing the \emph{effect} $b$; but the
 step is allowed only if the `arguments' of $s$ can contribute by
 producing $a$, which acts as the \emph{trigger} of $A$ (see
 Fig.~\ref{fig:threecomptile}(i)). Triggers and effects are 
 \emph{observations} and tile vertices are called \emph{interfaces}.
The similarity between tile shapes and that of  the structural rules 
for the operational semantics used throughout the paper is evident.

\begin{defi}[Tile system]
  A \emph{tile system} is a tuple
  $\mathcal{R}=(\mathcal{H},\mathcal{V},N,R)$ where $\mathcal{H}$ and
  $\mathcal{V}$ are monoi\-dal categories over the same set of objects,
%  $O_{\mathcal{H}} = O_{\mathcal{V}}$,
  $N$ is the set of rule names
  and
  $R\from N\to\mathcal{H}\times\mathcal{V}
               \times\mathcal{V}\times\mathcal{H}$
 is a function such that for all $A\in N$,
 if $R(A)=\langle s,a,b,t\rangle$, then the sources and targets of  $s,a,b,t$ 
 match as in Fig.~\ref{fig:threecomptile}(i).
\end{defi}

\begin{figure}[t]
\[
\mbox{(i)}
 \vcenter{\xymatrix@R-1pc@C-1pc{
 {w_1} \ar[r]^{s} \ar[d]_{a} \ar@{}[rd]|{A} &
 {w_2} \ar[d]^{b}\\
 {w_3} \ar[r]_{t} &
 {w_4}}}
 \qquad\qquad
 \mbox{(ii)}
 \vcenter{\xymatrix@R-1pc@C-1pc{ {\cdot } \ar[r] \ar[d]
 \ar@{}[rd]|{A} & {\cdot } \ar[r] \ar[d] \ar@{}[rd]|{B} &
 {\cdot} \ar[d] \\ {\cdot} \ar[r] & {\cdot} \ar[r]
 & {\cdot} }}
\qquad
\mbox{(iii)}
\vcenter{\xymatrix@R-1pc@C-1pc{ {\cdot} \ar[r] \ar[d]
 \ar@{}[rd]|{A} & {\cdot} \ar[d] \\ {\cdot} \ar[r] \ar[d]
 \ar@{}[rd]|{B} & {\cdot} \ar[d] \\ {\cdot} \ar[r]
 & {\cdot} }}
\qquad
\mbox{(iv)}
\vcenter{\xymatrix@R-2pc@C-1pc{ & {\cdot} \ar[rr] \ar[dd] & & {\cdot}
 \ar[dd] \\ {\cdot} \ar[rr] \ar[dd] & & {\cdot} \ar[dd]
 \ar@{}[ru]|{B} \\ & {\cdot} \ar[rr] & & {\cdot} \\ {\cdot}
 \ar[rr] \ar@{}[ru]|{A} & { }
 & {\cdot} }}
\]
\caption{Examples of tiles and their composition.}
\label{fig:threecomptile}
\end{figure}

Like rewrite rules in rewriting logic, tiles can be seen as sequents
of \emph{tile logic}: the sequent $\cell{s}{a}{b}{t}$ is
\emph{entailed} by the tile system $\mathcal{R}=(\mathcal{H},\mathcal{V},N,R)$,
% written $\mathcal{R}\vdash\cell{s}{a}{b}{t}$, 
if it can be obtained by
horizontal, parallel, and vertical composition of some basic tiles in
$R$ plus  identity tiles
$\cell{id_{w_{i}}}{a}{a}{id_{w_{o}}}$  and
$\cell{s}{id_{w_{i}}}{id_{w_{o}}}{s}$.
% (and possibly other auxiliary tiles).
% such as horizontal symmetries 
% $\cell{\gamma}{a\otimes b}{b \otimes a}{\gamma}$ which swap
% the order in which concurrent observations are attached to the
% left and right interfaces). 
%
The ``borders'' of composed sequents are defined in
Fig.~\ref{fig:idtile}.
% Tiles can be composed horizontally, in parallel, or vertically.
 The horizontal composition coordinates the evolution of
 the initial configuration of $A$ with that of $B$, 
  `synchronising'  their rewrites  (see Fig.~\ref{fig:threecomptile}(ii)). This rule is 
  analogous to the rule~\ruleLabel{Cut} of the Petri and P/T calculi.
% Horizontal composition is possible only if the effect of $A$  provides the
% right trigger for $B$.  
The vertical composition is the
 sequential composition of computations (see Fig.~\ref{fig:threecomptile}(iii)).   This rule is 
  analogous to the rule~\ruleLabel{Weak*} of the Petri and P/T calculi.
 The parallel composition
 builds wider steps (see Fig.~\ref{fig:threecomptile}(iv)), as if the steps $A$ and $B$ were computed concurrently, side by side. Parallel composition corresponds to rule~\ruleLabel{Ten} of the Petri and P/T calculi.

\begin{figure}[t]
{\footnotesize
\begin{tabular*}{\textwidth}{@{\extracolsep{\fill}}lcr}
 $\infer[\ruleLabel{Hor}]{\cell{s;h}{a}{c}{t;f}}{\cell{s}{a}{b}{t}\quad\cell{h}{b}{c}{f}}$ &
 $\infer[\ruleLabel{Vert}]{\cell{s}{a;c}{b;d}{h}}{\cell{s}{a}{b}{t}\quad\cell{t}{c}{d}{h}}$ &
 $\infer[\ruleLabel{Mon}]{\cell{s\otimes h}{a\otimes c}{b\otimes d}{t\otimes f}}{\cell{s}{a}{b}{t}\quad\cell{h}{c}{d}{f}}$ 
\end{tabular*}
}
\caption{Main inference rules for tile logic.}
\protect\label{fig:idtile}
\end{figure}

 Tiles express the reactive behaviour of
 configurations in terms of $\mathrm{trigger}+\mathrm{effect}$ 
 labels. In this context, the
 usual notion of bisimilarity is called \emph{tile bisimilarity} ($\mathrel{\simeq_{\mathsf{tb}}}$).

\begin{defi}[Tile bisimilarity]\label{defin:bisim}
 Let $\mathcal{R} = (\mathcal{H},\mathcal{V},N,R)$ be a tile system. A
 symmetric relation $S$ on configurations is called
 a \emph{tile bisimulation} if whenever $(s,t)\in S$ and
 %$\mathcal{R}\vdash $
 $\cell{s}{a}{b}{s'}$, then $t'$ exists such that 
 %$\mathcal{R}\vdash$ 
 $\cell{t}{a}{b}{t'}$ and $(s',t')\in S$.
 The largest tile bisimulation is called \emph{tile bisimilarity} and
 it is denoted by $\simeq_{\mathsf{tb}}$.
\end{defi}

 Note that
 $s\mathrel{\simeq_{\mathsf{tb}}} t$ only if $s$ and $t$ have the same
 input-output interfaces. 

\subsection{Petri Tile Calculus}

The categories $\mathcal{H}$ and $\mathcal{V}$ are typically those freely generated from some
 (many-sorted, hyper-) signatures $\Sigma_{\mathcal{H}}$ and $\Sigma_{\mathcal{V}}$ over the same set of sorts.
 In the case of unsorted signatures, we denote objects of $\mathcal{H}$ and $\mathcal{V}$
 just as natural numbers and thus we freely generate $\mathcal{H}$ and $\mathcal{V}$ starting from families of  symbols $f \from n \to m$, also written $f:(n,m)$, with $n,m\in \nat$.
%  to stick with the notation used in the Petri and P/T calculi.
% Their objects are words on the
% sort alphabet $S$, which must be common to the signatures $\Sigma_{\mathcal{H}}$ and $\Sigma_{\mathcal{V}}$, i.e., $w_{i},w_{o}\in S^{*}$.
 %The symbols $\sigma\in\Sigma$ are seen as basic configurations
 %with source and target defined according to the domain and codomain of $\sigma$.
  Identity arrows $id_{n}: n \to n$ (and possibly other auxiliary arrows)
 %like symmetries 
 %$\gamma_{\tau_{1},\tau_{2}} \from \tau_{1}\tau_{2}\to \tau_{2}\tau_{1}$ that
 %rearrange the order of the sorts in the interface) 
 are introduced  by the free construction.
 
We are going to exploit one fundamental algebraic law of the tile model, namely the functoriality of the monoidal product 
imposed by the free construction, according to which we have  
$(f \otimes g) ; (f'\otimes g') = (f;f') \otimes (g;g')$
for any arrows (either all configurations or all observations) $f,f',g,g'$ such that $f;f'$ and $g;g'$ are well-defined.
In particular, for $a:n \to m$ and $a':n' \to m'$, we have
$(id_{n} \otimes a');(a \otimes id_{m'}) = 
a \otimes a' = 
(a \otimes id_{n'});(id_{m} \otimes a')$.

\subsubsection*{Configurations.}
We take as horizontal signature (i.e., to represent the states of the system) the set of stateless connectors together with one constant for the empty place  $\tokensplace{0}:(1,1)$ and one constant for tokens $\token:(0,1)$. The set of horizontal configurations is  just the free monoidal category generated from this signature. 

More concretely, we can equivalently define  the syntax of the Petri Tile Calculus with the BNF below.

\begin{eqnarray*}
P & \bnfEq & \tokensplace{0} \bnfSep \token
\bnfSep \tw
\bnfSep \diag \bnfSep \codiag 
\bnfSep \rightEnd \bnfSep \leftEnd 
\bnfSep \ldiag \bnfSep \lcodiag
\bnfSep \rzero \bnfSep \lzero \bnfSep \\
& & \id \bnfSep P\ten P \bnfSep P \comp P
\end{eqnarray*}

Note however that the signature of configuration only consists of the symbol in the first line, while the items in the second line are added automatically by the construction of the free monoidal category:
 the constant $\id : (1,1)$ is an auxiliary arrow (it is the identity $id_{1}$ of the category) and the parallel and sequential composition are subject to the axioms of monoidal categories.
 The identity $id_{0}$ is the neutral element of parallel composition.

%*** Should we list the laws of monoidal categories in appendix? ***

Places containing several tokens are defined as the combination of the constants  $\tokensplace{0}$ and $\token$ as follows:  we let $\addtoken \defeq (\id \otimes \token);\lcodiag$ and define
 $\tokensplace{n+1} \defeq \tokensplace{n};\addtoken = (\tokensplace{n}\otimes \token);\lcodiag$ for any $n\geq 0$,
 where the last equality is due to the identity law and functoriality of the tensor $\otimes$ in the monoidal category of configurations:
\[
\tokensplace{n};\addtoken =
\tokensplace{n};(\id \otimes \token);\lcodiag =
(\tokensplace{n}\otimes id_{0});(\id \otimes \token);\lcodiag =
((\tokensplace{n};\id) \otimes (id_{0};\token));\lcodiag =
(\tokensplace{n} \otimes \token);\lcodiag
\]
% For instance,  $\tokensplace{2}:(1,1) = (((\tokensplace{0}\otimes \token);\lcodiag)\otimes \token);\lcodiag$.
Roughly, $\tokensplace{n}$ can be seen as a cluster made of one instance of $\tokensplace{0}$ and  $n$ instances of $\token$, all connected via a ``tree'' of $\lcodiag$ symbols.

\subsubsection*{Observations.}
We take as vertical signature (i.e., for the actions observed over the interfaces) the unary symbols $1:(1,1)$ and $\tau:(1,1)$.
The former is used to represent observed tokens, the latter is used as a separator of sequences of tokens in \emph{epochs}, so that tokens from different epochs cannot interfere with each other.
 The set of vertical observations is  just the free monoidal category generated from this signature. 
 Roughly, an epoch is an observation that does not involve any $\tau$. 
 Since the symbols of the signature are unary, for any observation $a:(n,m)$ we have $n=m$ and moreover we can express $a$ as the parallel product $a_{1}\otimes ... \otimes a_{n}$ of suitable $n$ ``unary'' observations $a_{1}:(1,1), ..., a_{n}:(1,1)$.
  Note that, as a tile system, here we let the observation $0$ (i.e., the absence of a token) be the identity $id_{1}$, so that, e.g. $0;1 = 1 = 1;0$.
We also let $n+1\defeq n;1$ for any $n>0$. This definition characterises the fact that the  consumption/production of several tokens over the same interface is actually serialised, i.e., an observation of $n$ tokens over an interface  corresponds to $n$ sequential steps that observe one token each. The fact that $0$ is the identity $id_{1}$, and e.g. $1 = 0; 1 = 0; 0; 1= ...$  intuitively means that the particular step in which a token is observed is irrelevant. Therefore, we can anticipate/postpone the observation of a token as needed. As a consequence, we lose the possibility to express causal dependency  between observations. For instance, we equate as equivalent the following three tiles $\cell{s}{1;0}{0;1}{s}$ (i.e., $s$ first receives a token and then produces one),  $\cell{s}{0;1}{1;0}{s}$ (i.e., $s$ first produces  a token and then consumes one),  $\cell{s}{1}{1}{s}$ (i.e., $s$ simultaneously consumes and produces  a token). This feature is suitable for the banking semantics, but it is insufficient for dealing with the strong case since, e.g., we expect an empty place not to produce a token before receiving it.  For this reason we have introduced a novel observation $\tau$ to  separate computation steps in epochs: tokens can be rearranged along the steps within the same epoch but cannot be moved across the delimiter(s) $\tau$ of the epoch they belongs to.  For instance, the observation $a = 2;\tau;0;\tau;1$ has three epochs: the first has two tokens observed, the second has no token observed and the third has one token observed. Tokens within each of these three epochs can be rearranged in different sequential steps as needed, e.g., $a = 1;1;\tau; 0;\tau; 1 = 1;0;1;\tau; 0;0;\tau; 0;1;0 =\ldots $.
But tokens cannot be rearranged across epochs, e.g., $a \neq 1;\tau;1;\tau;1$.
The tile model we will focus on corresponds to the strong case, but the weak case can be recovered by taking $\tau = id_{1}$ (i.e., removing the observation of epochs).

In the following, we let $\bar{n} \defeq \tau;n$ (i.e., $\bar{n}$ has one epoch with $n$ tokens) and, for any $a:(k,k)$, $\bar{a} \defeq \tau^{k};a$, where we recall that $\tau^{k}$ is the parallel composition of $\tau$ for $k$ times.  
As a special case we have $\bar{0} = \tau;0 = \tau;id_{1}= \tau$. 
Sometimes we will need to ``slice'' observations along epochs.
The following results provides a canonical representation for observations. 

\begin{lem}
For any observation $a:(1,1)$ there exist unique $k, n_{1},...,n_{k}\in \nat$ such that  
%\begin{iteMize}
%\item $a = \overline{n_{1}}; \overline{n_{2}}; ... ; \overline{n_{k}}$ or
%\item 
$a = n_{1}; \overline{n_{2}}; ... ; \overline{n_{k}}$.
%\end{iteMize}
\end{lem}
\begin{proof}
By structural induction on $a$.
\end{proof}

As a corollary, for any observation $a:(1,1)$ there exist unique $k, n_{1},...,n_{k}\in \nat$ such that $\bar{a} = \overline{n_{1}}; \overline{n_{2}}; ... ; \overline{n_{k}}$. 

\subsubsection*{Terminology and notation.}
As already said, we call \emph{epoch} any $\tau$-free observation.
We say that $a:(h,h)$ is \emph{valid} if there exist $k\in \nat$ and epochs $a_{1}:(h,h),...,a_{k}:(h,h)$ such that $\bar{a} = \overline{a_{1}}; \overline{a_{2}}; ... ; \overline{a_{k}}$. Note that if $a$ is valid, then there is a unique such $k$ (because the decomposition ``aligns'' all $\tau$ separating one epoch from another), in which case we say that $a$ has \emph{age} $k$, written $\ageof{a} = k$.
We say that $a$ is \emph{elementary} if $\ageof{a} = 1$ and 
 that $a$ and $b$ are  \emph{coetaneous} if  $\ageof{a} = \ageof{b}$ (i.e., if they have the same age).
 It is obvious by definition that the relation of being coetaneous is reflexive, commutative and transitive, i.e.,
 it is an equivalence relation.
 For example $1$ and $5$ are coetaneous, while $1$ and $\overline{5};\overline{3}$ are not coetaneous.
Slightly abusing the notation, we say that the empty observation $id_{0}:(0,0)$ is coetaneous to any valid $a$.

For two coetaneous $a:(1,1)$ and $b:(1,1)$ such that 
$\bar{a} = \overline{n_{1}}; \overline{n_{2}}; ... ; \overline{n_{k}}$ and
$\bar{b} = \overline{m_{1}}; \overline{m_{2}}; ... ; \overline{m_{k}}$ 
we let $a+b \defeq  n_{1}+m_{1}; \overline{n_{2}+m_{2}}; ... ; \overline{n_{k}+m_{k}}$.

We say that $a:(1,1)$ is \emph{idle} if there exists $k\in \nat$ such that $\bar{a} = \underbrace{\tau; ... ; \tau}_{k}$.
We say that $a:(h,h)$ is \emph{idle} if it is valid and $a = a_{1}\otimes \cdots\otimes a_{h}$ for some suitable $a_{1}:(1,1), ..., a_{h}:(1,1)$ that are idle (note that, since $a$ is valid then each $a_i$ contains the same number of occurrences of $\tau$).

For $a:(1,1)$ such that $\bar{a} = \overline{n_{1}}; \overline{n_{2}}; ... ; \overline{n_{k}}$, we let 
$\counttok{a} = \sum_{i=1}^{k} n_{i}$.

In the following, we shall often denote parallel composition of observation just as juxtaposition, to keep the labels of tiles as compact as possible (see, e.g., rules~\ruleLabel{Tw$_{h,k}$}, \ruleLabel{$\diag$}, \ruleLabel{$\codiag$}, \ruleLabel{$\ldiag_{h}$}, \ruleLabel{$\lcodiag_{h}$} in Fig.~\ref{fig:tilecalcopsem} and rule~\ruleLabel{Mon} in Fig.~\ref{fig:ordtiles}).

\subsubsection*{Tiles.}
The operational semantics is given by the tile system freely generated from the basic tiles in Fig.~\ref{fig:tilecalcopsem}, using the composition rules in Fig.~\ref{fig:ordtiles}.
Rules for stateless connectors are analogous to those of the P/T Calculus shown in Fig.~\ref{fig:ptcalcopsem}. Rule~\ruleLabel{TkI} is a particular case of rule~\ruleLabel{TkIO$_{n,h,k}$}. Rule~\ruleLabel{TkO} is in charge of adding a new epoch to the observation when a  token is consumed.

% OPERATIONAL SEMANTICS
\begin{figure}[t]
\[
%\derivationRule{n,h,k\in \{0,1\}\ \ k\leq n\ \ n+h\leq 1}{\tokensplace{n} \dtrans{h}{k} \tokensplace{n+h-k}}{TkIO} \quad
\derivationRule{}{\tokensplace{0} \dtrans{1}{0} \tokensplace{1}}{TkI} \quad
\derivationRule{}{\token \dtrans{\quad}{\epoch;1} \lzero}{TkO} \quad
\derivationRule{h,k\in \{0,1\}}{\tw \dtrans{hk}{kh} \tw}{Tw$_{h,k}$} \quad
\derivationRule{}{\rightEnd \dtrans{1}{} \rightEnd}{$\rightEnd$} \quad
\derivationRule{}{\leftEnd \dtrans{}{1} \leftEnd}{$\leftEnd$}
\]
\[
\derivationRule{}{\diag \dtrans{1}{1 \labelSep 1} \diag}{$\diag$} \quad
\derivationRule{}{\codiag \dtrans{1 \labelSep 1}{1} \codiag}{$\codiag$} \quad 
\derivationRule{h\in \{0,1\}}{\ldiag \dtrans{1}{ h\labelSep (1-h)} \ldiag}{$\ldiag_{h}$} \quad
\derivationRule{h\in \{0,1\}}{\lcodiag \dtrans{h\labelSep (1-h)}{1} \lcodiag}{$\lcodiag_{h}$} \quad
\derivationRule{\mathsf{C}\typ\sort{k}{l}\mbox{ \footnotesize a basic connector}}{\mathsf{C}\dtrans{\epoch^k}{\epoch^l} \mathsf{C}}{Epoch}
%\derivationRule{P\typ\sort{k}{l}}{P\dtrans{\epoch^k}{\epoch^l} P}{Epoch}
\]
\caption{Basic tiles for the Petri Tile calculus.\label{fig:tilecalcopsem}}
\end{figure}

\begin{figure}[t]
\[
\derivationRule{P\typ\sort{k}{l}}{P\dtrans{0^k}{0^l} P}{Idle} \quad
\derivationRule{P\dtrans{a}{c} Q \quad R\dtrans{c}{b} S}
{P\comp R \dtrans {a}{b} Q\comp S}{Hor} \quad
\derivationRule{P\dtrans{a}{b} Q\quad R\dtrans{c}{d} S} 
{P\ten R \dtrans{a\labelSep c}{b\labelSep d} Q\ten S}{Mon} \quad
\derivationRule{P\dtrans{a}{b} P'\quad P'\dtrans{a'}{b'} Q }
{P\dtrans{a;a'}{b;b'} Q}{Vert}
\]
\caption{Ordinary rules for tile systems.\label{fig:ordtiles}}
\end{figure}

Note that the rules
\[
  \derivationRule{}{\id  \dtrans{1}{1} \id }{Id$_{1}$} 
\qquad  \derivationRule{}{\id  \dtrans{\tau}{\tau} \id }{I$_{\tau}$} 
\]
\noindent
%are not needed as basic tiles, because they 
are generated by the free construction of the tile system.

Although the rule \ruleLabel{Epoch} is given for basic connectors only, it can be immediately proved by structural induction that for any $P$ we have 
\[
\derivationRule{P\typ\sort{k}{l}}{P\dtrans{\epoch^k}{\epoch^l} P}{Epoch}
\]

In the weak case, where $\tau = 0$, then  the rule \ruleLabel{Epoch} just becomes the ordinary \ruleLabel{Idle} rule of tiles and the basic tile \ruleLabel{TkO} becomes 

\[
\derivationRule{}{\token \dtrans{}{1} \lzero}{WeakTkO}
\]

We note that the rules of the Petri tile calculus are in the so-called \emph{basic source} format.

\begin{defi}[Basic source]\label{defin:basicsource}
  A tile system $\mathcal{R}=(\mathcal{H},\mathcal{V},N,R)$ enjoys
  the basic source property if for each $A \in N$ if $R(A)=\langle
  s,a,b,t\rangle$, then $s\in\Sigma_{\mathcal{H}}$.
\end{defi}

The basic source property is a syntactic criterion ensuring that tile
bisimilarity is a congruence (in both the strong and the weak case of the Petri tile calculus).

%We recall that tile bisimilarity is a congruence (in both the strong and the weak case) because 

\begin{lem}[cf.~\cite{DBLP:conf/birthday/GadducciM00}]\label{lemma:congr}
If a tile system ${\mathcal R}$ enjoys the basic source property, then tile
bisimilarity is a congruence 
%for ${\mathcal R}$ 
(w.r.t. $\_;\_$ and $\_\otimes\_$).
\end{lem}

As a consequence, a tile for $P\otimes Q$ can always be obtained as the parallel
composition of one tile for $P$ and one for $Q$; and similarly, for $P;Q$. 
This fact is implicitly exploited in many proofs, which can thus be performed by structural induction.

The behaviour of  basic connectors is characterised below (See Appendix~\ref{app:proof-ptiles} for more details and all technical lemmas cited in the proofs). 
For any observation $a,b,c,d$ and $h,k\in \nat$, we have:
\begin{iteMize}{$-$}
\item $\rightEnd \dtrans{a}{} \rightEnd$ and $\leftEnd \dtrans{}{a} \leftEnd$.
%,  $\diag \dtrans{a}{a \labelSep a} \diag$, and
% $\codiag \dtrans{a \labelSep a}{a} \codiag$.
\item $a=b=c$ if and only if $\diag \dtrans{a}{b \labelSep c} \diag$ and $
\codiag \dtrans{b \labelSep c}{a} \codiag$.
\item $a=d$, $b=c$ and $a,b$ are coetaneous if and only if $\tw \dtrans{ab}{cd} \tw$
\item $a=b+c$ if and only if $\ldiag \dtrans{a}{ b\labelSep c} \ldiag$,  $\lcodiag \dtrans{b\labelSep c}{a} \lcodiag$.
\item $a$ is idle if and only if $\lzero \dtrans{}{a} \lzero$, $\rzero \dtrans{a}{} \rzero$.
%\item $\tw \dtrans{hk}{kh} \tw \quad
%\rightEnd \dtrans{k}{} \rightEnd \quad
%\leftEnd \dtrans{}{k} \leftEnd \quad
%\diag \dtrans{k}{k \labelSep k} \diag \quad
%\codiag \dtrans{k \labelSep k}{k} \codiag \quad 
%\ldiag \dtrans{h+k}{ h\labelSep k} \ldiag \quad
%\lcodiag \dtrans{h\labelSep k}{h+k} \lcodiag
%$\\
%$
%\tw \dtrans{\overline{hk}}{\overline{kh}} \tw \quad
%\rightEnd \dtrans{\overline{k}}{} \rightEnd \quad
%\leftEnd \dtrans{}{\overline{k}} \leftEnd \quad
%\diag \dtrans{\overline{k}}{\overline{k \labelSep k}} \diag \quad
%\codiag \dtrans{\overline{k \labelSep k}}{\overline{k}} \codiag \quad 
%\ldiag \dtrans{\overline{h+k}}{ \overline{h\labelSep k}} \ldiag \quad
%\lcodiag \dtrans{\overline{h\labelSep k}}{\overline{h+k}} \lcodiag
%$
\end{iteMize}

\begin{lem}\label{lem:PtoP}
Let $P:(h,l)$ be any stateless connector. If $P \dtrans{a}{b} Q$ then $Q=P$.
\end{lem}
\begin{proof}
By straightforward structural induction on $P$.
\end{proof}

\begin{lem}\label{lem:coetaneousbasic}
Let $P:(h,l)$ be any basic stateless connector. If $P \dtrans{a}{b} P$ then $a$ and $b$ are (valid and) coetaneous.
\end{lem}
\begin{proof}
The property is obvious for $\rightEnd$, $\leftEnd$, $\lzero$, $\rzero$. 
For the other connectors, the property is an immediate consequence 
of some technical lemmas reported in Appendix~\ref{app:proof-ptiles}:
Lemma~\ref{lem:coetaneousdiag} (for $\diag$ and $\codiag$); 
Lemma~\ref{lem:coetaneoustw} (for $\tw$); 
Lemma~\ref{lem:coetaneousldiag} (for $\ldiag$ and $\lcodiag$).
\end{proof}

\begin{lem}\label{lem:nminuskplush}
For any $h,k,n$, if $k\leq n$ then $\tokensplace{n} \dtrans{\bar{h}}{\bar{k}} P'$ with $P' \simeq_{\mathsf{tb}} \tokensplace{n+h-k}$.
Moreover, for any $h,k,n$, if $\tokensplace{n} \dtrans{\bar{h}}{\bar{k}} P'$ then $P' \simeq_{\mathsf{tb}} \tokensplace{n+h-k}$ and $k\leq n$.
\end{lem}
\begin{proof}
See Appendix~\ref{app:proof-ptiles}.
%The proof is along the lines of Lemma~\ref{lem:nminusk}, showing that for any $k\leq n$, $\tokensplace{n} \dtrans{\bar{0}}{\bar{k}} P'$ with $P' \simeq_{\mathsf{tb}} \tokensplace{n-k}$ and that for any $k,n$, if $\tokensplace{n} \dtrans{\bar{0}}{\bar{k}} P'$ then $P' \simeq_{\mathsf{tb}} \tokensplace{n-k}$  and $k\leq n$;
%but it additionally exploits Lemma~\ref{lem:bar0plush}.
\end{proof}

%\begin{lem}\label{lem:nminuskplush}
%For any $h,k,n$, if $k\leq n$ then $\tokensplace{n} \dtrans{\bar{h}}{\bar{k}} P'$ with $P' \simeq_{\mathsf{tb}} \tokensplace{n+h-k}$.
%\end{lem}

\subsection{Correspondence with P/T calculus (strong case)}

We prove the correspondence theorem between Petri tile calculus and P/T nets with boundaries (strong), by transitivity, proving the correspondence with P/T calculus (strong).

By abusing notation, we use the identity mapping to associate Petri tile configurations to P/T calculus terms and vice versa, since  $\token$ can be read as the P/T calculus term $\lzero;\tokensplace{1}$, (see Appendix~\ref{app:proof-ptiles},  Lemma~\ref{lem:onetokenalone}).

Regarding observations, %as evident from the correspondence theorems below, 
we map the label $\alpha = n_{1}n_{2}\cdots n_{k}$ to the observation $\bar{\alpha} = \bar{n}_{1}\bar{n}_{2}\cdots \bar{n}_{k}$.

\begin{lem}\label{lem:strongtile}
If $P \dtrans{\alpha}{\beta} P'$ in the P/T calculus, then $P \dtrans{\bar{\alpha}}{\bar{\beta}} P''$ in the Petri tile calculus with $P' \simeq_{\mathsf{tb}} P''$. Vice versa, If $P \dtrans{\bar{\alpha}}{\bar{\beta}} P'$ in the Petri tile calculus, then $P \dtrans{\alpha}{\beta} P''$ in the P/T calculus with $P' \simeq_{\mathsf{tb}} P''$.
\end{lem}
\begin{proof}
The proof is by structural induction on $P$, exploiting the technical Lemmas~\ref{lem:strongtilebasic} and~\ref{lem:strongtileplaces} in Appendix~\ref{app:proof-ptiles}.
\end{proof}

Next, we prove the correspondence at the level of sequences. One direction of the correspondence is easy.

\begin{thm}\label{theo:pttotile}
If there is a sequence of transitions $P
\dtrans{\alpha_{1}}{\beta_{1}} \dtrans{\alpha_{2}}{\beta_{2}} \cdots
\dtrans{\alpha_{k}}{\beta_{k}} P'$ in the P/T calculus, then $P
\dtrans{\bar{\alpha}_{1}}{\bar{\beta}_{1}}
\dtrans{\bar{\alpha}_{2}}{\bar{\beta}_{2}} \cdots
\dtrans{\bar{\alpha}_{k}}{\bar{\beta}_{k}} P''$ in the Petri tile
calculus with $P' \simeq_{\mathsf{tb}} P''$.
\end{thm}
\begin{proof}
By induction on the length $k$ of the computation, exploiting Lemma~\ref{lem:strongtile}.
\end{proof}

The other direction is less obvious, because a configuration $P$ can evolve via a tile $P \dtrans{a}{b} Q$ without $a$ and $b$ being necessarily valid and coetaneous. In fact, it is actually the case that tile bisimilarity is stronger than bisimilarity in the P/T calculus, as the following example shows.

\begin{exa}
Let us consider the term $\tw;\tw$ of the P/T calculus. We have clearly that $\tw;\tw \sim \id \otimes \id$ in the P/T calculus.
On the other hand, the tile configuration $\id \otimes \id$ can make concurrent steps like  
$\id \otimes \id \dtrans{1\; (\bar{5};\bar{3})}{1\; (\bar{5};\bar{3})} \id \otimes \id$ 
obtained as the parallel composition of two tiles 
$\id \dtrans{1}{1} \id$
and
$\id \dtrans{\bar{5};\bar{3}}{\bar{5};\bar{3}} \id$
that cannot be matched by
$\tw;\tw$ (because it admits valid and coetaneous observations only, cf. Lemma~\ref{lem:coetaneoustw}, while $1$ and $\bar{5};\bar{3}$ are not coetaneous).
\end{exa}

We now compare tile bisimilarity $\simeq_{\mathsf{tb}}$ to strong bisimilarity $\sim$ for the P/T calculus.

\begin{thm}
$P\simeq_{\mathsf{tb}} Q$ implies $P\sim Q$.
\end{thm}
\begin{proof}
Direct consequence of Theorem~\ref{theo:pttotile}.
\end{proof}

This fact is quite interesting, because it shows that tile bisimilarity is able to characterise a finer concurrent semantics (than the P/T calculus) where no assumption is made about the timing of concurrent events. Instead, both the Petri calculus and the P/T calculus (in the strong case) force the simultaneous observation of a step across disconnected parts of the net.

One may argue that $\tw$ should not synchronise the interfaces, and tiles can deal with this situation by allowing the exchange of any $a$ and $b$, even non-coetaneous ones. In fact, this would correspond to take a symmetric monoidal category of observations, with $\tw$ being an auxiliary arrow. However, we prefer to keep the synchronising $\tw$, because we can then exploit it to recover exactly the semantics of P/T calculus, and by transitivity, that of P/T nets with boundaries.

Let us denote by $\tw_{n,n}:(n+1,n+1)$ the configuration inductively defined as:
\[
\tw_{0,0} \defeq \id\qquad
\tw_{1,1} \defeq \tw;\tw\qquad
\tw_{n+1,n+1} \defeq (\tw \otimes \id_{n}); \tw_{n,n};(\tw \otimes \id_{n})
\]

\begin{lem}\label{lem:twtw}
For any $n,a,b$ we have that $\tw_{n,n} \dtrans{a}{b} P$ if and only if $P=\tw_{n,n}$ and $a=b$.
\end{lem}
\begin{proof}
By induction on $n$, exploiting Lemma~\ref{lem:coetaneoustw}.
\end{proof}

Roughly, one can think of $\tw_{n,n}$ like a (stateless) connector that behaves like $n+1$ identities, but filters out non valid and non coetaneous sequences. 
Then, for $P:(h,k)$, let us denote by $\synch{P}$ the term
\[
\synch{P} \defeq (\id_{h}\otimes \lzero);\tw_{h,h};(P\otimes \id);\tw_{k,k};(\id_{k}\otimes \rzero)
\]

Essentially, $\synch{P}$ embeds $P$ in parallel with some sort of ``clock'' wire $\id$, then synchronises the left and right interfaces of $P$ and the ``clock'' (the additional clock wiring is needed because the left and the right interfaces of $P$ may be disconnected, like in $P=\rightEnd ; \leftEnd$), and finally hide the clock using $\lzero$ and $\rzero$.

Then, if we embed any $P:(h,k)$ within $\synch{P}$ we are not dramatically changing the overall behaviour of $P$, because we can always find a valid and coetaneous step $P \dtrans{a'}{b'} Q$ for any non-valid or non-coetaneous step $P \dtrans{a}{b} Q$
(cf. Lemmas~\ref{lem:allcoeteneous}--\ref{lem:slicemany}). Moreover, if $P$ has no concurrent activities, then clearly $P  \simeq_{\mathsf{tb}} \synch{P}$.

\begin{lem}
$\synch{P} \sim P$.
\end{lem}
\begin{proof}
Let $P:(h,k)$. 
We first prove that, for any $n$, $\tw_{n,n} \sim \id_{n+1}$ (by induction on $n$).
Therefore $\synch{P} \sim (\id_{h}\otimes \lzero);(P\otimes \id);(\id_{k}\otimes \rzero)
\sim P \otimes (\lzero;\rzero) \sim P$.
\end{proof}

\begin{thm}\label{theo:tiletopt}
If $\synch{P} \dtrans{a}{b} Q$ in the Petri tile calculus then there
is a sequence of transitions $\synch{P} \dtrans{\alpha_{1}}{\beta_{1}}
\dtrans{\alpha_{2}}{\beta_{2}} \cdots \dtrans{\alpha_{k}}{\beta_{k}}
Q'$ in the P/T calculus with $\bar{a} = \bar{\alpha}_{1};
\bar{\alpha}_{2}; \cdots ; \bar{\alpha}_{k}$, $\bar{b} =
\bar{\beta}_{1}; \bar{\beta}_{2}; \cdots ; \bar{\beta}_{k}$, and $Q'
\simeq_{\mathsf{tb}} Q$.
\end{thm}
\begin{proof}
See Appendix~\ref{app:proof-ptiles}.
%By Lemma~\ref{lem:synchP}, $a$ and $b$ are valid and coetaneous and $Q = \synch{P'}$ for some $P'$ such that 
%$P \dtrans{a}{b} P'$.
%By Lemma~\ref{lem:slicemanyvalid},
%there exist
%$k$ elementary observations $\alpha_{1},...,\alpha_{k}$ with $\bar{\alpha}_{1};\cdots;\bar{\alpha}_{k} = \bar{a}$, 
%$k$ elementary observations $\beta_{1},...,\beta_{k}$ with $\bar{\beta}_{1};\cdots;\bar{\beta}_{k} = \bar{b}$, 
%and $k$ terms $P_{1},P_{2},...,P_{k}$ such that
%$P \dtrans{\bar{\alpha}_{1}}{\bar{\beta}_{1}} P_{1}$, then
%$P_{1} \dtrans{\bar{\alpha}_{2}}{\bar{\beta}_{1}} P_{2}$,
%...,
%$P_{k-1} \dtrans{\bar{\alpha}_{k}}{\bar{\beta}_{k}} P_{k} = P'$.
%Then, we proceed by straightforward  induction on the length $k$ of the computation, exploiting Lemma~\ref{lem:tiletopt1}.
\end{proof}

\begin{thm}
$P\sim Q$ if and only if $\synch{P} \simeq_{\mathsf{tb}} \synch{Q}$.
\end{thm}

\begin{rem}
Although we skip details here, the case of C/E nets can also be dealt with in the tile model by: (1) replacing the constant $\token:(0,1)$ of Petri tile calculus with the new constant $\tokensplace{1}:(1,1)$; (2) replacing the basic tile~\ruleLabel{TkO} in Fig.~\ref{fig:tilecalcopsem} with
$$
\derivationRule{}{\tokensplace{1} \dtrans{\epoch}{\epoch;1} \tokensplace{0}}{TkO'}
$$
and (3) prefixing with $\tau$ all the observations of the basic tiles in Fig.~\ref{fig:tilecalcopsem}, except for tile~\ruleLabel{Epoch} where $\tau$'s are already present. Then, the different semantics discussed in Remark~\ref{rmk:otherCEsemantics} can be recovered by considering the different combinations with additional tiles
$$
\derivationRule{}{\tokensplace{0} \dtrans{\epoch;1}{\epoch;1} \tokensplace{0}}{TkI2}
\qquad
\derivationRule{}{\tokensplace{1} \dtrans{\epoch;1}{\epoch;1} \tokensplace{1}}{TkO2}
$$
Note also that the weak case discussed below, where $\epoch = 0 = id_1$, subsumes the above tiles~\ruleLabel{TkI2} and~\ruleLabel{TkO2} as they can be derived starting from tiles~\ruleLabel{TkI} and~\ruleLabel{TkO'} thanks to the vertical composition of tiles.
\qed
\end{rem}

\subsection{Correspondence with P/T calculus (weak case)}

In the weak case, the correspondence between the Petri tile calculus and the P/T calculus is much easier to prove.
 We recall that in the weak case, the symbol $\tau$ is just the identity and that tile \ruleLabel{Epoch} coincides with \ruleLabel{Idle}.
 Consequently, the only observations allowed are sequences of $1$ (with $0$ still being the identity), that we still denote by natural numbers, i.e., we will write $n$ as observation to denote a sequence of $1$'s of length $n$. 
 In the following we denote by $\approxeq_{\mathsf{tb}}$ the tile bisimilarity for the weak case.

\begin{lem}\label{lem:tileweak}
For any $h,k,n,m,a,b$:
\begin{enumerate}[\em(1)]
\item $h=k=n$ if and only if 
$\diag \dtrans{n}{h \labelSep k} \diag$ and 
$\codiag \dtrans{h \labelSep k}{n} \codiag$;
\item $h=k$ and $n=m$ if and only if $\tw \dtrans{hn}{mk} \tw$;
\item $n=h+k$ if and only if $\ldiag \dtrans{n}{ h\labelSep k} \ldiag$ and
$\lcodiag \dtrans{h\labelSep k}{n} \lcodiag$;
\item $n=0$ if and only if $\lzero \dtrans{}{n} \lzero$ and
$\rzero \dtrans{n}{} \rzero$;
\item $h=k$ if and only if $\id \dtrans{h}{k} \id$;
\item if $P$ is a stateless connector and $P \dtrans{a}{b} Q$, then $Q=P$;
\item $k\leq n+h$ and $m=n+h-k$ if and only if $\tokensplace{n} \dtrans{h}{k} P$ with $P \approxeq_{\mathsf{tb}} \tokensplace{m}$
\end{enumerate}
\end{lem}
\begin{proof}
The proof of ($1$--$4$) immediately follows from Lemmas~\ref{lem:coetaneousdiag}--\ref{lem:idleobs}.
The proof of ($5$) follows by the property of identity $\id$.
The proof of ($6$) follows by Lemma~\ref{lem:PtoP}.
The proof of ($7$) is analogous to the proof of Lemma~\ref{lem:nminuskplushbar}.
\end{proof}
 
 By abusing the notation, we use the identity mapping to associate Petri tile configurations (resp. observations) to P/T calculus terms (resp. labels), and vice versa.

\begin{thm}
$P \dtrans{h}{k} P'$ in the   (weak) Petri tile calculus if and only if $P\dtransw{h}{k} P''$ in the P/T calculus  with $P' \approxeq_{\mathsf{tb}} P''$.
\end{thm}
\begin{proof}
By induction on the structure of $P$.
On basic connectors we exploit Lemma~\ref{lem:tileweak}.
For composite terms, we observe that rules \ruleLabel{Seq}, \ruleLabel{Par} and \ruleLabel{Weak} of the P/T calculus directly correspond to tile compositions rules \ruleLabel{Hor}, \ruleLabel{Mon} and \ruleLabel{Vert}, respectively.
\end{proof}

It is now easy to compare tile bisimilarity 
$\approxeq_{\mathsf{tb}}$ to weak bisimilarity $\approx$ for the P/T calculus.

\begin{cor}
$P \approx Q$ if and only if $P \approxeq_{\mathsf{tb}} Q$.
\end{cor}

\section{Related work}
\label{sec:related}

\subsubsection*{Composable nets.}
Process algebras and Petri nets are two of the most popular models of concurrent systems and many works addressed their joint use by defining suitable ``calculi of nets'', where process-algebra like syntax is used to build more complex nets out of a small set of basic nets. One of the most successful proposals along this thread of research is the so-called \emph{Petri Box calculus}~\cite{DBLP:conf/concur/KoutnyEB94,DBLP:journals/tcs/KoutnyB99,DBLP:journals/iandc/BestDK02,DBLP:conf/ac/BestK03}. The key idea is to develop a general algebraic theory of net compositions, without relying on any preconceived set of basic nets and operators. Roughly, any set of safe and clean nets can provide the basic components, called \emph{plain boxes}. Similarly, suitable nets, called \emph{operator boxes}, can be chosen to provide composition-by-refinement: if the operator box has $n$ transitions, it should receive $n$ arguments (e.g. plain boxes) that are used to refine the transitions element-wise. The operator boxes guarantee that the result is also a plain box. Moreover, the algebra provides suitable syntax for denoting the position of tokens within the box hierarchy (by overlining and underlining expressions). In fact, while the structure of the net is not affected by the firing of transitions, the dynamic evolution is reflected in the changing markings. This is modelled by differentiating \emph{static expressions} (i.e. structure) from \emph{dynamic expressions}  (i.e. structure plus state): in the latter case, an overlined expression $\overline{a}$ means  a token is present before $a$ (thus enabling it) and an underlined expression $\underline{a}$ means a token is present after $a$ (e.g. after $a$ has been executed).  Any ambiguity in the over-/under-lining is banned by a suitable structural equivalence over dynamic expressions and the operational semantics is then defined in the SOS style over dynamic expressions only by ``moving'' the over-/under-lining (i.e. without changing the underlying static expression that fixes the overall structure of the plain boxes). The Petri Box calculus has been also enriched in~\cite{DBLP:journals/fuin/DevillersKKP03} with buffer places  where different transitions may deposit and remove tokens to represent asynchronous communication. Although the flavour of the Petri Box approach is different from ours, because it addresses a particular class of well-behaving nets and does not fix a minimal algebra that generates \emph{all} nets, it would be interesting to investigate how the Petri Box approach can be extended to deal with ``boundaries'' for composition.

Approaches such as~\cite{DBLP:journals/fuin/BernardinelloMP07,DBLP:conf/apn/Fabre06} study the problem of composing nets over a well-defined communication protocol shared by components, called the interface. Each component is 
seen as a refinement of the interface and the composition operation merges all components by fusing those parts that are mapped to the same elements of the interface. Technically speaking, components are characterised by 
refinement morphisms that map elements of the components to the shared interface. Then, the composition  is modelled as a product in a suitable category of nets. Differently from these approaches, our nets are composed over shared interfaces just by juxtaposing components and, hence,  boundaries are the only elements fused during composition.

Other process algebraic approaches to the representation of nets are~\cite{DBLP:journals/mscs/LeiferM06}, where a particular flavour of reactive systems, called \emph{link graphs}, are shown to be capable of modelling C/E nets and provide them with an LTS semantics for which bisimilarity is a congruence. A similar, if more direct, approach was developed in~\cite{Sassone2005b}. In~\cite{DBLP:conf/fossacs/BuscemiS01} P/T nets are characterised as a suitable typed fragment of the join calculus. One main difference w.r.t. our approach is that, in both cases, nets are modelled by ``uniquely naming'' places and transitions and exploiting classical name-handling mechanisms of nominal calculi to compose nets. Since names can be shared and made private, there can be some analogy with our synchronising and hiding connectors, but not with the ones for mutual exclusion and inaction.

An approach maybe closer to our objective is the one in~\cite{DBLP:conf/concur/NielsenPS95,DBLP:journals/tcs/PrieseW98}, where a notion of Petri nets with interfaces is introduced in order to design a set of net combinators for which suitable behavioural congruences can be defined. The interfaces in~\cite{DBLP:conf/concur/NielsenPS95,DBLP:journals/tcs/PrieseW98} consist of ``public'' transitions and places that are used by the net to communicate with its surrounding context. The approach led to the definition of an elementary calculus in which one can construct any Petri net with an interface from trivial constants (single places, single transitions) by drawing arcs, adding tokens, and hiding public places and transitions. The behavioural congruences are defined by considering a \emph{universal context} $U$ such that two Petri nets behave the same in any context if their behaviour is equal in the universal context. The key difference with our approach is that by having ports in the interface, instead of places and transitions, we can define behavioural congruences without needing a universal context for experimenting, because our nets come equipped with an interactive operational semantics. Moreover, it seems that our notion of composition is slightly more powerful, because of the combinatorial way of composing transitions attached to the same port.

A similar idea is followed in~\cite{DBLP:journals/mscs/BaldanCEH05}, which introduces \emph{open nets}. Open nets come equipped with a distinguished set of places, called \emph{open places}, that form the interface between the system and the environment. In~\cite{DBLP:journals/mscs/BaldanCEH05}, the basic building blocks of any system are the transitions and the main operation for composition is given in terms of category theory as a pushout. Essentially, the composition glues two open nets together along their common open places and it is general enough to accommodate  both interaction through open places and synchronisation of transitions. The deterministic process semantics is shown to be compositional with respect to such a composition operation. Given the particular role played by open places, open nets are maybe the model closest in spirit to our approach. One main difference is that our approach focus on the operational and abstract semantics and not on the process semantics.

\subsubsection*{Connectors.}

Different studies about primitive forms of connectors have appeared in the literature. 
Our approach to connectors is much indebted to~\cite{DBLP:journals/igpl/Stefanescu98,DBLP:journals/tcs/BruniGM02}. 

In~\cite{DBLP:journals/tcs/BruniLM06}, the algebra of stateless connectors  inspired by previous work on simpler algebraic structures~\cite{DBLP:journals/tcs/BruniGM02,DBLP:journals/igpl/Stefanescu98} was presented.
The operational, observational and denotational semantics of  connectors are first formalised separately and then shown to coincide.  Moreover, a complete normal-form axiomatisation is available for them. 
%These networks are quite expressive: for instance it is shown~\cite{DBLP:journals/tcs/BruniLM06} that they can model all the (stateless) connectors of the architectural design language CommUnity~\cite{DBLP:journals/scp/FiadeiroM97}.
The work in~\cite{DBLP:journals/tcs/BruniLM06} also reconciles the algebraic and  categorical approaches to system modelling.
The algebraic approach models systems as terms in a suitable algebra.  Operational and abstract semantics are then usually based on inductively defined labelled transition systems.  The categorical approach models systems as objects in a category, with morphisms defining relations such as subsystem or refinement. Complex software architectures can be modelled as diagrams in the category, with universal constructions, such as colimit, building an object in the same category that behaves as the whole system and that is uniquely determined up to isomorphisms.
While equivalence classes are usually abstract entities in the algebraic approach, having a normal form gives a concrete representation that matches a nice feature of the categorical approach, namely that the colimit of a diagram is its best concrete representative.

Reo~\cite{DBLP:journals/mscs/Arbab04} is an exogenous coordination model based on channel-like connectors that mediate the flow of data among components. Notably, a small set of point-to-point primitive connectors is sufficient to express a large variety of interesting constraints over the behaviour of connected components, including various forms of mutual exclusion, synchronisation, alternation, and context-dependency. 
Typical primitive connectors are the synchronous / asynchronous / lossy channels and the asynchronous one-place buffer. They are attached to ports called Reo nodes. Components and primitive connectors can be composed into larger Reo circuits  by disjoint union up-to the merging of shared Reo nodes. The semantics of Reo has been formalised in several ways, exploiting co-algebraic techniques~\cite{conf/wadt/ArbabR02}, constraint-automata~\cite{journals/scp/BaierSAR06}, colouring tables~\cite{journals/scp/ClarkeCA07}, and the tile model~\cite{DBLP:conf/wadt/ArbabBCLM08}. See~\cite{JA11} for a recent survey.

BIP~\cite{DBLP:conf/sefm/BasuBS06} 
is a component framework for constructing 
systems by superposing three layers of modelling, called 
Behaviour,
Interaction, and
Priority.
At the global level, the behaviour of a BIP system can be faithfully represented by a safe 
Petri net with priorities, whose single transitions are obtained by fusion of component transitions
according to the permitted interactions, and priorities are assigned accordingly.
In absence of priorities, an algebraic presentation of BIP connectors with vacuous
priorities is given in~\cite{DBLP:journals/tc/BliudzeS08}.
One key feature of BIP is the so-called \emph{correctness by construction},
which allows the
specification of architecture transformations preserving certain properties of 
the underlying behaviour.
For instance it is possible to provide (sufficient) conditions for compositionality 
and composability which guarantee deadlock-freedom.
The BIP component framework has been implemented in a language and a tool-set. 
%A compositional version of BIP systems is presented in~\cite{BruniPSI2011}.
The formal relation between BIP and nets with boundaries has been studied in~\cite{BruniPSI2011}.  
Firstly, it is shown that any BI(P) system (without priorities) can be mapped into a 1-safe Petri net that 
preserves computations. Intuitively, the places of the net are in one-to-one correspondence with the
states of the components, while the transitions of the net represent the synchronised
execution of the transitions of the components. In addition, \cite{BruniPSI2011} introduces a composition operation for 
BI(P) systems that enables the hierarchical definition of systems. Then, this compositional version of BI(P) systems is  used 
to define a compositional mapping of BI(P) systems into bisimilar nets with boundaries. Finally, 
it is shown that any net with boundaries without left interface can be
encoded as a BI(P) system consisting on just one component. It is in this sense that 
BI(P) systems and nets with boundaries are retained equivalent.  

\subsubsection*{Tiles and Wires.}
Considered as process algebras, the operations of the systems presented in this paper are fundamentally different to those traditionally considered by process algebraists. Indeed, they are closer in nature to the algebra of tile logic~\cite{DBLP:conf/birthday/GadducciM00,Bru:TL} and the algebra of Span(Graph)~\cite{Katis1997a} (which are both based on the algebra of monoidal categories) than, say, to the primitives of CCS such as a commutative parallel composition operation. 

The Tile Model offers a flexible and
adequate semantic setting for concurrent systems~\cite{RM:SHR,DBLP:journals/iandc/FerrariM00,DBLP:journals/tcs/BruniM02} and also for  defining the operational and abstract semantics of suitable classes of connectors. 
Tiles resemble Plotkin's SOS inference rules~\cite{DBLP:journals/jlp/Plotkin04a}, but  take inspiration also from Structured Transition Systems~\cite{Corradini-Montanari/92} and context systems~\cite{LX:CTOSC}.
The Tile Model also extends rewriting logic~\cite{Mes:CRL} (in the non-conditional case) by taking into account rewrite with side effects and rewrite synchronisation. 
While in this paper we exploit horizontal connectors only, in~\cite{DBLP:journals/tcs/BruniM02} it is shown how to benefit from the interplay of connectors in both the horizontal and vertical dimensions for defining causal semantics.

In \cite{Katis1997b} Span(Graph) is used to capture the state space of P/T nets; that work is close in spirit to the translations from nets to terms given in this paper. A process algebra, called \emph{wire calculus}, based on similar operations has been previously studied in~\cite{DBLP:journals/corr/abs-0912-0555}. The wire calculus shares strong similarities with the (simplest monoidal version of the) tile model, in the sense that it has sequential and parallel compositions and exploits trigger-effect pairs labels as observations. However, the tile model can be extended to deal with more sophisticated kinds of configurations and observations. The wire calculus has a more friendly process algebra presentation instead of relying on categorical machinery and it exploits a different kind of vertical composition.
The usual action prefixes $a.P$ of process algebras are extended in the wire calculus by the simultaneous input of a trigger $a$ and output of an effect $b$, written $\frac{a}{b}.P$, where $a$ (resp. $b$) is a string of actions, one for each input port (resp. output port) of the process.

\section{Conclusions}
\label{sec:conclusions}

In theoretical computer science, it is very frequent that quite different representations are shown to be equally expressive by providing mutual encoding with tight semantics correspondence: thus, in the end, they are different ways to represent the same abstract concept. 

In this paper, we have contributed to the above thread by relating the expressiveness of nets with boundaries, process calculi and tile model across several spectra: i) condition/event approach (one-place buffers) vs place/transition approach (unbounded buffers); ii) strong semantics vs weak semantics; iii) tile systems vs SOS rules.

The constructions and equivalences presented in this paper witness that we can move smoothly from one model to the other, emphasising the crucial rules that needs to be changed. Still, one anomaly emerged from our study that we think is worth remarking here.

The anomaly, somehow foreseeable, is that the strong tile bisimilarity for the Petri tile calculus nets is finer than the strong bisimilarity for the P/T calculus. This is due to the inherent concurrency of the tile model, that leads tile bisimilarity to distinguish concurrent behaviours arising in disconnected subsystems just because no synchronisation mechanism can be enforced on their observations. When this feature is not wanted, then it can be solved simply by making sure to connect together all subsystems by suitable ``transparent'' connectors that behave as identity, except for providing a shared ``clock'' synchronisation when needed. Note that even in this case, the inherent concurrency of sub-system is fully maintained between one tick of the clock and the next.

%The second and  most surprising anomaly is probably the fact that all bisimilarities we %have defined are congruences, except for the strong bisimilarity of C/E nets with %boundaries. We are now confident that the problem is due to the fact that the strong %semantics we have defined enforces a binary observation on shared ports, i.e., it %transfers to the set of ports the same policy imposed on places by C/E systems. We %conjecture that the congruence property can be restored once different policies are %applied on ports and places. In particular, if we relax the independence requirement on %ports when firing sets of transitions, it may happen that a natural number $n$ is observed %on a port, because several arcs are attached to it. Then, when two nets are composed in %series, we may be able to find a match for the $n$ arcs on one side with other $n$ arcs on %the other side. At the level of process algebras, this would correspond to mix one-place %buffers of the Petri calculus (with rules {\sc TkI} and {\sc TkO}) with the more general %stateless connectors of the P/T calculus. Another way to see the mix, is like the %{\ruleLabel{Weak}} rule was applied to stateless connectors only, not to buffers. 

%Beside the formal assessment of the two issues above, our results open several interesting %perspectives on future work.
%Among them, 
Among several possibilities for future work we mention: i) study and compare the expressive power of fragments of the Petri calculus and of the P/T calculus where certain connectors are excluded; ii) exploit the analogy with the Petri Box approach to define high level composition operators that can preserve suitable properties of nets with boundaries; iii) investigate the expressiveness of symmetric monoidal tile models, where the connector $\tw$ does not enforce any synchronisation and characterise the corresponding bisimilarity equivalence at the level of nets with boundaries.

\section*{Acknowledgement}
The authors  acknowledge the anonymous reviewers  for their careful reading of the manuscript and their insightful comments.  
Research supported by the EU Integrated Project 257414 {\sc ASCENS}, the Italian MIUR Project IPODS (PRIN 2008),  EU FP7-project MEALS, Italian MIUR Project CINA (PRIN 2010), ANPCyT Project BID-PICT-2008-00319, and UBACyT  20020090300122. 

%\bibliographystyle{abbrv}
%\bibliography{connectorsbibitems}

\appendix

\section{Proofs from Section~\ref{sec:nets} (C/E Nets with boundaries)}
\label{app:proof-netswithboundaries}

\begin{proof} [{\bf Proof of Lemma~\ref{lem:strongTransitionDecomposition}}]
Induction on the size of $|U\cup V|$. Base cases are when
$(U,V)$ is a minimal synchronisation, in that case
we are finished---the singleton set $\{(U,V)\}$ satisfies
the hypothesis. 
Else take
any minimal synchronisation $(U',V')$,
contained in $(U,V)$. Since $(U,V)$ and $(U',V')$ are synchronisations, 
$\target{U} = \source{V}$ and $\target{U'} = \source{V'}$ hold. Then, 
$\target{(U\backslash V)} = \target{U} \backslash \target{V} = \source{U'} \backslash \source{V'} = \source{(V\backslash V')}$.
Hence,  
$(U\backslash U', V\backslash V')$ is a synchronisation.
 Apply inductive hypothesis to
$(U\backslash U', V\backslash V')$ to obtain a set of minimal
synchronisations, to which we add $(U',V')$. All the conditions
required of the set are clearly satisfied. 
\end{proof}

\begin{proof} [{\bf Proof of Theorem~\ref{thm:netdecomposition}}]
($\Rightarrow$)
Suppose $\marking{M;N}{X+Y} \dtrans{\alpha}{\beta} \marking{M;N}{X'+Y'}$.
Then there exists a mutually independent set of minimal synchronisations
$\{(U_i,V_i)\}_{i\in I}$ such
that $\characteristic{\source{(\bigcup{U_i})}}=\alpha$
and $\characteristic{\target{(\bigcup{V_i})}}=\beta$. It follows that 
$\bigcup_{i}U_i$ is a mutually independent set of transitions
in $M$ and $\bigcup_{i}V_i$ is a mutually independent set of
transitions in $N$. Moreover 
$\target{(\bigcup_i{U_i})}=\bigcup_i{(\target{U_i})}=\bigcup_i{\source{V_i}}=
\source{(\bigcup_i{V_i})}$. Let $\gamma=\characteristic{\bigcup_i{\target{U_i}}}
=\characteristic{\bigcup_i{\source{V_i}}}$; we obtain
$\marking{M}{X} \dtrans{\alpha}{\gamma} \marking{M}{X'}$
and 
$\marking{Y}{Y} \dtrans{\gamma}{\beta} \marking{N}{Y'}$
as required.

($\Leftarrow$) If 
$\marking{M}{X} \dtrans{\alpha}{\gamma} \marking{M}{X'}$ and
$\marking{N}{Y} \dtrans{\gamma}{\beta} \marking{N}{Y'}$ then
there exist contention-free subsets 
$U\subseteq T_M$ and $V\subseteq T_N$ with
$\target{U}=\source{V}$. Using the conclusion of 
Lemma~\ref{lem:strongTransitionDecomposition} we obtain
a mutually independent set of transitions of $M;N$ that
induces the transition
$\marking{M;N}{X+Y} \dtrans{\alpha}{\beta} \marking{M;N}{X'+Y'}$.
\end{proof}

\begin{proof} [{\bf Proof of Proposition~\ref{pro:netcongruence}}]
We only show that 
\[A\Defeq \{(\marking{M^1}{X_1};\marking{N}{Y},\,
		    \marking{M^2}{X_2};\marking{N}{Y})
	        \;|\;
	        \marking{M^1}{X_1}\sim\marking{M^2}{X_2}\}\]
is a bisimulation.
If 
$\marking{M^1}{X_1};\marking{N}{Y} 
	\dtrans{\alpha}{\beta}
 		\marking{M^1}{X_1'};\marking{N}{Y'}$ 
then using the ``only-if'' direction
 of Theorem~\ref{thm:netdecomposition} we have 
$\marking{M^1}{X_1}\dtrans{\alpha}{\gamma}\marking{M^1}{X_1'}$
and $\marking{N}{Y}\dtrans{\gamma}{\beta}\marking{N}{Y'}$
for some $\gamma$.
Using the assumption, there exists $X_2'$ with 
$\marking{M^2}{X_2}\dtrans{\alpha}{\gamma}\marking{M^2}{X_2'}$
with $\marking{M^1}{X_1'}\sim\marking{M^2}{X_2'}$.
Then, using the ``if'' direction of Theorem~\ref{thm:netdecomposition}
we obtain that 
$\marking{M^2}{X_2};\marking{N}{Y}\dtrans{\alpha}{\beta}\marking{M^2}{X_2'};\marking{N}{Y'}$
and $(\marking{M^1}{X_1'};\marking{N}{Y'},\,\marking{M^2}{X_2'};\marking{N}{Y'})\in A$.
\end{proof}

\section{Proofs from Section~\ref{sec:ptboundaries} ({P/T} nets with boundaries)}
\label{app:proof-ptboundaries}

\begin{proof}[{\bf Proof of Lemma~\ref{lem:weakTransitionDecomposition}}]
Let $d = |\target{\mathcal{U}}| = |\source{\mathcal{V}}|$. We proceed
by induction on $d$.
The base case is $d=0$. This means
that whenever $\mathcal{U}(t)\neq 0$ then $|\target{t}|=0$
and whenever $\mathcal{V}(t)\neq 0$ then $|\source{t}|=0$.
The required family is then
\[
\{ ( \mathcal{U}(t), (\{t\},\emptyset) ) \}_{t\in T_M} 
\cup
\{ ( \mathcal{V}(t), (\emptyset, \{t\}) ) \}_{t\in T_N}.
\]
Now if $d>0$ and 
$(\mathcal{U},\mathcal{V})$ is a minimal synchronisation then we are finished, taking
the one member family $\{(1,(\mathcal{U},\mathcal{V}))\}$. Otherwise
let $(\mathcal{U}',\mathcal{V}')$ be any minimal synchronisation with
$\emptyset\neq\mathcal{U}'\subseteq \mathcal{U}$ and $\emptyset\neq\mathcal{V}'\subseteq \mathcal{V}$.
By definition we have $\target{\mathcal{U}'}=\source{\mathcal{V}'}$ and so
$\target{(\mathcal{U} - \mathcal{U}')}=\target{\mathcal{U}}-\target{\mathcal{U}'}=
\source{\mathcal{V}}-\source{\mathcal{V}'}=\source{(\mathcal{V}-\mathcal{V}')}$. We remark that 
$|\target{\mathcal{U'}}| = |\source{\mathcal{V'}}| > 0$ because  $(\mathcal{U}',\mathcal{V}')$ is minimal with
 $\mathcal{U}' \neq \emptyset$ and 
$\mathcal{V}'\neq \emptyset$.
Hence we apply the inductive hypothesis to 
$\mathcal{U} - \mathcal{U}'$
and 
$\mathcal{V} - \mathcal{V}'$ to obtain a family 
$F=\{(b_i,(\mathcal{U}_i,\mathcal{V}_i))\}_{i\in I}$ satisfying the expected requirements.
If $\exists i\in I$ with $(\mathcal{U}_i,\mathcal{V}_i)=(\mathcal{U}',\mathcal{V}')$
then the required family is 
$\{(b_j,(\mathcal{U}_j,\mathcal{V}_j)\}_{j\neq i}\cup \{(b_i+1,(\mathcal{U}_i,\mathcal{V}_i)\}$,
otherwise the required family is
$F \cup \{(1,(\mathcal{U}',\mathcal{V}')\}$.
\end{proof}

\medskip

\begin{proof} [{\bf Proof of Theorem~\ref{thm:ptnetdecomposition}}]~
\begin{enumerate}[(i)]
\item
$(\Rightarrow)$ 
If $M;N_{\mathcal{X}+\mathcal{Y}} \dtrans{\alpha}{\beta} M;N_{\mathcal{X}'+\mathcal{Y}'}$ then there exists
$\mathcal{W} \in\multiset{T_{M;N}}$, with 
$\characteristic{\source{\mathcal{W}}}=\alpha$ and 
$\characteristic{\target{\mathcal{W}}}=\beta$.
Define 
$\mathcal{W}_{M}\in \multiset{T_M}$ and $\mathcal{W}_{N}\in\multiset{T_N}$ as 
\[
\mathcal{W}_M = \bigcup_{(\mathcal{U},\mathcal{V})\in T_{M;N}} \mathcal{W}(\mathcal{U},\mathcal{V}) \cdot \mathcal{U} 
\qquad
\mathcal{W}_N = \bigcup_{(\mathcal{U},\mathcal{V})\in T_{M;N}} \mathcal{W}(\mathcal{U},\mathcal{V}) \cdot \mathcal{V}
\]
Analogously to Lemma~\ref{thm:netdecomposition} ($ii$), it can be shown that 
$\characteristic{\target{\mathcal{W}_M}} = \characteristic{\source{\mathcal{W}_N}}$. Finally, 
because $\pre{\mathcal{W}} \subseteq {\mathcal{X}+\mathcal{Y}}$, $\post{\mathcal{W}} \subseteq {\mathcal{X}'+\mathcal{Y}'}$ and  
${(\mathcal{X}}+\mathcal{Y})- \pre{\mathcal{W}} =
{(\mathcal{X}'+\mathcal{Y}')}- \post{\mathcal{W}}$, we conclude that
\begin{iteMize}{$-$}
\item $\pre{\mathcal{W}}_M \subseteq {\mathcal{X}}$, $\post{\mathcal{W}_M} \subseteq {\mathcal{X}'}$ and  
${\mathcal{X}}- \pre{\mathcal{W}_M} =
{\mathcal{X'}}- \post{\mathcal{W}_M}$, and
\item $\pre{\mathcal{W}}_N \subseteq {\mathcal{Y}}$, $\post{\mathcal{W}_N} \subseteq {\mathcal{Y}'}$ and  
${\mathcal{Y}}- \pre{\mathcal{W}_N} =
{\mathcal{Y}'}- \post{\mathcal{W}_N}$.
\end{iteMize}

Thus we have shown that $M_{\mathcal{X}} \dtrans{\alpha}{\gamma} M_{\mathcal{X}'}$
and 
$N_{\mathcal{Y}}  \dtrans{\gamma}{\beta} N_{\mathcal{Y}'} $.

\smallskip

$(\Leftarrow)$ 
Suppose 
$M_{\mathcal{X}} \dtrans{\alpha}{\gamma} M_{\mathcal{X}'}$ and 
$N_{\mathcal{Y}} \dtrans{\gamma}{\beta}  N_{\mathcal{Y}'}$ for some $\gamma\in \N^n$.
Then there exist $\mathcal{W}_M\in\multiset{T_M}$ such
that $\characteristic{\source{\mathcal{W}_M}}=\alpha$ and 
$\characteristic{\target{\mathcal{W}_M}}=\gamma$,
and $\mathcal{W}_N\in\multiset{T_N}$ such that
$\characteristic{\source{\mathcal{W}_N}}=\gamma$ and $\characteristic{\target{\mathcal{W}_N}}=\beta$.

By the conclusion of Lemma~\ref{lem:weakTransitionDecomposition}, there exists
a family $\{(b_i,(\mathcal{U}_i,\mathcal{V}_i)\}_{i\in I}$ where for each $i\in I$ we
have $b_i\in \N$ and $(\mathcal{U}_i,\mathcal{V}_i)\in Synch(M,N)$. Moreover
$\bigcup_{i\in I} b_i \cdot \mathcal{U}_i = \mathcal{W}_M$ and
$\bigcup_{i\in I} b_i \cdot \mathcal{V}_i = \mathcal{W}_N$.
Let $\mathcal{W} \Defeq \bigcup_{i\in I} b_i\cdot (\mathcal{U}_i,\mathcal{V}_i)$.
Clearly we have that $M;N_{\mathcal{X}+\mathcal{Y}}\dtrans{\alpha}{\beta}M;N_{\mathcal{X}'+\mathcal{Y}'}$,
as evidenced by $\mathcal{W}$.
\item Follows analogously to the previous case.
\end{enumerate}
\end{proof}

\section{Proofs from Section~\ref{sec:petritile} ({P/T Calculus})}
\label{app:proof-petritile}

\begin{proof} [{\bf Proof of Lemma~\ref{lemma:place-inv-weak}}]
$(\Rightarrow)$ The proof follows by induction on the structure of the derivation. First we note that the only applicable rules are  \ruleLabel{TkIO$_{n,h,k}$} and 
\ruleLabel{Weak}. Case \ruleLabel{TkIO$_{n,h,k}$}
follows immediately, since $k\leq n\leq n+h$. If the last applied rule is \ruleLabel{Weak*}, then the derivation has the following shape:
   \[
                   \namedrule{\cell{\tokensplace{n}}{\alpha'}{\beta'}{P'} \quad  \cell{P'}{\alpha''}{\beta''}{Q}}
                             {\cell{\tokensplace{n}}{\alpha'+\alpha''}{\beta'+\beta''}{Q}}
                             {Weak*}
   \] 

By inductive hypothesis on the first premise, we have $P'=\tokensplace{n'}$ with $\alpha'=h'$, $\beta' = k'$, $k' \leq n + h'$ and $n' = n+h'-k'$ ($1$). By inductive hypothesis on the second premise, we conclude $Q=\tokensplace{m}$  with $\alpha''=h''$, $\beta'' = k''$, $k'' \leq n' + h''$ and $m = n'+h''-k''$. We use ($1$) to substitute  $n'$ by $n+h'-k$  in $m$. Then, by rearranging terms we have $m = n'+h''-k'' = (n+h'-k')+ h'' - k'' = n+ (h'+ h'') - (k'+k'')$. Similarly,  $k'' \leq n' + h'' = (n+h'-k') + h''$ and 
hence $k' + k'' \leq n+(h'+ h'')$.

$(\Leftarrow)$ By \ruleLabel{TkIO$_{n,h,0}$} we have $\tokensplace{n} \dtransw{h}{0} \tokensplace{n+h}$. 
By rule  \ruleLabel{TkIO$_{n+h,0,k}$} we have $\tokensplace{n+h} \dtransw{0}{k} \tokensplace{m}$.
We conclude by applying \ruleLabel{Weak*}.
\end{proof}

\section{Proofs from Section~\ref{sec:netstosyntax} ({Translating nets to terms})}
\label{app:proof-cetranslations}

\begin{proof}[\bf Proof of Lemma~\ref{lem:inverseFunctionalForms}]
Here we concentrate on left inverse functional form
and the strong semantics. The proof for right inverse functional
forms is symmetric, and the argument for 
the weak semantics (and, thus, similarly for P/T calculus semantics)  
follows the same structure and relies on the
characterisation in Proposition~\ref{pro:weakConstants}.
Any function
$f:\underline{l}\to\underline{k}$ can be decomposed uniquely 
into $f=f_2\circ f_1\circ f_0$ where
\begin{enumerate}[(i)] 
\item $f_0:\underline{l}\to \underline{l}$ is a permutation
\item $f_1:\underline{l}\to \underline{m}$ is surjective and monotone (with
respect to the obvious ordering on elements of the ordinal)
and
\item $f_2:\underline{m}\to \underline{k}$ is
injective and monotone.
\end{enumerate}
Let $\mathrm{liff}_{f_0}:(l,l)$ be a term in $T_{\{\tw\}}$ whose behaviour is characterised
by 
\[
\mathrm{liff}_{f_0} \dtrans{\alpha}{\beta} \mathrm{liff}_{f_0}
\ \Leftrightarrow \exists U\subseteq \underline{l} \mbox{ s.t. }
\beta = \characteristic{U} \mbox{ and }
\alpha = \characteristic{f_0^{-1}(U)}
\]
Define $\mathrm{liff}_{f_1}:(l,m)\in T_{\{\codiag\}}$ 
as  
$\mathrm{liff}_{f_1}\Defeq \bigotimes_{i<m} \codiag_{|{f_1}^{-1}(i)|}$,
then it follows that
\[
\mathrm{liff}_{f_1}  \dtrans{\alpha}{\beta} \mathrm{liff}_{f_1}
\ \Leftrightarrow \exists U\subseteq \underline{m} \mbox{ s.t. }
\beta = \characteristic{U} \mbox{ and }
\alpha = \characteristic{f_1^{-1}(U)}.
\] 
Now let $\mathrm{liff}_{f_2}:(m,k)\in T_{\{\leftEnd\}}$ 
be $\mathrm{liff}_{f_2}\Defeq \bigotimes_{i<k} 
\begin{cases} 
\id & \mbox{if } f_2^{-1}(i)\neq\varnothing \\
\leftEnd & \mbox{otherwise.}
\end{cases}$,
it follows easily that
\[
\mathrm{liff}_{f_2}  \dtrans{\alpha}{\beta} \mathrm{liff}_{f_2}
\ \Leftrightarrow \exists U\subseteq \underline{k} \mbox{ s.t. }
\beta = \characteristic{U} \mbox{ and }
\alpha = \characteristic{f_2^{-1}(U)}.
\]
It follows that 
$\mathrm{liff}_{f}\Defeq \mathrm{liff}_{f_0};\mathrm{liff}_{f_1};\mathrm{liff}_{f_2}$
is in left inverse functional form and has the required semantics.
\end{proof}

\begin{proof}[\bf Proof of Lemma~\ref{lem:directFunctionalForms}]
We concentrate only on right direct functional forms and strong
semantics. The other cases follow as in the
proof Lemma~\ref{lem:inverseFunctionalForms}.
Also as in that proof,
we decompose $f=f_2\circ f_1\circ f_0$ where 
$f_0:\underline{l}\to \underline{l}$ is a permutation, 
$f_1:\underline{l}\to \underline{m}$ is surjective and monotone
and $f_2:\underline{m}\to \underline{k}$ is injective and monotone. 

Let $\mathrm{rdff}_{f_0}\in T_{\tw}$ be a term such that
\[
\mathrm{rdff}_{f_0} \dtrans{\alpha}{\beta} \mathrm{rdff}_{f_0}
\ \Leftrightarrow \exists U\subseteq \underline{l} \mbox{ s.t. }
\alpha = \characteristic{U} \mbox{ and }
\beta = \characteristic{f_0(U)}
\]
Next let 
$\mathrm{rdff}_{f_1}\typ\sort{l}{m}\Defeq \bigotimes_{i<m} \lcodiag_{|f_0^{-1}(i)|}$,
which is clearly in $T_{\{\lcodiag\}}$, then
\[
\mathrm{rdff}_{f_1} \dtrans{\alpha}{\beta} \mathrm{rdff}_{f_1}
\ \Leftrightarrow 
\exists U\subseteq \underline{l} \mbox{ s.t. }
\alpha = \characteristic{U} \mbox{ and }
\beta = \characteristic{f_1(U)}
\]
The third ingredient is
$\mathrm{rdff}_{f_2}\typ\sort{m}{k}\Defeq \bigotimes_{i<k}
\begin{cases} 
	\id & \mbox{if } f_2^{-1}(i) \neq \varnothing \\ 
	\lzero & \mbox{otherwise.} \end{cases}$
which is in $T_{\{\lzero\}}$ and whose behaviour is clearly
\[
\mathrm{rdff}_{f_2} \dtrans{\alpha}{\beta} \mathrm{rdff}_{f_2}
\ \Leftrightarrow 
\exists U\subseteq \underline{l} \mbox{ s.t. }
\alpha = \characteristic{U} \mbox{ and }
\beta = \characteristic{f_2(U)}.
\]
Finally let 
$\mathrm{rdff}_f\Defeq \mathrm{rdff}_{f_0}\comp 
\mathrm{rdff}_{f_1} \comp \mathrm{rdff}_{f_2}$.
\end{proof}

\begin{proof}[\bf Proof of Theorem~\ref{thm:strongcenetstosyntax}]
(i) If $\marking{N}{X}\dtrans{\alpha}{\beta}\marking{N}{Y}$ then there
exists a set $U\subseteq \underline{t}$ of mutually independent transitions such that
$\marking{N}{X} \rightarrow_{U} \marking{N}{Y}$, with $\alpha=\characteristic{\source{U}}$
and $\beta=\characteristic{\target{U}}$.
Using the conclusion of Lemma~\ref{lem:encodingTech}, we have
%It is easy to verify that the conditions 
%$X\subseteq \pre{U}$, $Y\subseteq \post{U}$ and $X\backslash \pre{U}=Y\backslash\post{U}$
%are equivalent to the existence of a transition
\[
w_{N_X} \dtrans{\characteristic{\post{U}}}{\characteristic{\pre{U}}} w_{N_Y}.
\]
Now, using the conclusion of Lemma~\ref{lem:relationalForms} 
and \ruleLabel{Cut}
we obtain transition
\[
\rho_{\post{-}}\comp w_{N_X} \comp \lambda_{\pre{-}} 
\dtrans{\characteristic{U}}{\characteristic{U}}
\rho_{\post{-}}\comp w_{N_Y} \comp \lambda_{\pre{-}} 
\]
and subsequently
\[
\nabla_t \comp
\rho_{\post{-}}\comp w_{N_X} \comp \lambda_{\pre{-}} \comp \Delta_t
\dtrans{\characteristic{U}\characteristic{U}}{\characteristic{U}\characteristic{U}}
\nabla_t \comp \rho_{\post{-}}\comp w_{N_Y} \comp \lambda_{\pre{-}} \comp \Delta_t
\]
Certainly $\#_t\dtrans{\characteristic{U}}{\characteristic{U}}\#_t$, thus
using the semantics of $d_t$ and $e_t$ we obtain:
\[
T_{\marking{N}{X}} \dtrans{\characteristic{\source{U}}}{\characteristic{\target{U}}} T_{\marking{N}{Y}}
\]
as required.

\smallskip
\noindent
(ii) If $T_{\marking{N}{X}}\dtrans{\alpha}{\beta} Q$ then 
$Q=(d_t\ten \lambda_{\source{-}})\comp Q_1 \comp (e_t\ten \rho_{\target{-}})$
and
\[
{\#}_t\ten (\nabla_t\comp \rho_{\post{-}} \comp w_{N_X} \comp \lambda_{\pre{-}} \comp \Delta_t)
\dtrans{\characteristic{U}\characteristic{U}\characteristic{V}}
{\characteristic{U'}\characteristic{U'}\characteristic{V'}} 
Q_1
\]
For some $U,V,U',V'\subseteq \underline{t}$ with $\alpha=\characteristic{\source{V}}$
and $\beta=\characteristic{\target{V'}}$. The structure of \ruleLabel{Ten} and the
semantics of $\#_t$ imply that $U=U'$, mutually independent, and $Q_1=\id_t\ten Q_2$ with 
\[
\nabla_t\comp \rho_{\post{-}} \comp w_{N_X} \comp \lambda_{\pre{-}} \comp \Delta_t
\dtrans{\characteristic{U}\characteristic{V}}
{\characteristic{U}\characteristic{V'}} Q_2
\]
Now the semantics of $\Delta_t$ implies that $U=V$ and conversely, the semantics of
$\nabla_t$ that $U=V'$, moreover $Q_2=\nabla_t\comp Q_3\comp \delta_t$ with
\[
\rho_{\post{-}} \comp w_{N_X} \comp \lambda_{\pre{-}} 
\dtrans{\characteristic{U}}{\characteristic{U}}
Q_3
\]
Finally, using the conclusion of Lemma~\ref{lem:relationalForms}, we obtain 
$Q_3=\rho_{\post{-}}\comp Q_4\comp \lambda_{\pre{-}}$ and  
\[
w_{N_X} \dtrans{\characteristic{\post{U}}}{\characteristic{\pre{U}}} Q_4
\]
In particular, we obtain that $Q_4=w_{N_Y}$ and $\marking{N}{X}\dtrans{\alpha}{\beta}\marking{N}{Y}$.
\end{proof}

\begin{proof}[\bf Proof of Lemma~\ref{lem:amplifiers}]
Here we give the proof for right amplifiers, the argument for left
amplifiers follows by symmetry. The proof
follows by induction: the base case $!1=\id$ is obvious.
Now 
\begin{align*}
&!(k+1)\dtransw{a}{b} !(k+1) &\Leftrightarrow \\
&(I\ten !k)\comp\lcodiag \dtransw{aa}{b}(I\ten !k)\comp\lcodiag &\Leftrightarrow \\
&!k\dtransw{a}{b'} !k \mbox{ and }b=a+b'. 
\end{align*}
By the inductive
hypothesis $b'=ka$ and so $b=(k+1)a$ as required.
\end{proof}

\section{Proofs from Section~\ref{sec:petri-tile-calculus} (Petri Tile Calculus)}
\label{app:proof-ptiles}

The following technical lemmas state the admissible tiles for basic connectors. 
 
\begin{lem}\label{lem:basicageof1}
For any $h,k\in \nat$ we have:
\[
\tw \dtrans{hk}{kh} \tw \qquad
\rightEnd \dtrans{k}{} \rightEnd \qquad
\leftEnd \dtrans{}{k} \leftEnd \qquad
\diag \dtrans{k}{k \labelSep k} \diag \qquad
\codiag \dtrans{k \labelSep k}{k} \codiag \qquad 
\ldiag \dtrans{h+k}{ h\labelSep k} \ldiag \qquad
\lcodiag \dtrans{h\labelSep k}{h+k} \lcodiag
\]
\[
\tw \dtrans{\overline{hk}}{\overline{kh}} \tw \qquad
\rightEnd \dtrans{\overline{k}}{} \rightEnd \qquad
\leftEnd \dtrans{}{\overline{k}} \leftEnd \qquad
\diag \dtrans{\overline{k}}{\overline{k \labelSep k}} \diag \qquad
\codiag \dtrans{\overline{k \labelSep k}}{\overline{k}} \codiag \qquad 
\ldiag \dtrans{\overline{h+k}}{ \overline{h\labelSep k}} \ldiag \qquad
\lcodiag \dtrans{\overline{h\labelSep k}}{\overline{h+k}} \lcodiag
\]
\end{lem}
\begin{proof} 
The proof is trivial, by construction. We show the cases of 
\[
\lcodiag \dtrans{h\labelSep k}{h+k} \lcodiag\qquad
\lcodiag \dtrans{\overline{h\labelSep k}}{\overline{h+k}} \lcodiag
\]
\noindent
the other cases are analogous.

The two instances of rule $\lcodiag$ are:
\[
\lcodiag \dtrans{0\labelSep 1}{1} \lcodiag\qquad
\lcodiag \dtrans{1\labelSep 0}{1} \lcodiag
\]
By composing $h$ instances of the former and $k$ of the latter using \ruleLabel{Vert}, we have (recall that the observation $h$ is the sequential composition of $h$ observations $1$):
\[
\lcodiag \dtrans{0\labelSep h}{h} \lcodiag\qquad
\lcodiag \dtrans{k\labelSep 0}{k} \lcodiag
\]
Finally we compose the above rules using \ruleLabel{Vert} (recall that $0$ is the identity, hence $0;h = h = h;0$):
\[
\lcodiag \dtrans{k\labelSep h}{h+k} \lcodiag\qquad
\]
Then, we can use the instance of \ruleLabel{Epoch}:
\[
\lcodiag \dtrans{\tau\labelSep \tau}{\tau} \lcodiag\qquad
\]
By \ruleLabel{Vert} we have:
\[
\lcodiag \dtrans{(\tau\labelSep \tau);(h\labelSep k)}{\tau;(h+k)} \lcodiag\qquad
\mbox{i.e.}\qquad
\lcodiag \dtrans{(\tau;h)\labelSep (\tau;k)}{\tau;(h+k)} \lcodiag\qquad
\mbox{i.e.}\qquad
\lcodiag \dtrans{\overline{h\labelSep k}}{\overline{h+k}} \lcodiag
\]
\end{proof}

\begin{lem}
For any observation $a$ we have:
\[
\rightEnd \dtrans{a}{} \rightEnd \qquad
\leftEnd \dtrans{}{a} \leftEnd \qquad
\diag \dtrans{a}{a \labelSep a} \diag \qquad
\codiag \dtrans{a \labelSep a}{a} \codiag \qquad 
\]
\end{lem}
\begin{proof}
By straightforward induction on the structure of $a$. If $a=0$ the thesis follows by \ruleLabel{Idle}.
If $a=\tau;a'$ then the thesis follows by \ruleLabel{Epoch}, \ruleLabel{Vert} and the inductive hypothesis on $a'$.
If $a=1;a'$ then the thesis follows by the basic tiles (for $\rightEnd$, $\leftEnd$, $\diag$, $\codiag$), \ruleLabel{Vert} and the inductive hypothesis on $a'$.
\end{proof}

\begin{lem}\label{lem:coetaneousdiag}
For any observations $a,b,c$ we have $a=b=c$ if and only if:
\[
\diag \dtrans{a}{b \labelSep c} \diag \qquad
\codiag \dtrans{b \labelSep c}{a} \codiag \qquad 
\]
\end{lem}
\begin{proof}
The only tiles applicable to the source configuration $\diag$ (resp. $\codiag$) are $\diag$ (resp. $\codiag$), \ruleLabel{Epoch}, \ruleLabel{Idle} and \ruleLabel{Vert}.
The equality of the observations is an invariant of the application of such rules.
\end{proof}

\begin{lem}\label{lem:coetaneoustw}
For any observations $a,b,c,d$ we have that $a=d$, $b=c$ and $a,b$ are coetaneous if and only if:
\[
\tw \dtrans{ab}{cd} \tw \qquad
\]
\end{lem}
\begin{proof}
The only tiles applicable to the source configuration $\tw$ are $\tw$, \ruleLabel{Epoch}, \ruleLabel{Idle} and \ruleLabel{Vert}.
The requirement on the observations is an invariant of the application of such rules.
\end{proof}

\begin{lem}\label{lem:coetaneousldiag}
For any observations $a,b,c$ we have that $a,b,c$ are coetaneous and $a=b+c$ if and only if:
\[
\ldiag \dtrans{a}{ b\labelSep c} \ldiag \qquad
\lcodiag \dtrans{b\labelSep c}{a} \lcodiag
\]
\end{lem}
\begin{proof}
The only tiles applicable to the source configuration $\ldiag$ (resp. $\lcodiag$) are $\ldiag$ (resp. $\lcodiag$), \ruleLabel{Epoch}, \ruleLabel{Idle} and \ruleLabel{Vert}.
The constraint on the observations is an invariant of the application of such rules.
\end{proof}

\begin{lem}\label{lem:idleobs}
An observation $a$ is idle if and only if:
\[
\lzero \dtrans{}{a} \lzero \qquad
\rzero \dtrans{a}{} \rzero
\]
\end{lem}
\begin{proof}
The only tiles applicable to the source configuration $\lzero$ (resp. $\rzero$) are \ruleLabel{Epoch}, \ruleLabel{Idle} and \ruleLabel{Vert}.
The constraint on the observations is an invariant of the application of such rules.
\end{proof}

\begin{lem} \label{lem:tokenvalid}
If $\token \dtrans{}{b} P'$, then either $b$ is idle and $P'=\token$ or $b = b';\bar{1};b''$ with $b',b''$ idle and $P' = \lzero$.
\end{lem}
\begin{proof}
Immediate, by noting the only rules applicable to $\token$ are \ruleLabel{TkO}, \ruleLabel{Epoch}, \ruleLabel{Idle} and \ruleLabel{Vert}.
\end{proof}

\begin{lem}\label{lem:ntokensplus0}
For any $n\in \nat$ we have $\tokensplace{n} \simeq_{\mathsf{tb}} (\tokensplace{n}\otimes \lzero);\lcodiag$.
\end{lem}
\begin{proof}
We observe that $\id \simeq_{\mathsf{tb}} (\id\otimes \lzero);\lcodiag$. In fact by Lemma~\ref{lem:idleobs} and Lemma~\ref{lem:coetaneousldiag} for any $a$ and for the unique idle $c$ that is coetaneous of $a$ we know that:
\[
\id \dtrans{a}{a} \id \qquad
\lzero \dtrans{}{c} \lzero \qquad
\lcodiag \dtrans{a\labelSep c}{a} \lcodiag
\]
Therefore, by \ruleLabel{Mon} and \ruleLabel{Vert} we have that
\[
(\id\otimes \lzero);\lcodiag \dtrans{a}{a} (\id\otimes \lzero);\lcodiag
\]
is the only admissible move. Hence the relation 
$\{ (\id , (\id\otimes \lzero);\lcodiag) , ((\id\otimes \lzero);\lcodiag, \id)\}$ is a tile bisimulation and $\id \simeq_{\mathsf{tb}} (\id\otimes \lzero);\lcodiag$.

Since tile bisimilarity is a congruence and $\tokensplace{n} = \tokensplace{n} ; \id$ we can conclude that
 $\tokensplace{n} = \tokensplace{n} ; \id \simeq_{\mathsf{tb}} \tokensplace{n} ; (\id\otimes \lzero);\lcodiag = (\tokensplace{n}\otimes \lzero);\lcodiag$.
\end{proof}

\begin{lem}
$\lzero;\tokensplace{0} \simeq_{\mathsf{tb}} \lzero$.
\end{lem}
\begin{proof}
The only rules applicable to $\tokensplace{0}$ are \ruleLabel{TkI}, \ruleLabel{Epoch} and \ruleLabel{Idle}, but \ruleLabel{TkI} would require $1$ as trigger, that $\lzero$ cannot produce as effect. Therefore all and only observations that $\lzero;\tokensplace{0}$ can give rise to are idle observations, i.e. the same as $\lzero$.
\end{proof}

\begin{lem}\label{lem:onetokenalone}
$\token \simeq_{\mathsf{tb}} \lzero;\tokensplace{1}$.
\end{lem}
\begin{proof}
Recall that $\tokensplace{1} \defeq (\tokensplace{0}\otimes \token);\lcodiag$.
Then, \\
$
\lzero;\tokensplace{1} = 
\lzero;(\tokensplace{0}\otimes \token);\lcodiag = 
((\lzero;\tokensplace{0})\otimes \token);\lcodiag \simeq_{\mathsf{tb}}
(\lzero \otimes \token);\lcodiag =
\token;(\lzero \otimes \id);\lcodiag \simeq_{\mathsf{tb}}
\token;\id =
\token
$ .
\end{proof}

\begin{lem}\label{lem:nminusk}
For any $k\leq n$, $\tokensplace{n} \dtrans{\bar{0}}{\bar{k}} P'$ with $P' \simeq_{\mathsf{tb}} \tokensplace{n-k}$.
Moreover, for any $k,n$, if $\tokensplace{n} \dtrans{\bar{0}}{\bar{k}} P'$ then $P' \simeq_{\mathsf{tb}} \tokensplace{n-k}$  and $k\leq n$.
\end{lem}
\begin{proof}
The first part is by induction on $k$. For $k=0$ the thesis holds trivially by \ruleLabel{Epoch}.
For $k>0$, we have that $n-1\geq k-1\geq 0$ and $\tokensplace{n} = (\tokensplace{n-1} \otimes \token);\lcodiag$. By inductive hypothesis we have  $\tokensplace{n-1} \dtrans{\bar{0}}{\overline{k-1}} P''$ with $P'' \simeq_{\mathsf{tb}} \tokensplace{n-k}$. We also know that $\token \dtrans{}{\bar{1}} \lzero$ by \ruleLabel{TkO} and that $\lcodiag \dtrans{(\overline{k-1})\bar{1}}{\bar{k}} \lcodiag$ by Lemma~\ref{lem:coetaneousldiag}.
Therefore we conclude by letting $P' = (P'' \otimes \lzero);\lcodiag$, and by \ruleLabel{Mon} and \ruleLabel{Hor} and Lemma~\ref{lem:ntokensplus0} we have $\tokensplace{n} \dtrans{\bar{0}}{\bar{k}} (P'' \otimes \lzero);\lcodiag \simeq_{\mathsf{tb}} (\tokensplace{n-k} \otimes \lzero);\lcodiag \simeq_{\mathsf{tb}} \tokensplace{n-k}$.

For the second part, we proceed by induction on $n+k$.
For $n+k=0$ the thesis follows immediately because the only possibility is that  $\tokensplace{0} \dtrans{\bar{0}}{\bar{0}} \tokensplace{0}$ by \ruleLabel{Epoch}, \ruleLabel{Idle} and \ruleLabel{Vert}.
Let $n+k>0$. If $k>n$, then it is not possible that $\tokensplace{n} \dtrans{\bar{0}}{\bar{k}} P'$, because 
the term $\tokensplace{n}$ does not contain enough tokens $\token$. Therefore, it must be $k\leq n$ and $n>0$.
%Let $\token^{n}$ defined inductively as $\token^{1} = (\id \otimes \token);\lcodiag$ and $\token^{n+1} = (\token^{n}\otimes \token);\lcodiag$. Hence $\tokensplace{n} = \tokensplace{0};\token^{n}$.
If $k=0$ the thesis follows trivially.
Otherwise, let $k>0$, $h = k-1$, $m = n-1 \geq 0$ and $\tokensplace{n} = (\tokensplace{m} \otimes \token) ; \lcodiag$.
If $\tokensplace{n} \dtrans{\bar{0}}{\bar{k}} P'$, then (using \ruleLabel{Mon} and \ruleLabel{Hor}):
\begin{iteMize}{$-$}
\item either $\tokensplace{m} \dtrans{\bar{0}}{\bar{k}} P''$, $\token\dtrans{\bar{0}}{\bar{0}} \token$ and $\lcodiag\dtrans{\overline{k0}}{\bar{k}} \lcodiag$,
\item or $\tokensplace{m} \dtrans{\bar{0}}{\bar{h}} P''$, $\token\dtrans{\bar{0}}{\bar{1}} \lzero$ and $\lcodiag\dtrans{\overline{h1}}{\bar{k}} \lcodiag$.
\end{iteMize}
In the first case, by inductive hypothesis, $P'' \simeq_{\mathsf{tb}} \tokensplace{m-k}$ and therefore
\[
P' =(P'' \otimes \token) ; \lcodiag \simeq_{\mathsf{tb}} (\tokensplace{m-k} \otimes \token) ; \lcodiag = \tokensplace{m+1-k} = \tokensplace{n-k} .
\]

In the second case, by inductive hypothesis, $P'' \simeq_{\mathsf{tb}} \tokensplace{m-h} = \tokensplace{n-k}$ and therefore
$P' =(P'' \otimes \lzero) ; \lcodiag \simeq_{\mathsf{tb}} (\tokensplace{n-k} \otimes \lzero) ; \lcodiag \simeq_{\mathsf{tb}} \tokensplace{n-k}$ by Lemma~\ref{lem:ntokensplus0}.
\end{proof}

\begin{lem}\label{lem:nplush}
For any $h$, $\tokensplace{n} \dtrans{h}{0} \tokensplace{n+h}$.
Moreover, for any $h,k$, if $\tokensplace{n} \dtrans{h}{k} P'$ then $P' = \tokensplace{n+h}$ and $k=0$.
\end{lem}
\begin{proof}
Both parts are proved by  induction on $h$.
\end{proof}

\begin{lem}\label{lem:bar0plush}
For any $h$, $\tokensplace{0} \dtrans{\bar{h}}{\bar{0}} \tokensplace{h}$.
Moreover, for any $h$, if $\tokensplace{0} \dtrans{\bar{h}}{\bar{0}} P'$ then $P' = \tokensplace{h}$.
\end{lem}
\begin{proof}
By Lemma~\ref{lem:nplush}, using \ruleLabel{Epoch} and \ruleLabel{Vert}.
\end{proof}

\begin{proof}[\bf Proof of Lemma~\ref{lem:nminuskplush}]
The proof is along the lines of Lemma~\ref{lem:nminusk}, showing that for any $k\leq n$, $\tokensplace{n} \dtrans{\bar{0}}{\bar{k}} P'$ with $P' \simeq_{\mathsf{tb}} \tokensplace{n-k}$ and that for any $k,n$, if $\tokensplace{n} \dtrans{\bar{0}}{\bar{k}} P'$ then $P' \simeq_{\mathsf{tb}} \tokensplace{n-k}$  and $k\leq n$;
but it additionally exploits Lemma~\ref{lem:bar0plush}.
\end{proof}

\begin{lem}\label{lem:nminuskplusbarh}
For any $h,k,n$ there is no $P'$ such that $\tokensplace{n} \dtrans{\bar{h}}{k} P'$.
\end{lem}
\begin{proof}
By induction on $n$, if $\tau$ is observed in the left interface, then it must be observed on the right interface too (even when $h=k=0$). 
\end{proof}

\begin{lem}\label{lem:nminuskplushbar}
For any $h,k,n$, if $\tokensplace{n} \dtrans{h;\tau}{\bar{k}} P'$ then $P' \simeq_{\mathsf{tb}} \tokensplace{n+h-k}$ and $k\leq n+h$.
\end{lem}
\begin{proof}
The thesis follows as a combination of Lemma~\ref{lem:nplush} and Lemma~\ref{lem:nminuskplush} using \ruleLabel{Vert} (after noting that $h;\bar{0} = h;\tau$ and $k=0;k$).
\end{proof}

\begin{lem} \label{lem:ncoeteneous}
If $\tokensplace{n} \dtrans{a}{b} P'$, then $a$ and $b$ are coetaneous and $P' \simeq_{\mathsf{tb}} \tokensplace{n+h-k}$ for $h=\counttok{a}$ and $k=\counttok{b}$.
\end{lem}
\begin{proof}
By induction on $\ageof{a}$, exploiting Lemmas~\ref{lem:nplush}--\ref{lem:nminuskplushbar}.
\end{proof}

\begin{lem} \label{lem:allcoeteneous}
Let $P:(h,l)$ be any connector. If $P \dtrans{a}{b} P'$, then there exist idle $a',b'$ such that 
$P \dtrans{a;a'}{b;b'} P'$ with $a;a'$ and $b;b'$ valid and coetaneous.
\end{lem}
\begin{proof}
%The proof is analogous to that of later Lemma~\ref{lem:statelesscoeteneous} and thus omitted.
By structural induction on $P$. 

If $P$ is $\token$, then $a$ and $b$ are trivially valid and coetaneous by Lemma~\ref{lem:tokenvalid}.

If $P$ is $\tokensplace{0}$, then $a$ and $b$ are coetaneous by Lemma~\ref{lem:ncoeteneous} and they are trivially valid because $a:(1,1)$ and $b:(1,1)$.

If $P$ is a basic stateless connector, then $a$ and $b$ are valid and coetaneous by Lemma~\ref{lem:coetaneousbasic}.

If $P=Q\otimes R$ then it must be the case that $Q \dtrans{a_{1}}{b_{1}} Q'$ and $R \dtrans{a_{2}}{b_{2}} R'$ with $a=a_{1}a_{2}$, $b=b_{1}b_{2}$.
By inductive hypothesis:
\begin{iteMize}{$-$}
\item there must exist idle $a'_{1},b'_{1}$ such that  $Q \dtrans{a_{1};a'_{1}}{b_{1};b'_{1}} Q'$ with $a_{1};a'_{1}$ and $b_{1};b'_{1}$ valid and coetaneous;
\item there must exist idle $a'_{2},b'_{2}$ such that  $R \dtrans{a_{2};a'_{2}}{b_{2};b'_{2}} R'$ with $a_{2};a'_{2}$ and $b_{2};b'_{2}$ valid and coetaneous.
\end{iteMize}
Let $k_{1}=\ageof{(a_{1};a'_{1})(b_{1};b'_{1})}$ and $k_{2}=\ageof{(a_{2};a'_{2})(b_{2};b'_{2})}$.
If $k_{1} = k_{2}$ then $(a_{1};a'_{1})(a_{2};a'_{2})$ and $(b_{1};b'_{1})(b_{2};b'_{2})$ are (valid and) coetaneous and we are done.
Otherwise, assume without loss of generality that $k_{1} > k_{2}$.
Then, by applying \ruleLabel{Epoch} and \ruleLabel{Vert} for $k_{1}-k_{2}$ times to $R'$, we have 
$R' \dtrans{a'_{3}}{b'_{3}} R'$ for the unique idle and coetaneous $a'_{3},b'_{3}$ such that $\ageof{a'_{3}b'_{3}} = k_{1}-k_{2}$.
Then, by \ruleLabel{Vert}, $R \dtrans{a_{2};a'_{2};a'_{3}}{b_{2};b'_{2};b'_{3}} R'$. Now, it is obvious that $a_{2};a'_{2};a'_{3}$ and $b_{2};b'_{2};b'_{3}$ are coetaneous and that $\ageof{(a_{2};a'_{2};a'_{3})(b_{2};b'_{2};b'_{3})} = k_{1}$.
Therefore we let $a' = a'_{1}(a'_{2};a'_{3})$ and $b' = b'_{1}(b'_{2};b'_{3})$ and we are done.

If $P=Q; R$ then it must be the case that $Q \dtrans{a}{c} Q'$ and $R \dtrans{c}{b} R'$.
By inductive hypothesis:
\begin{iteMize}{$-$}
\item there must exist idle $a'_1,c'_1$ such that 
$Q \dtrans{a;a'_1}{c;c'_1} Q'$ with $a;a'_1$ and $c;c'_1$ valid and coetaneous;
\item there must exist idle $c'_2,b'_2$ such that 
$R \dtrans{c;c'_2}{b;b'_2} R'$ with $c;c'_2$ and $b;b'_2$ valid and coetaneous.
\end{iteMize}
Let $k_{1}=\ageof{c;c'_1}$ and $k_{2}=\ageof{c;c'_2}$.
If $k_{1} = k_{2}$ then $c'_1=c'_2$ and, by transitivity, $a;a'_1$ and $b;b'_1$ are (valid and) coetaneous and we are done.
Otherwise, assume without loss of generality that $k_{1} > k_{2}$.
Then, by applying \ruleLabel{Epoch} and \ruleLabel{Vert} for $k_{1}-k_{2}$ times to $R'$, we have 
$R' \dtrans{c'_3}{b'_3} R'$ for the unique idle and coetaneous $c'_3,b'_3$ such that $\ageof{c'_3 b'_3} = k_{1}-k_{2}$.
Then, by \ruleLabel{Vert}, $R \dtrans{c;c'_2;c'_3}{b;b'_2;b'_3} R'$. Now, it is obvious that $c;c'_1 = c;c'_2;c'_3$.
Therefore we let $a' = a'_1$ and $b' = b'_2;b'_3$ and we are done.
\end{proof}

\begin{lem}\label{lem:sliceone}
If $P \dtrans{a;\bar{c}}{b} P'$ with $a;\bar{c}$ and $b$ (valid and) coetaneous then there exist $d,\bar{e},P''$ such that $P \dtrans{a}{d} P''$, $P'' \dtrans{\bar{c}}{\bar{e}} P'$, $b=d;\bar{e}$, and $\ageof{a} = \ageof{d}$.
\end{lem}
\begin{proof}
We proceed by structural induction on $P$. 
The most interesting case is that of sequential composition.
Let $P=Q; R$, with $Q \dtrans{a;\bar{c}}{x} Q'$, $R \dtrans{x}{b} R'$, and $P'=Q';R'$. By inductive hypothesis:
\begin{iteMize}{$-$}
\item there exist $y,\bar{z},Q''$ such that $Q \dtrans{a}{y} Q''$, $Q'' \dtrans{\bar{c}}{\bar{z}} Q'$, $x=y;\bar{z}$, and $\ageof{a} = \ageof{y}$.
\end{iteMize}
Since $x=y;\bar{z}$, by inductive hypothesis on $R \dtrans{y;\bar{z}}{b} R'$:
\begin{iteMize}{$-$}
\item there exist $d,\bar{e},R''$ such that $R \dtrans{y}{d} R''$, $R'' \dtrans{\bar{z}}{\bar{e}} R'$, $b=d;\bar{e}$, and $\ageof{y} = \ageof{d}$.
\end{iteMize}
Then, by applying \ruleLabel{Hor} we obtain:
\begin{iteMize}{$-$}
\item $P=Q;R \dtrans{a}{d} Q'';R''$, 
\item $Q'';R'' \dtrans{\bar{c}}{\bar{e}} Q';R'$.
\end{iteMize}
Then, we take $P'' = Q'';R''$ and we are done.
\end{proof}

\begin{lem}\label{lem:slicemanyvalid}
If $P \dtrans{a}{b} P'$ with $a$ and $b$ (valid and) coetaneous, then there exist
$k$ elementary observations $a_{1},...,a_{k}$ with $\bar{a}_{1};\cdots;\bar{a}_{k} = \bar{a}$,
$k$ elementary observations $b_{1},...,b_{k}$ with $\bar{b}_{1};\cdots;\bar{b}_{k} = \bar{b}$, and  
$k$ terms $P_{1},P_{2},...,P_{k}$ such that
$P \dtrans{\bar{a}_{1}}{\bar{b}_{1}} P_{1}$,
$P_{1} \dtrans{\bar{a}_{2}}{\bar{b}_{2}} P_{2}$,
...,
$P_{k-1} \dtrans{\bar{a}_{k}}{\bar{b}_{k}} P_{k} = P'$.
\end{lem}
\begin{proof}
By \ruleLabel{Epoch} and \ruleLabel{Vert} we have $P \dtrans{\bar{a}}{\bar{b}} P'$.
Then the thesis follows by induction on $\ageof{\bar{a}}$, exploiting Lemma~\ref{lem:sliceone}.
The base case is trivial. For the inductive case, let $\ageof{\bar{a}} = k + 1$ for some $k\geq 0$ and assume the thesis is valid for any $c$ with $\ageof{\bar{c}} < n$. By $\ageof{\bar{a}} = k + 1$ we know that there exist epochs $a_{1},...,a_{k+1}$ such that $\bar{a} = \bar{a}_{1};\cdots;\bar{a_{k+1}}$. Let $\bar{c} = \bar{a}_{2};\cdots;\bar{a}_{k+1}$. By hypothesis we have $P \dtrans{\bar{a}_{1};\bar{c}}{b} P'$. By Lemma~\ref{lem:sliceone}, there exist $d,\bar{e},P''$ such that $P \dtrans{a}{d} P''$, $P'' \dtrans{\bar{c}}{\bar{e}} P'$, $\bar{b}=d;\bar{e}$, and $\ageof{d} = \ageof{a} = 1$.
Then, take $\bar{b_{1}} = d$ and $P_{1}=P''$. By inductive hypothesis about 
$P_{1} \dtrans{\bar{c}}{\bar{e}} P'$ we have that there exist suitable
elementary observations $b_{2},...,b_{k+1}$ and $P_{2},...,P_{k+1}$ such that 
$P_{1} \dtrans{\bar{a}_{2}}{\bar{b}_{2}} P_{2}$,
...,
$P_{k} \dtrans{\bar{a}_{k+1}}{\bar{b}_{k+1}} P_{k+1} = P'$ and we are done.
\end{proof}

\begin{defi}
We let $\updownarrow$ denote the least congruence\footnote{By requiring $\updownarrow$ to be a congruence, we are implicitly assuming that it is an equivalence relation (reflexive, symmetric, transitive) and that it is preserved by the sequential and parallel composition of observations.} on observations defined by the following rules:
\[
\derivationRule{a:(1,1)}{a\; \updownarrow\; \tau;a}{After}\qquad
\derivationRule{a:(1,1)}{a\; \updownarrow\; a;\tau}{Before}
\]
\end{defi}

As an easy invariant preserved by $\updownarrow$, observe that whenever $a\updownarrow a'$ the number of $1$'s occurring in $a$ is the same as the number of $1$'s occurring in $a'$.
Roughly, given a non valid observation $a$, we can always find a valid observation $a' \updownarrow a$ that differs from $a$ by the insertion of one or more $\tau$ to align different epochs. For example, given 
$a = \bar{2} (\bar{1};\bar{0};\bar{3}) (\bar{4};\bar{7})$ and 
$a' = (\bar{0};\bar{2};\bar{0}) (\bar{1};\bar{0};\bar{3}) (\bar{4};\bar{0};\bar{7})$ we have that $a \updownarrow a'$.
Also
$a'' = (\bar{0};\bar{2};\bar{0};\bar{0}) (\bar{0};\bar{1};\bar{0};\bar{3}) (\bar{4};\bar{0};\bar{7};\bar{0})$ is such that $a\updownarrow a''$.
In particular, note that for any $a:(h,h)$ and $b:(h,h)$ we have $\bar{a} \updownarrow a$ and $a;b \updownarrow a \updownarrow b;a$ if and only if $b$ is idle.
 
%\begin{lem}
%Let $a:(h,h)$ and $b:(h,h)$.
%Then $a;b \updownarrow a \updownarrow b;a$ if and only if $b$ is idle.
%\end{lem}

\begin{lem} \label{lem:statelesscoeteneous}
Let $P:(h,l)$ be any stateless connector. If $P \dtrans{a}{b} P$, then there exist idle $a',b'$ such that 
$P \dtrans{a;a'}{b;b'} P$ with $a;a'$ and $b;b'$ valid and coetaneous, $a;a'\updownarrow a$ and $b;b' \updownarrow b$.
\end{lem}
%
%\begin{proof} See Appendix~\ref{app:proof-ptiles}.
%\end{proof}

\begin{proof}%[{\bf Proof of Lemma~\ref{lem:statelesscoeteneous}}]
By structural induction on $P$, along the lines of Lemma~\ref{lem:allcoeteneous}. 

If $P$ is a basic connector, then $a$ and $b$ are valid and coetaneous by Lemma~\ref{lem:coetaneousbasic}.

If $P=Q\otimes R$ then it must be the case that $Q \dtrans{a_{1}}{b_{1}} Q$ and $R \dtrans{a_{2}}{b_{2}} R$ with $a=a_{1}a_{2}$, $b=b_{1}b_{2}$.
By inductive hypothesis:
\begin{iteMize}{$-$}
\item there must exist idle $a'_{1},b'_{1}$ such that  $Q \dtrans{a_{1};a'_{1}}{b_{1};b'_{1}} Q$ with $a_{1};a'_{1}$ and $b_{1};b'_{1}$ valid and coetaneous, $a_{1};a'_{1}\updownarrow a_{1}$ and $b_{1};b'_{1} \updownarrow b_{1}$;
\item there must exist idle $a'_{2},b'_{2}$ such that  $R \dtrans{a_{2};a'_{2}}{b_{2};b'_{2}} R$ with $a_{2};a'_{2}$ and $b_{2};b'_{2}$ valid and coetaneous, $a_{2};a'_{2}\updownarrow a_{2}$ and $b_{2};b'_{2} \updownarrow b_{2}$.
\end{iteMize}
Let $k_{1}=\ageof{(a_{1};a'_{1})(b_{1};b'_{1})}$ and $k_{2}=\ageof{(a_{2};a'_{2})(b_{2};b'_{2})}$.
If $k_{1} = k_{2}$ then $(a_{1};a'_{1})(a_{2};a'_{2})$ and $(b_{1};b'_{1})(b_{2};b'_{2})$ are (valid and) coetaneous and we are done, because 
$a=a_{1}a_{2} \updownarrow (a_{1};a'_{1})(a_{2};a'_{2})$ and $b=b_{1}b_{2} \updownarrow (b_{1};b'_{1})(b_{2};b'_{2})$
(recall that $\updownarrow$ is a congruence).
Otherwise, assume without loss of generality that $k_{1} > k_{2}$.
Then, by applying \ruleLabel{Epoch} and \ruleLabel{Vert} for $k_{1}-k_{2}$ times to $R$, we have 
$R \dtrans{a'_{3}}{b'_{3}} R$ for the unique idle and coetaneous $a'_{3},b'_{3}$ such that $\ageof{a'_{3}b'_{3}} = k_{1}-k_{2}$.
Then, by \ruleLabel{Vert}, $R \dtrans{a_{2};a'_{2};a'_{3}}{b_{2};b'_{2};b'_{3}} R$. Now, it is obvious that $a_{2};a'_{2};a'_{3}$ and $b_{2};b'_{2};b'_{3}$ are coetaneous and that $\ageof{(a_{2};a'_{2};a'_{3})(b_{2};b'_{2};b'_{3})} = k_{1}$.
Therefore we let $a' = a'_{1}(a'_{2};a'_{3})$ and $b' = b'_{1}(b'_{2};b'_{3})$ and we are done.

If $P=Q; R$ then it must be the case that $Q \dtrans{a}{c} Q$ and $R \dtrans{c}{b} R$.
By inductive hypothesis:
\begin{iteMize}{$-$}
\item there must exist idle $a'_1,c'_1$ such that 
$Q \dtrans{a;a'_1}{c;c'_1} Q$ with $a;a'_1$ and $c;c'_1$ valid and coetaneous, $a;a'_1\updownarrow a$ and $c;c'_1 \updownarrow c$;
\item there must exist idle $c'_2,b'_2$ such that 
$R \dtrans{c;c'_2}{b;b'_2} R$ with $c;c'_2$ and $b;b'_2$ valid and coetaneous, $c;c'_2\updownarrow c$ and $b;b'_2 \updownarrow b$.
\end{iteMize}
Let $k_{1}=\ageof{c;c'_1}$ and $k_{2}=\ageof{c;c'_2}$.
If $k_{1} = k_{2}$ then $c'_1=c'_2$. Then, $a;a'_1$ and $b;b'_2$ are (valid and) coetaneous and we are done.
Otherwise, assume without loss of generality that $k_{1} > k_{2}$.
Then, by applying \ruleLabel{Epoch} and \ruleLabel{Vert} for $k_{1}-k_{2}$ times to $R$, we have 
$R \dtrans{c'_3}{b'_3} R$ for the unique idle and coetaneous $c'_3,b'_3$ such that $\ageof{c'_3 b'_3} = k_{1}-k_{2}$.
Then, by \ruleLabel{Vert}, $R \dtrans{c;c'_2;c'_2}{b;b'_2;b'_3} R$. Now, it is obvious that $c;c'_1 = c;c'_2;c'_3$.
Therefore we let $a' = a'_1$ and $b' = b'_2;b'_3$ and we are done.
\end{proof}

\begin{lem}\label{lem:slicemany}
If $P \dtrans{a}{b} P'$, then there exist
$k$ elementary observations $a_{1},...,a_{k}$ such that  $\bar{a}_{1};\cdots;\bar{a}_{k} \updownarrow a$,
$k$ elementary observations $b_{1},...,b_{k}$ such that  $\bar{b}_{1};\cdots;\bar{b}_{k} \updownarrow b$, and 
$k$ terms $P_{1},P_{2},...,P_{k}$ such that
$P \dtrans{\bar{a}_{1}}{\bar{b}_{1}} P_{1}$,
$P_{1} \dtrans{\bar{a}_{2}}{\bar{b}_{2}} P_{2}$,
...,
$P_{k-1} \dtrans{\bar{a}_{k}}{\bar{b}_{k}} P_{k} = P'$.
\end{lem}
\begin{proof}
By Lemma~\ref{lem:allcoeteneous} we know that there exist idle $a',b'$ such that 
$P \dtrans{a;a'}{b;b'} P'$ with $a;a'$ and $b;b'$ valid and coetaneous.
Then, the thesis follows by Lemma~\ref{lem:slicemanyvalid}.
\end{proof}

We start by proving a strong correspondence between one step reductions in the P/T calculus and Petri tile system.

\begin{lem}\label{lem:strongtilebasic}
Let $P$ be any basic stateless connector. $P \dtrans{\alpha}{\beta} Q$ in the P/T calculus if and only if $P \dtrans{\bar{\alpha}}{\bar{\beta}} Q$ in the Petri tile system.
\end{lem}
\begin{proof}
By case analysis, exploiting Lemmas~\ref{lem:basicageof1}--\ref{lem:idleobs}.
\end{proof}

\begin{lem}\label{lem:strongtileplaces}
If $\tokensplace{n} \dtrans{h}{k} Q$ in the P/T calculus then $\tokensplace{n} \dtrans{\bar{h}}{\bar{k}} Q'$  in the Petri tile system with $Q' \simeq_{\mathsf{tb}} Q$.
Vice versa, if $\tokensplace{n} \dtrans{\bar{h}}{\bar{k}} Q'$  in the Petri tile system, then $\tokensplace{n} \dtrans{h}{k} \tokensplace{n+h-k}$  in the P/T calculus.\end{lem}
\begin{proof}
Immediate, by Lemma~\ref{lem:nminuskplush}.
\end{proof}

\begin{lem}\label{lem:synchP}
If $\synch{P} \dtrans{a}{b} Q$, then $a,b$ are valid and coetaneous and $Q = \synch{P'}$ for some $P'$ such that 
$P \dtrans{a}{b} P'$.
\end{lem}
\begin{proof}
Consider  $P:(h,k)$. By Lemmas~\ref{lem:twtw} and~\ref{lem:idleobs}, the only admissible moves \linebreak
$(\id_{h}\otimes \lzero);\tw_{h,h} \dtrans{a}{c} R$ are when $R=(\id_{h}\otimes \lzero);\tw_{h,h}$ and $c = aa'$ with $a':(1,1)$ the only idle observation that is coetaneous with $a$. Similarly, the only admissible moves $\tw_{k,k};(\id_{k}\otimes \rzero) \dtrans{d}{b} R'$ are when $R' = \tw_{k,k};(\id_{k}\otimes \rzero)$ and $d = bb'$with $b':(1,1)$ the only idle observation that is coetaneous with $b$. Since the tile $\synch{P} \dtrans{a}{b} Q$ must be completed by  finding a suitable tile $(P\otimes \id) \dtrans{aa'}{bb'} Q'$, the only possibility is that $Q' = P'\otimes \id$, $a'=b'$ and $P \dtrans{a}{b} P'$. Hence, $Q=\synch{P'}$ and $a$ and $b$ must be coetaneous  by transitivity.
\end{proof}

\begin{lem}
If for any $a,b$ such that $P \dtrans{a}{b} Q$ then $a$ and $b$ are valid and coetaneous, then $P  \simeq_{\mathsf{tb}} \synch{P}$.
\end{lem}
\begin{proof}
Direct consequence of Lemma~\ref{lem:twtw}.
\end{proof}

\begin{lem}\label{lem:tiletopt1}
If $\synch{P} \dtrans{\bar{\alpha}}{\bar{\beta}} Q$ in the Petri tile calculus, then $\synch{P} \dtrans{\alpha}{\beta} Q'$ in the P/T calculus with $Q' \simeq_{\mathsf{tb}} Q$. 
\end{lem}
\begin{proof}
If $\synch{P} \dtrans{\bar{\alpha}}{\bar{\beta}} Q$,
by Lemma~\ref{lem:synchP} we know that $P \dtrans{\bar{\alpha}}{\bar{\beta}} P'$ in the Petri tile calculus and $Q = \synch{P'}$ for some $P'$.
Then, by Lemma~\ref{lem:strongtile} we know that $P \dtrans{\alpha}{\beta} P''$ in the P/T calculus for some $P''  \simeq_{\mathsf{tb}} P'$. Therefore, $\synch{P}\dtrans{\alpha}{\beta} \synch{P''}$ in the P/T calculus.
We conclude by taking $Q' = \synch{P''}$, since tile bisimilarity is a congruence and therefore $Q' = \synch{P''} \simeq_{\mathsf{tb}} \synch{P'} = Q$. 
\end{proof}

\begin{lem}\label{lem:tiletopt2}
If $\synch{P} \dtrans{\alpha}{\beta} Q$ in the Petri tile calculus,  $\synch{P} \dtrans{\alpha}{\beta} Q'$ in the P/T calculus with $Q' \simeq_{\mathsf{tb}} Q$. 
\end{lem}
\begin{proof}
If $\synch{P} \dtrans{\alpha}{\beta} Q$, by \ruleLabel{Epoch} and \ruleLabel{Vert} we know also that $\synch{P} \dtrans{\bar{\alpha}}{\bar{\beta}} Q$ and conclude by Lemma~\ref{lem:tiletopt1}. 
\end{proof}

\begin{proof}[\bf Proof of Theorem~\ref{theo:tiletopt}]
By Lemma~\ref{lem:synchP}, $a$ and $b$ are valid and coetaneous and $Q = \synch{P'}$ for some $P'$ such that 
$P \dtrans{a}{b} P'$.
By Lemma~\ref{lem:slicemanyvalid},
there exist
$k$ elementary observations $\alpha_{1},...,\alpha_{k}$ with $\bar{\alpha}_{1};\cdots;\bar{\alpha}_{k} = \bar{a}$, 
$k$ elementary observations $\beta_{1},...,\beta_{k}$ with $\bar{\beta}_{1};\cdots;\bar{\beta}_{k} = \bar{b}$, 
and $k$ terms $P_{1},P_{2},...,P_{k}$ such that
$P \dtrans{\bar{\alpha}_{1}}{\bar{\beta}_{1}} P_{1}$, then
$P_{1} \dtrans{\bar{\alpha}_{2}}{\bar{\beta}_{1}} P_{2}$,
...,
$P_{k-1} \dtrans{\bar{\alpha}_{k}}{\bar{\beta}_{k}} P_{k} = P'$.
Then, we proceed by straightforward  induction on the length $k$ of the computation, exploiting Lemma~\ref{lem:tiletopt1}.
\end{proof}

\end{document}

%% file: macros.tex
\usepackage{hyperref}
\usepackage{amsmath}

\usepackage[all]{xy}\CompileMatrices 
\usepackage{subfigure}
\usepackage{stmaryrd}
\usepackage{amssymb}	
\usepackage{rotating}	
\usepackage{epsfig}
\usepackage{amsfonts}	
\usepackage{proof}
\usepackage{amsbsy}
\usepackage{extarrows}

\usepackage{commands}
\usepackage{fancybox}
\usepackage{prooftree}
\usepackage{enumerate}

\theoremstyle{plain}
\theoremstyle{plain}
\theoremstyle{plain}
\theoremstyle{plain}
\theoremstyle{plain}
\theoremstyle{plain}
\theoremstyle{plain}

\newcommand{\tagliare}[1]{}
\makeatletter 
\def\moverlay{\mathpalette\mov@rlay} 
\def\mov@rlay#1#2{\leavevmode\vtop{% 
\baselineskip\z@skip \lineskiplimit-\maxdimen 
\ialign{\hfil$#1##$\hfil\cr#2\crcr}}} 
\makeatother

\newcommand{\bnfEq}{\; ::= \;}
\newcommand{\bnfSep}{\;\; | \;\;}
\newcommand{\emptyplace}{\ensuremath{\bigcirc}}
\newcommand{\tokenplace}{\ensuremath{\moverlay{\bigcirc\cr\bullet}}}
\newcommand{\seqComp}{\mathrel{;}}
\newcommand{\ten}{\otimes}
\newcommand{\union}{\mathrel{\cup}}
\newcommand{\intersection}{\mathrel{\cap}}
\newcommand{\labelSep}{}
\newcommand{\id}{\ensuremath{\mathsf{I}}}

\newcommand{\diag}{\ensuremath{\mathsf{\Delta}}}
\newcommand{\codiag}{
\ensuremath{\mathchoice
{{\rotatebox[origin=c]{180}{$\diag$}}\displaystyle} 
{{\rotatebox[origin=c]{180}{$\diag$}}\textstyle} 
{\!{\rotatebox[origin=c]{180}{$\scriptstyle\diag$}}\!\scriptstyle} 
{{\rotatebox[origin=c]{180}{$\scriptscriptstyle\diag$}}\scriptscriptstyle}
}}

\newcommand{\comp}{\mathrel{;}}
\newcommand{\ldiag}{\ensuremath{\mathsf{\Lambda}}}
\newcommand{\lcodiag}{\ensuremath{\mathsf{V}}}
\newcommand{\leftEnd}{\ensuremath{\pmb{\top}}}
\newcommand{\rightEnd}{\ensuremath{\pmb{\bot}}}
\newcommand{\tw}{\ensuremath{\mathsf{X}}}
\newcommand{\lzero}{\ensuremath{\,\pmb{\uparrow}\,}}
\newcommand{\rzero}{\ensuremath{\,\pmb{\downarrow}\,}}

\newcommand{\derivationRule}[3]{{\prooftree{\scriptstyle #1}\justifies{\scriptstyle #2}\using\ruleLabel{#3}\endprooftree}}
\newcommand{\reductionRule}[2]{{\prooftree{\scriptstyle #1}\justifies{\scriptstyle #2}\endprooftree}}
\newcommand{\ent}{\vdash}
\newcommand{\typ}{\mathrel{:}}
\newcommand{\typeJudgment}[3]{{#1\;\ent\; {#2} \,\typ\, {#3}}}
\newcommand{\sort}[2]{\ensuremath{(#1,\,#2)}}
\newcommand{\rring}[1]{\ensuremath{\mathbb{#1}}}
\newcommand{\N}{\rring{N}}
\newcommand{\fig}[1]{Fig.~{#1}}

\newcommand\twarr[2]{%
\mathrel{\mathop{\moverlay{\scriptstyle\xrightarrow{\,#1\,}\cr{\lower.2em\hbox{$\scriptstyle{}_{#2}$}}}}}}
\newcommand\twarrw[2]{%
\mathrel{\mathop{\moverlay{\scriptstyle\Longrightarrow\cr{\lower-.6em\hbox{$\scriptstyle{}_{#1}$}}
\cr{\lower.3em\hbox{$\scriptstyle{}_{#2}$}}}}}}

\newcommand{\dtrans}[2]{\raise1pt\hbox{$\;\twarr{#1}{#2}\;$}}
\newcommand{\dtransw}[2]{\raise1pt\hbox{$\;\twarrw{#1}{#2}\;$}}
\newcommand{\dtranswl}[2]{\xLongrightarrow[#2]{#1}}

\newcommand{\ruleLabel}[1]{\textsc{\scriptsize{(#1)}}}
\newcommand{\lowerPic}[3]{\lower#1\hbox{\includegraphics[#2]{#3}}}
\newcommand{\pre}[1]{{^\circ{#1}}}
\newcommand{\post}[1]{{#1^\circ}}
\newcommand{\preandpost}[1]{{^\circ{#1^\circ}}}
\newcommand{\Defeq}{\stackrel{\mathrm{def}}{=}}
\renewcommand{\len}[1]{{\#{#1}}}
\newcommand{\sortedLTS}[2]{\ensuremath{\sort{k}{l}\!-\!\mathrm{LTS}}}
\newcommand{\characteristic}[1]{{\ulcorner #1 \urcorner}}
\newcommand{\source}[1]{{^\bullet{#1}}}
\newcommand{\target}[1]{{{#1}^\bullet}}
\newcommand{\multiset}[1]{\mathcal{M}_{#1}}
\newcommand{\semanticsOf}[1]{\llbracket{#1}\rrbracket}
\newcommand{\marking}[2]{#1_{#2}}

\newcommand{\weaktranslation}[1]{\ensuremath{\boldsymbol{T}_{#1}}}

\newcommand{\ageof}[1]{\ensuremath{\delta(#1)}}
\newcommand{\addtoken}{\ensuremath{\mathsf{inc}}}
\newcommand{\counttok}[1]{\ensuremath{\mathit{count}(#1)}}
\newcommand{\synch}[1]{\ensuremath{\mathit{sync}(#1)}}

%% file: main.bbl
\begin{thebibliography}{10}

\bibitem{DBLP:journals/mscs/Arbab04}
F.~Arbab.
\newblock {R}eo: a channel-based coordination model for component composition.
\newblock {\em Math.\ Struct.\ in Comp.\ Science}, 14(3):329--366, 2004.

\bibitem{DBLP:conf/wadt/ArbabBCLM08}
F.~Arbab, R.~Bruni, D.~Clarke, I.~Lanese, and U.~Montanari.
\newblock Tiles for {R}eo.
\newblock In A.~Corradini and U.~Montanari, editors, {\em WADT 2008}, volume
  5486 of {\em Lect.\ Notes in Comput.\ Sci.}, pages 37--55. Springer, 2009.

\bibitem{conf/wadt/ArbabR02}
F.~Arbab and J.~Rutten.
\newblock A coinductive calculus of component connectors.
\newblock In M.~Wirsing, D.~Pattinson, and R.~Hennicker, editors, {\em WADT
  2002}, volume 2755 of {\em Lect.\ Notes in Comput.\ Sci.}, pages 34--55.
  Springer, 2002.

\bibitem{journals/scp/BaierSAR06}
C.~Baier, M.~Sirjani, F.~Arbab, and J.~Rutten.
\newblock Modeling component connectors in {R}eo by constraint automata.
\newblock {\em Sci. Comput. Program.}, 61(2):75--113, 2006.

\bibitem{DBLP:journals/mscs/BaldanCEH05}
P.~Baldan, A.~Corradini, H.~Ehrig, and R.~Heckel.
\newblock Compositional semantics for open {P}etri nets based on deterministic
  processe.
\newblock {\em Math.\ Struct.\ in Comp.\ Science}, 15(1):1--35, 2005.

\bibitem{BarbosaB04}
M.~A. Barbosa and L.~S. Barbosa.
\newblock Specifying software connectors.
\newblock In Z.~Liu and K.~Araki, editors, {\em ICTAC 2004}, volume 3407 of
  {\em Lect.\ Notes in Comput.\ Sci.}, pages 52--67. Springer, 2004.

\bibitem{DBLP:conf/sefm/BasuBS06}
A.~Basu, M.~Bozga, and J.~Sifakis.
\newblock Modeling heterogeneous real-time components in {BIP}.
\newblock In {\em SEFM 2006}, pages 3--12. IEEE Computer Society, 2006.

\bibitem{DBLP:journals/fuin/BernardinelloMP07}
L.~Bernardinello, E.~Monticelli, and L.~Pomello.
\newblock On preserving structural and behavioural properties by composing net
  systems on interfaces.
\newblock {\em Fundam. Inform.}, 80(1-3):31--47, 2007.

\bibitem{DBLP:journals/iandc/BestDK02}
E.~Best, R.~R. Devillers, and M.~Koutny.
\newblock The {B}ox algebra = {P}etri nets + process expressions.
\newblock {\em Inf. Comput.}, 178(1):44--100, 2002.

\bibitem{DBLP:conf/ac/BestK03}
E.~Best and M.~Koutny.
\newblock Process algebra: A {P}etri-net-oriented tutorial.
\newblock In J.~Desel, W.~Reisig, and G.~Rozenberg, editors, {\em Lectures on
  Concurrency and Petri Nets}, volume 3098 of {\em Lect.\ Notes in Comput.\
  Sci.}, pages 180--209. Springer, 2003.

\bibitem{DBLP:journals/tc/BliudzeS08}
S.~Bliudze and J.~Sifakis.
\newblock The algebra of connectors - structuring interaction in {BIP}.
\newblock {\em IEEE Trans. Computers}, 57(10):1315--1330, 2008.

\bibitem{Bru:TL}
R.~Bruni.
\newblock {\em Tile Logic for Synchronized Rewriting of Concurrent Systems}.
\newblock PhD thesis, Computer Science Department, University of Pisa, 1999.

\bibitem{DBLP:journals/tcs/BruniGM02}
R.~Bruni, F.~Gadducci, and U.~Montanari.
\newblock Normal forms for algebras of connection.
\newblock {\em Theoret.\ Comput.\ Sci.}, 286(2):247--292, 2002.

\bibitem{DBLP:journals/tcs/BruniLM06}
R.~Bruni, I.~Lanese, and U.~Montanari.
\newblock A basic algebra of stateless connectors.
\newblock {\em Theoret.\ Comput.\ Sci.}, 366(1-2):98--120, 2006.

\bibitem{DBLP:conf/concur/BruniMM11}
R.~Bruni, H.~C. Melgratti, and U.~Montanari.
\newblock A connector algebra for {P/T} nets interactions.
\newblock In J.-P. Katoen and B.~K{\"o}nig, editors, {\em CONCUR 2011}, volume
  6901 of {\em Lect.\ Notes in Comput.\ Sci.}, pages 312--326. Springer, 2011.

\bibitem{BruniPSI2011}
R.~Bruni, H.~C. Melgratti, and U.~Montanari.
\newblock Connector algebras, petri nets, and bip.
\newblock In E.~M. Clarke, I.~Virbitskaite, and A.~Voronkov, editors, {\em PSI
  2011, Ershov Memorial Conference}, volume 7162 of {\em Lect.\ Notes in
  Comput.\ Sci.}, pages 19--38. Springer, 2012.

\bibitem{DBLP:journals/tcs/BruniM02}
R.~Bruni and U.~Montanari.
\newblock Dynamic connectors for concurrency.
\newblock {\em Theoret.\ Comput.\ Sci.}, 281(1-2):131--176, 2002.

\bibitem{DBLP:conf/fossacs/BuscemiS01}
M.~G. Buscemi and V.~Sassone.
\newblock High-level {P}etri nets as type theories in the {J}oin calculus.
\newblock In F.~Honsell and M.~Miculan, editors, {\em FoSSaCS 2001}, volume
  2030 of {\em Lect.\ Notes in Comput.\ Sci.}, pages 104--120. Springer, 2001.

\bibitem{journals/scp/ClarkeCA07}
D.~Clarke, D.~Costa, and F.~Arbab.
\newblock Connector colouring {I}: Synchronisation and context dependency.
\newblock {\em Sci. Comput. Program.}, 66(3):205--225, 2007.

\bibitem{Corradini-Montanari/92}
A.~Corradini and U.~Montanari.
\newblock An algebraic semantics for structured transition systems and its
  application to logic programs.
\newblock {\em Theoret.\ Comput.\ Sci.}, 103:51--106, 1992.

\bibitem{DBLP:conf/apn/Devillers88}
R.~Devillers.
\newblock The semantics of capacities in {P/T} nets.
\newblock In G.~Rozenberg, editor, {\em European Workshop on Applications and
  Theory in Petri Nets}, volume 424 of {\em Lect.\ Notes in Comput.\ Sci.},
  pages 128--150. Springer, 1988.

\bibitem{DBLP:journals/fuin/DevillersKKP03}
R.~Devillers, H.~Klaudel, M.~Koutny, and F.~Pommereau.
\newblock Asynchronous box calculus.
\newblock {\em Fundam. Inform.}, 54(4):295--344, 2003.

\bibitem{Dickson1913}
L.~E. Dickson.
\newblock Finiteness of the odd perfect and primitive abundant numbers with $n$
  distinct prime factors.
\newblock {\em Amer. Journal Math.}, 35(4):413--422, 1913.

\bibitem{DBLP:conf/apn/Fabre06}
E.~Fabre.
\newblock On the construction of pullbacks for safe petri nets.
\newblock In S.~Donatelli and P.~S. Thiagarajan, editors, {\em ICATPN 2006},
  volume 4024 of {\em Lect.\ Notes in Comput.\ Sci.}, pages 166--180. Springer,
  2006.

\bibitem{DBLP:journals/iandc/FerrariM00}
G.~L. Ferrari and U.~Montanari.
\newblock Tile formats for located and mobile systems.
\newblock {\em Inf. Comput.}, 156(1-2):173--235, 2000.

\bibitem{DBLP:journals/scp/FiadeiroM97}
J.~L. Fiadeiro and T.~Maibaum.
\newblock Categorical semantics of parallel program design.
\newblock {\em Sci. Comput. Program.}, 28(2-3):111--138, 1997.

\bibitem{DBLP:conf/birthday/GadducciM00}
F.~Gadducci and U.~Montanari.
\newblock The tile model.
\newblock In {\em Proof, Language, and Interaction}, pages 133--166. The MIT
  Press, 2000.

\bibitem{DBLP:conf/amast/1997}
M.~Johnson, editor.
\newblock {\em Algebraic Methodology and Software Technology, 6th International
  Conference, AMAST'97, Sydney, Australia, December 13-17, 1997, Proceedings},
  volume 1349 of {\em Lect.\ Notes in Comput.\ Sci.} Springer, 1997.

\bibitem{JA11}
S.-S. T.~Q. Jongmans and F.~Arbab.
\newblock Correlating formal semantic models of {R}eo connectors: Connector
  coloring and constraint automata.
\newblock In A.~Silva, S.~Bliudze, R.~Bruni, and M.~Carbone, editors, {\em ICE
  2011}, volume~59 of {\em Elect.\ Proc.\ in Th.\ Comput.\ Sci.}, pages
  84--103, 2011.

\bibitem{Katis1997b}
P.~Katis, N.~Sabadini, and R.~F.~C. Walters.
\newblock Representing {Place/Transition} nets in {Span(Graph)}.
\newblock In Johnson \cite{DBLP:conf/amast/1997}, pages 322--336.

\bibitem{Katis1997a}
P.~Katis, N.~Sabadini, and R.~F.~C. Walters.
\newblock {Span(Graph)}: A categorial algebra of transition systems.
\newblock In Johnson \cite{DBLP:conf/amast/1997}, pages 307--321.

\bibitem{DBLP:conf/birthday/KleijnK11}
J.~Kleijn and M.~Koutny.
\newblock Causality in structured occurrence nets.
\newblock In C.~B. Jones and J.~L. Lloyd, editors, {\em Dependable and Historic
  Computing}, volume 6875 of {\em Lect.\ Notes in Comput.\ Sci.}, pages
  283--297. Springer, 2011.

\bibitem{Kleijn2012185}
J.~Kleijn and M.~Koutny.
\newblock Localities in systems with a/sync communication.
\newblock {\em Theoret.\ Comput.\ Sci.}, 429(0):185--192, 2012.

\bibitem{Kleijn2012189}
J.~Kleijn, M.~Koutny, and M.~Pietkiewicz-Koutny.
\newblock Regions of petri nets with a/sync connections.
\newblock {\em Theoret.\ Comput.\ Sci.}, 454(0):189--198, 2012.

\bibitem{DBLP:journals/tcs/KoutnyB99}
M.~Koutny and E.~Best.
\newblock Operational and denotational semantics for the {B}ox algebra.
\newblock {\em Theoret.\ Comput.\ Sci.}, 211(1-2):1--83, 1999.

\bibitem{DBLP:conf/concur/KoutnyEB94}
M.~Koutny, J.~Esparza, and E.~Best.
\newblock Operational semantics for the {P}etri {B}ox calculus.
\newblock In B.~Jonsson and J.~Parrow, editors, {\em CONCUR'94}, volume 836 of
  {\em Lect.\ Notes in Comput.\ Sci.}, pages 210--225. Springer, 1994.

\bibitem{LX:CTOSC}
K.~G. Larsen and L.~Xinxin.
\newblock Compositionality through an operational semantics of contexts.
\newblock In M.~Paterson, editor, {\em ICALP'90}, volume 443 of {\em Lect.\
  Notes in Comput.\ Sci.}, pages 526--539. Springer, 1990.

\bibitem{DBLP:journals/mscs/LeiferM06}
J.~J. Leifer and R.~Milner.
\newblock Transition systems, link graphs and {P}etri nets.
\newblock {\em Math.\ Struct.\ in Comput.\ Sci.}, 16(6):989--1047, 2006.

\bibitem{Mes:CRL}
J.~Meseguer.
\newblock Conditional rewriting logic as a unified model of concurrency.
\newblock {\em Theoret.\ Comput.\ Sci.}, 96:73--155, 1992.

\bibitem{RM:SHR}
U.~Montanari and F.~Rossi.
\newblock Graph rewriting, constraint solving and tiles for coordinating
  distributed systems.
\newblock {\em Applied Categorical Structures}, 7(4):333--370, 1999.

\bibitem{DBLP:conf/concur/NielsenPS95}
M.~Nielsen, L.~Priese, and V.~Sassone.
\newblock Characterizing behavioural congruences for {P}etri nets.
\newblock In I.~Lee and S.~A. Smolka, editors, {\em CONCUR'95}, volume 962 of
  {\em Lect.\ Notes in Comput.\ Sci.}, pages 175--189. Springer, 1995.

\bibitem{Wolf92foundationsfor}
D.~E. Perry and E.~L. Wolf.
\newblock Foundations for the study of software architecture.
\newblock {\em ACM SIGSOFT Software Engineering Notes}, 17:40--52, 1992.

\bibitem{DBLP:journals/jlp/Plotkin04a}
G.~D. Plotkin.
\newblock A structural approach to operational semantics.
\newblock {\em J. Log. Algebr. Program.}, 60-61:17--139, 2004.

\bibitem{DBLP:journals/tcs/PrieseW98}
L.~Priese and H.~Wimmel.
\newblock A uniform approach to true-concurrency and interleaving semantics for
  {P}etri nets.
\newblock {\em Theoret.\ Comput.\ Sci.}, 206(1-2):219--256, 1998.

\bibitem{Reisig1985}
W.~Reisig.
\newblock {\em Petri Nets: An Introduction}, volume~4 of {\em Monographs in
  Theoretical Computer Science. An EATCS Series}.
\newblock Springer, 1985.

\bibitem{Sassone2005b}
V.~Sassone and P.~Soboci\'{n}ski.
\newblock A congruence for {P}etri nets.
\newblock {\em Electr.\ Notes in Theor.\ Comput.\ Sci.}, 127(2):107--120, 2005.

\bibitem{Selinger2009}
P.~Selinger.
\newblock A survey of graphical languages for monoidal categories.
\newblock {\em New structures for physics}, pages 289--355, 2011.

\bibitem{DBLP:journals/corr/abs-0912-0555}
P.~Soboci{\'n}ski.
\newblock A non-interleaving process calculus for multi-party synchronisation.
\newblock In F.~Bonchi, D.~Grohmann, P.~Spoletini, and E.~Tuosto, editors, {\em
  ICE 2009}, volume~12 of {\em Elect.\ Proc.\ in Th.\ Comput.\ Sci.}, pages
  87--98, 2009.

\bibitem{DBLP:conf/concur/Sobocinski10}
P.~Soboci{\'n}ski.
\newblock Representations of {P}etri net interactions.
\newblock In P.~Gastin and F.~Laroussinie, editors, {\em CONCUR 2010}, volume
  6269 of {\em Lect.\ Notes in Comput.\ Sci.}, pages 554--568. Springer, 2010.

\bibitem{Sobocinski2012}
P.~Soboci{\'n}ski.
\newblock Relational presheaves as labelled transition systems.
\newblock In D.~Pattinson and L.~Schr{\"o}der, editors, {\em CMCS 2012}, volume
  7399 of {\em Lect.\ Notes in Comput.\ Sci.}, pages 40--50. Springer, 2012.

\bibitem{DBLP:journals/igpl/Stefanescu98}
G.~Stefanescu.
\newblock Reaction and control {I}. {M}ixing additive and multiplicative
  network algebras.
\newblock {\em Logic Journal of the IGPL}, 6(2):348--369, 1998.

\end{thebibliography}
